\documentclass[12pt]{article}
\usepackage{lMac}
\usepackage{becMac}

\def\showintremarks{n}   %%%    y displays internal remarks
\def\showissues{n}       %%%    y displays "issue" internal remarks
\def\SHOWissues{n}       %%%    y prints "issue" internal remarks to the shell

\newcommand{\bbbB}{B}

\begin{document}
\title{The Small Field Parabolic Flow for Bosonic Many--body Models:\\
        \Large Part 2 --- Fluctuation Integral and Renormalization}

\author{Tadeusz Balaban}
\affil{\small Department of Mathematics \authorcr
       % Hill Center--Busch Campus \authorcr
       Rutgers, The State University of New Jersey \authorcr
     %  110 Frelinghuysen Rd \authorcr
      % Piscataway, NJ 08854-8019 \authorcr
       tbalaban@math.rutgers.edu\authorcr
       \  }

\author{Joel Feldman\thanks{Research supported in part by the Natural 
                Sciences and Engineering Research Council 
                of Canada and the Forschungsinstitut f\"ur 
                Mathematik, ETH Z\"urich.}}
\affil{Department of Mathematics \authorcr
       University of British Columbia \authorcr
     %  Vancouver, B.C.,   %\authorcr
      %  Canada\ \   V6T 1Z2 \authorcr
       feldman@math.ubc.ca \authorcr
       http:/\hskip-3pt/www.math.ubc.ca/\squig feldman/\authorcr
       \  }

\author{Horst Kn\"orrer}
\author{Eugene Trubowitz}
\affil{Mathematik \authorcr
       ETH-Z\"urich \authorcr
      %  ETH-Zentrum \authorcr
      % CH-8092 Z\"urich, %\authorcr
      %  Switzerland \authorcr
       knoerrer@math.ethz.ch, trub@math.ethz.ch \authorcr
       http:/\hskip-3pt/www.math.ethz.ch/\squig knoerrer/}

%\date{\editdate}

\maketitle

\begin{abstract}
\noindent
This paper is a contribution to a program to see symmetry breaking in a
weakly interacting many Boson system on a three dimensional lattice at 
low temperature.  It is part of an analysis of the ``small field''  approximation 
to the ``parabolic flow'' which exhibits the formation of a ``Mexican hat''
potential well. Here we complete the analysis of a renormalization group
step, started in \cite{PAR1}, by ``evaluating'' the fluctuation integral and
renormalizing the chemical potential.

\end{abstract}

\newpage
\tableofcontents

\setcounter{section}{4}
\newpage

Part of our program to construct and analyze an interacting many Boson system 
on a three dimensional lattice in the thermodynamic limit is the 
``small field parabolic flow'' which exhibits the formation of a 
potential well in the effective interaction. For an overview of this 
part, see \cite{ParOv}. The starting point of this program is a 
representation of a ``small field approximation'' to the  partition function 
which is written in the form of a functional integral 
$
\int_{\cX_0} e^{\cA_0 } 
$
over a $3+1$ dimensional unit lattice $\cX_0$,
with an action $\cA_0$   of the form described in 
\cite[\S\sectINTstartPoint]{PAR1}.  
This action is
the outcome of the previous step in our program
that had settled the temporal ultraviolet problem in imaginary time 
(see \cite{UV},  \cite[Appendix \appSZrewrite]{PAR1}).
For the ``small field parabolic flow'' we perform a number of  approximate
block spin transformations
$\bbbt_0^{(SF)}, \cdots ,  \bbbt_n^{(SF)}$, each followed by a rescaling.
Our main result \cite[Theorem \thmTHmaintheorem]{PAR1} is a 
representation of 
$\big(\!(\bbbs \bbbt_n^{(SF)}) \circ \cdots \circ (\bbbs \bbbt_0^{(SF)})\!\big) 
\big(\!e^{\cA_0 }\! \big)$ for all integers $n$ smaller than a given number
$\np$ defined in \cite[Definition \defHTbasicnorm.b]{PAR1}. 
The representation clearly shows the development of 
the potential well, see \cite[(\eqnINTdeepwell)]{PAR1}.

The proof of the main theorem consists of several steps, outlined 
in \cite{ParOv}. It is a combination of block spin transformation and 
complex stationary phase techniques. In \cite{PAR1} the algebraic aspects 
of these steps are presented in detail. The estimates needed to show that 
these algebraic steps are meaningful are presented in \cite{BGE} and 
in this paper.
\cite{BGE} deals with the existence of, and estimates on, the 
background fields (introduced in 
\cite[Definition \defHTbackgrounddomaction]{PAR1} ) on
which the represention of the effective action in the main theorem is based.
Here, we use these estimates as input to 
 complete the inductive proof of \cite[Theorem \thmTHmaintheorem]{PAR1}
which rewrites the representation $e^{\cC_n} \cF_n$ for 
$\big(\!(\bbbs \bbbt_n^{(SF)}) \circ \cdots \circ (\bbbs \bbbt_0^{(SF)})\!\big) 
\big(\!e^{\cA_0 }\! \big) (\psi_*,\psi)$ given in 
\cite[Corollary \corSTmainCor]{PAR1} in the form specified in 
\cite[Theorem \thmTHmaintheorem]{PAR1}. The main steps are
\begin{itemize}[topsep=2pt, itemsep=2pt, parsep=0pt]
\item
``evaluation'' of the fluctuation integral and
\item
renormalization of the chemical potential.
\end{itemize} 
They are performed in the  two main sections of this paper.

The symmetry breaking in the many Boson system is 
expected to happen only when the chemical potential is above a critical 
value. The renormalization of the chemical potential performed in 
this paper gives some insight into the leading term of the 
expansion of this critical chemical potential in powers of the 
coupling constant. This is presented in Appendix \ref{appMustar}.
The more technical Appendixes \ref{appLocal} and 
\ref{appSCscaling} deal with the localization operation 
that we use during the  course of renormalization, 
and with the effect of scaling on the norms we use.

We keep the terminology and notation of \cite{PAR1}, which is summarized in 
\cite[Appendix \appDefinitions]{PAR1}.

\newpage
%%%%%%%%%%%%%%%%%%%%%%%%%%%%%%%%%%
\section{One Block Spin Transformation --- The Fluctuation Integral}\label{chapOSFfluct}
%%%%%%%%%%%%%%%%%%%%%%%%%%%%%%%%%%

In this section, we evaluate the fluctuation integral. 
Fix any $0\le n\le \np$. 
We assume that, if $n\ge 1$, the conclusions of 
\cite[Theorem \thmTHmaintheorem\ and Remark \remHTpreciseinduction]{PAR1} 
hold. In the case of $n=0$, we use the data of 
\cite[\S\sectINTstartPoint]{PAR1}.  

We start by
introducing the main norm that will be used in this section. 
In this and the following section we abbreviate the weight factors
of \cite[Definition \defHTbasicnorm]{PAR1} by
\begin{equation}\label{eqnOSFabbrevwt}
\begin{aligned}
\ka&=\ka(n) &
\ka'&=\ka'(n) \\
\bar\ka &=\ka(n+1) =L^\eta\ka &
\bar\ka' &=\ka'(n+1) =L^{\eta'}\ka' &
\bar\ka_\fl & = \ka_\fl(n+1)=4r_n
\end{aligned}
\end{equation}
We also use the notation
\begin{equation*}
\fv_n=\sfrac{\fv_0}{L^n} = 2{\|\cV_n^{(u)}\|}_{2m}
\end{equation*}
\begin{remark}\label{remOSFvnmum}
By \cite[Remark \remHTpreciseinduction\ and Corollary \corPARmunvn]{PAR1},
\begin{align*}
{\|\cV_n\|}_{2m}&\le\fv_n\\
|\mu_n|&\le 2L^{2n}(\mu_0-\mu_*)+\fv_0^{1-\eps}
\le 4\min\big\{\fv_0^{5\eps}\,,\,L^{2n}\fv_0^{\frac{8}{9}+\eps}\big\}
\end{align*}
\end{remark}
By Remark \ref{remOSFvnmum}, 
   \cite[Definition \defHTbasicnorm, 
   Remark \remHTbasicnorm\  
   and Lemma \lemPARcompradan.b]{PAR1}, we have, choosing $\fv_0$ small 
enough depending on $\eps$ and $L$,
\begin{equation}\label{eqnOSFweightineqsA}
\max\big\{\ L^2|\mu_n|\ ,\  
  \|V_n\|_{2m} (\bar\ka+L^9\bar\ka_\fl)(\bar\ka+\bar\ka'+L^9\bar\ka_\fl)\ \big\}
\le \half\rrho_\bg
\end{equation}
with the $\rrho_\bg$ of \cite[Convention \convBGEconstants]{BGE}.
Denote by $\lun\tilde\cG\run$
the norm of the analytic function $\tilde\cG(\tilde\psi_*,\tilde\psi,z_*,z)$ 
with mass $2m$ which associates
the weight $\,\bar\ka\,$ to the fields $\,\psi_*, \,\psi\,$, weight
$\,\bar\ka'\,$ to the fields $\,\psi_{\nu *},\psi_\nu\,$, $\,\nu=0,\cdots,3\,$,
and the weight $\bar\ka_\fl$ to the fields $z_*,z$. Similarly we 
denote by $\lutn\tilde F\rutn$
the norm of the field map $\tilde F(\tilde\psi_*,\tilde\psi,z_*,z)$ 
with the same mass and field weights.  See \cite[Definition \defDEFkrnel]{PAR1}.

In Lemma \ref{lemOSFmainlem}, below, we prove the bounds that
will be needed for evaluation of the fluctuation integral. It uses

\begin{definition}[Scaling Divergence Factor]\label{defOSFsdf}
Set, for each constant $C\ge 1$ and each $\vp=(p_u,p_0,p_\sp)$, 
\begin{align*}
\sdf(\vp;C)&=\Big(\frac{C}{L^{3/2}}\frac{\bar\ka}{\ka}\Big)^{p_u}
         \Big(\frac{C}{L^{7/2}}\frac{\bar\ka'}{\ka'}\Big)^{p_0}
         \Big(\frac{C}{L^{5/2}}\frac{\bar\ka'}{\ka'}\Big)^{p_\sp}
= C^{p_u+p_0+p_\sp} L^{-\De(\vp)}L^{\eta p_u+ \eta'p_0+\eta'p_\sp}
\end{align*}
where
\begin{equation*}
\De(\vp)=\sfrac{3}{2}p_u+\sfrac{7}{2}p_0+\sfrac{5}{2}p_\sp
\end{equation*}
Furthermore set
\begin{equation*}
\sdf(C)=\sup_{\vp\notin\fD_\rel}\sdf(\vp;C)
\end{equation*}
\end{definition}

\begin{remark}\label{remOSFsdf} 
Assuming that $L\ge \big(2C^8\big)^{1/\eps}$, 
we have
$
\sdf(C)\le\sfrac{1}{2L^5} 
$

\end{remark}
\begin{proof}
When $\vp =(p_u,p_0,p_\sp)$ and $|\vp|=p_u+p_0+p_\sp$
\begin{align*}
\log_L \sdf(\vp, C)
%= -(\sfrac{3}{2}-\eta) p_u - (\sfrac{7}{2}-\eta') p_0-(\sfrac{5}{2}-\eta') p_\sp
%    +(p_u+p_0+p_\sp)\log_L C
&= \log_L \sdf(\vp, 1)
   +|\vp|\log_L C
\end{align*}
and
\begin{align*}
\log_L \sdf(\vp, 1)
&= -(\sfrac{3}{2}-\eta)( p_u+p_0+p_\sp) 
   -\big(1+(\eta-\eta')\big)(p_0+ p_\sp)
   - p_0 \\
&\le  -\begin{cases}
  8(\sfrac{3}{2} - \eta) & \text{if $|\vp| \ge 8$}\\
%  11 -5\eta -\eta' & \text{if $|\vp| =6 $ and $p_0 \ge 1$}\\
  10 -5\eta -\eta' & \text{if $|\vp| =6 $ and $p_0+p_\sp \ge 1$}\\
  8-3\eta -\eta' & \text{if $|\vp| =4 $ and $p_0 \ge 1$}\\
  8 -2\eta -2\eta' & \text{if $|\vp| =4 $ and $p_\sp \ge 2$}\\
  7 -2\eta' & \text{if $p_u=p_\sp=0, \ p_0=2 $}
    \end{cases}\\
&\le -5-\max\big\{\eps\,,\,\sfrac{1}{2}(|\vp|-8)\big\}
\end{align*}
Consequently, 
\begin{align*}
\log_L\big(2L^5\sdf(\vp,C)\big)
&\le \log_L(2C^8) + (|\vp|-8)\log_L C
      -\max\big\{\eps\,,\,\sfrac{1}{2}(|\vp|-8)\big\}\\
&\le \eps - 
   \begin{cases}\eps & \text{if $|\vp|\le 8$}\\
          \sfrac{1}{4}(|\vp|-8) & \text{if $|\vp|\ge 10$}
   \end{cases} \\
&\le 0
\end{align*}
\end{proof}

\begin{remark}\label{remOSFdrelmotivation}
This remark provides the motivation for our choice $\fD_\rel$ 
(in \cite[Definition \defINTrelmonomial]{PAR1}) and $\fD$ (in \cite[(\eqnINTfDdef)]{PAR1}).

Let $\cM$ be a monomial of type $\vp$, as in \cite[Definition \defINTmonomialtype]{PAR1}. By \cite[Lemma \lemSAscaletoscale.b]{PAR1},
\begin{equation*}
\lun\bbbs\cM\run\le L^5\,\sdf(\vp;1)\,\|\cM\|^{(n)}
\end{equation*}
(If the mass were zero in both norms, this would be an equality.)
So the ``scale (n+1) norm'' of the scaled monomial $\bbbs\cM$
is smaller than the ``scale n norm'' of the monomial $\cM$ when 
$L^5\,\sdf(\vp;1)<1$. This is the case if and only if $\vp\in\fD_\rel$.
In fact it was exactly this that determined our choice of $\fD_\rel$.
Monomials of type $\vp$ with $\vp\notin\fD_\rel$ are said to be 
``scaling--weight irrelevant''. When such terms are generated during the
course of renormalization group step number $n$, they are placed in the 
``high degree'' part, $\cE_n$ of the action.

Now let $\cM$ be a monomial of type $\vp\in\fD_\rel$. For some 
$\vp$'s, the size of the kernel of $\cM$ decreases, or at least does
not increase, under scaling.
(This does not contradict  $\lun\bbbs\cM\run > \|\cM\|^{(n)}$ because
the field weights in $\lun\bbbs\cM\run$ are greater than the field
weights in $\|\cM\|^{(n)}$.) Indeed, by \cite[Lemma \lemSAscaletoscale.a]{PAR1},
\begin{equation*}
\|\bbbs\cM\|_{2m}\le L^5\,L^{-\frac{3}{2}p_u-\frac{7}{2}p_0-\frac{5}{2}p_\sp}\,\|\cM\|_m
\end{equation*}
(Again, if the mass were zero in both norms, this would be an equality.)
The only $\vp$'s with 
$\sfrac{3}{2}p_u+\sfrac{7}{2}p_0+\sfrac{5}{2}p_\sp<5$, i.e. the only scaling
relevent monomials, are those with $\vp=(2,0,0), (1,0,1)$.
Here is what we do with monomials $\cM$ of type $\vp\in\fD_\rel$
that are generated during the course of renormalization group step number $n$.
See \S\ref{chapRen}.
\begin{itemize}[leftmargin=*, topsep=2pt, itemsep=2pt, parsep=0pt]
\item
 If $\vp=(6,0,0),\ (1,1,0),\ (0,1,1),\ (0,0,2)$,
(i.e. if $\vp\in\fD$) the monomial is placed in the ``low degree'' 
part, $\cR_n$ of the action. 
\item
If $\vp=(4,0,0)$, the monomial is placed in the 
``main'' part, $A_n$ of the action, renormalizing $\cV$.
\item
 If $\vp=(3,0,1)$, Lemma \ref{lemLlocalize}.b is used
to express $\cM$ as a sum of monomials of type $\vp'$ with 
$\vp'=(2,1,1),\ (2,0,2)$, which are placed in the 
``high degree'' part, $\cE_n$ of the action.
\item
If $\vp=(2,0,0)$, Lemma \ref{lemLlocalize}.c is used
to express $\cM$ as a sum of a local degree two monomial, which 
is placed in the ``main'' part, $A_n$ of the action, renormalizing 
the chemical potential $\mu$, and a sum of monomials of type $\vp'$ with 
$\vp'=(1,1,0), (0,1,1),\ (0,0,2)$, which are placed in the 
``low degree'' part, $\cR_n$ of the action.
\item 
If $\vp=(1,0,1)$, Lemma \ref{lemLlocalize}.a is used
to express $\cM$ as a sum of monomials of type $\vp'$ with 
$\vp'=(0,1,1),\ (0,0,2)$, which are placed in the 
``low degree'' part, $\cR_n$ of the action. 
\end{itemize}
There is one other complication which we have supressed from these
bullets. Monomials generated by the fluctuation integral are naturally
functions of the fields $\psi_{(*)}$. But the ``low degree'' part 
$\cR_n$  and the chemical potential and interaction parts of the
``main'' part $A_n$ of the action are functions of the background field
$\phi_{(*)n}(\psi_*,\psi,\mu_n,\cV_n)$. Expressing the various functions
above in terms of the ``right'' fields complicates the above procedure,
but does not introduce any serious obstructions.
\end{remark}

\begin{lemma}\label{lemOSFmainlem}
There is a constant $\CC_\fl$
that depends only on $\Gam_\op$,  $\GGa_\bg$
and $\rrho_\bg$ such that the following holds.
\begin{enumerate}[label=(\alph*), leftmargin=*]
\item 
Let $\de A_n^{(2)}$ and $\de A_n^{(\ge 3)}$ be the parts
of $\de A_n$ that are of degree two and of degree at least three,
respectively, in $z_{(*)}$. Then
\begin{align*}
\lun \de A_n^{(2)}\run
&\le L^{42} \CC_\fl
 \,\{\|V_n\|_{2m}(\bar\ka+\bar\ka_\fl)^2+|\mu_n|\} \bar\ka_\fl^2   \\
\lun \de A_n^{(\ge 3)}\run
&\le L^{42} \CC_\fl\,\|V_n\|_{2m}(\bar\ka+\bar\ka_\fl)
              \bar\ka_\fl^3
\end{align*}

\item 
Let $\tilde\cE_n$ refer to the $\tilde\cE_n$ 
of \cite[Theorem \thmTHmaintheorem]{PAR1} for $n\ge 1$ and the $\cE_0$ of
\cite[\S\sectINTstartPoint]{PAR1} for $n=0$. 
There are analytic functions 
$\tilde\cE_{n+1,1}(\tilde\psi_*,\tilde\psi)$ and
$\de\tilde\cE_n(\tilde\psi_*,\tilde\psi,z_*,z)$ such that
\begin{align*}
\cE_{n+1,1}(\psi_*,\psi)&=
\tilde\cE_{n+1,1}\big((\psi_*,\{\partial_\nu\psi_*\})\,,\,
                  (\psi,\{\partial_\nu\psi\})\big)\\
\de\cE_n(\psi_*,\psi,z_*,z)&=
\de\tilde\cE_n\big((\psi_*,\{\partial_\nu\psi_*\})\,,\,
                    (\psi,\{\partial_\nu\psi\})\,,\, z_*\,,\,z\big)
\end{align*}
and
\begin{align*}
\lun \tilde\cE_{n+1,1}\run
  & \le  L^5\,\sdf(\CC_\fl)\, \big\| \tilde\cE_n\big\|^{(n)}   \\
\lun \de\tilde\cE_n\run
&  \le L^{18}\ \sfrac{\bar\ka_\fl}{\bar\ka'}
\ \sdf(\CC_\fl)\ \big\| \tilde\cE_n\big\|^{(n)}
\end{align*}
Furthermore, $\tilde\cE_{n+1,1}$ contains no scaling/weight relevant monomials.

\item
We have
\begin{equation*}
(\bbbs\cR_n)(\Phi_*,\Phi)
=\sum_{\vp\in\fD}
    (\bbbs\tilde\cR_n^{(\vp)})\big((\Phi_*,\{\partial_\nu\Phi_*\})\,,\,
                    (\Phi,\{\partial_\nu\Phi\})\big)
\end{equation*}
and
\begin{equation*}
\big\|\bbbs\tilde\cR_n^{(\vp)}\big\|_{2m}
   \le L^{5-\De(\vp)} \big\|\tilde\cR_n^{(\vp)}\big\|_m
\end{equation*}
For each $\vp=(p_u,p_0,p_\sp)\in\fD$, there is 
an analytic function 
$\de\tilde\cR_n^{(\vp)}(\tilde\psi_*,\tilde\psi,z_*,z)$ such that
\begin{equation*}
\de\cR_n(\psi_*,\psi,z_*,z)=\sum_{\vp\in\fD}
\de \tilde\cR_n^{(\vp)}\big((\psi_*,\{\partial_\nu\psi_*\})\,,\,
                   (\psi,\{\partial_\nu\psi\})\,,\, z_*\,,\,z\big)
\end{equation*}
and
\begin{equation*}
\lun \de \tilde\cR_n^{(\vp)}\run\le 
          \CC_\fl^{\De(\vp)}\ 
          L^{5-\De(\vp)}\ 
          \sfrac{\bar\ka^\vp}{\si_n(\vp)}\ 
            {\|\tilde\cR_n^{(\vp)}\|}_m
\end{equation*}
where
\begin{equation*}
\bar\ka^\vp
=\bar\ka^{p_u}\bar\ka'^{p_0+p_\sp}
\end{equation*}
and
\begin{equation*}
\si_n(\vp)=\frac{1}{L^{11}}
               \begin{cases} \sfrac{1}{L^2}\sfrac{\bar\ka'}{\bar\ka_\fl}
                          & \text{if $p_0+p_\sp\ne 0$}\\ \noalign{\vskip0.05in}
                     \sfrac{\bar\ka}{\bar\ka_\fl} 
                         & \text{if $p_0=p_\sp=0$}
               \end{cases}
\end{equation*}
\medskip
\begin{equation}\label{eqnOSFtable}
\renewcommand{\arraystretch}{1.3}
  \begin{tabular}{!{\vrule width 1pt}c|c|c|c|c!{\vrule width 1pt}}
    \noalign{\hrule height 1pt}
       $\vp$
      &$\bar\ka^\vp$
      &$\sfrac{1}{\si_n(\vp)}$
      &$L^{5-\De(\vp)}$        
      &$L^{5-\De(\vp)} \sfrac{\bar\ka^\vp}{\si_n(\vp)}$ \\ \noalign{\hrule height 1pt}
       $(1,1,0)$
      &$\bar\ka\bar\ka'$
      &${\sst L^{13}}\sfrac{\bar\ka_\fl}{\bar\ka'}$
      &$L^0$
      &$(L^2\bar\ka) (L^{11}\bar\ka_\fl)$ \\ \hline
       $(0,1,1)$
      &$\bar\ka'^2$
      &${\sst L^{13}}\sfrac{\bar\ka_\fl}{\bar\ka'}$
      &$L^{-1}$
      &$(L\bar\ka') (L^{11}\bar\ka_\fl)$ \\ \hline
       $(0,0,2)$
      &$\bar\ka'^2$
      &${\sst L^{13}}\sfrac{\bar\ka_\fl}{\bar\ka'}$
      &$L^0$
      &$(L^2\bar\ka') (L^{11}\bar\ka_\fl)$ \\ \hline
       $(6,0,0)$
      &$\bar\ka^6$
      &${\sst L^{11}}\sfrac{\bar\ka_\fl}{\bar\ka}$
      &$L^{-4}$
      &$L^{-4}\bar\ka^5 (L^{11}\bar\ka_\fl)$ \\ \noalign{\hrule height 1pt}
  \end{tabular}
\renewcommand{\arraystretch}{1.0}
\end{equation}
\end{enumerate}
\end{lemma}
\begin{proof} 
(a) We first consider the case $n\ge 1$. 
We apply \cite[Proposition \propBGEdephisoln.b]{BGE} with
$\fm  = 2m$, $\wf=\bar\ka$, $\wf'=\bar\ka'$ and $\wf_\fl=\bar\ka_\fl$. 
The hypotheses of \cite[Proposition \propBGEdephisoln]{BGE}
are fulfilled by \eqref{eqnOSFweightineqsA} so that
\begin{align*}
\luTN \de\hat\phi_{(*)n+1}^{(+)}\ruTN
  &\le  L^{29} \GGa_\bg\,\{\|V_n\|_{2m}(\bar\ka+\bar\ka_\fl)^2+|\mu_n|\}
              \bar\ka_\fl
\end{align*}
The claim follows easily using \cite[(\eqnOSAdeAndef.a)]{PAR1}.

We now consider $n=0$. The kernel $\half V_0^{(s)}=\half V_1^{(u)}$ 
of $\bbbs\cV_0$ is given in \cite[Remark \remSCscaling.h]{PAR1} and fulfills
$\|V_0^{(s)}\|_{2m} \le\sfrac{1}{L} \|V_0\|_{2m}$
by \cite[Lemma \lemSAscaletoscale.a]{PAR1}. Expanding the quartic
\begin{equation*}
(\bbbs\cV_0)\big(\hat\psi_*+\de\psi_*\,,\,
               \hat\psi+\de\psi\big) 
         \bigg|_{\hat\psi_{(*)}=\hat\psi_{0(*)}(\psi_*,\psi,\mu_0)\atop
                  \de\psi_{(*)}=L^{3/2}\bbbs D^{(0)(*)}\bbbs^{-1}z_{(*)}}
\end{equation*}
in powers of $z_{(*)}$, we get, by 
\cite[Remark \remCFpsisolnZero and Proposition \propBGEphivepssoln.a]{BGE} and 
\cite[Proposition \propSUBsubstitution.a]{SUB}
\begin{align*}
\lun \de A_0^{(2)}\run
&\!\le \!\smchoose{4}{2} \|V_1\|_{2m}
    \big[\|S_1(L^2\mu_0)^{(*)}Q_1^* \fQ_1\|_{2m} 
       \!+\!\GGa_\bg \|V_1\|_{2m}\bar\ka^2\big]^2 \bar\ka^2
     \big[L^{\frac{3}{2}}\|\bbbs D^{(0)}\bbbs^{-1}\|_{2m} \bar\ka_\fl\big]^2\\
&\hskip1in+|\mu_0|\,L^5\,\|\bbbs C^{(0)}\bbbs^{-1}\|_{2m} \bar\ka_\fl^2
\\
&\le L^{21} \CC_\fl \big(\|V_0\|_{2m}\bar\ka^2+|\mu_0|\big)\bar\ka_\fl^2   \\
\lun \de A_0^{(\ge 3)}\run
&\le \smchoose{4}{3} \|V_1\|_{2m}
       \big[\|S_1(L^2\mu_0)^{(*)}Q_1^* \fQ_1\|_{2m} 
       \!+\!\GGa_\bg \|V_1\|_{2m}\bar\ka^2\big]\,\bar\ka\ 
     \big[L^{\frac{3}{2}}\|\bbbs D^{(0)}\bbbs^{-1}\|_{2m} \bar\ka_\fl\big]^3\\
&\hskip1in+\smchoose{4}{4} \|V_1\|_{2m} 
     \big[L^{\frac{3}{2}}\|\bbbs D^{(0)}\bbbs^{-1}\|_{2m} \bar\ka_\fl\big]^4\\
&\le L^{42} \CC_\fl\,\|V_0\|_{2m}(\bar\ka+\bar\ka_\fl)
              \bar\ka_\fl^3
\end{align*}

\Item (b) Set
\begin{align*}
\tilde\cE_{n+1,1}(\tilde\psi_*,\tilde\psi)
&=(\bbbs\tilde\cE_n)\big(\,(\Psi_*,\{\Psi_{*\nu}\})\,,\,
                 (\Psi,\{\Psi_\nu\})\,\big)\Big|
       _{\atop{\Psi_{(*)}=\hat\psi_{(*)n}(\psi_*,\psi,\mu_n,\cV_n)}
         {\Psi_{(*)\nu}= 
            \hat \psi_{(*)n,\nu}(\psi_*,\psi,\psi_{*\nu},\psi_\nu,\mu_n,\cV_n)}}
\end{align*}
and
\begin{align*}
\de\tilde\cE_n(\tilde\psi_*,\tilde\psi,z_*,z)
&=(\bbbs \tilde\cE_n)\big(\tilde\Psi_*,\tilde\Psi\big)\Big|
       _{\atop{\Psi_{(*)}=\hat\psi_{(*)n}(\psi_*,\psi,\mu_n,\cV_n)
                        +L^{3/2}\bbbs D^{(n)(*)}\bbbs^{-1}z_{(*)}}
         {\Psi_{(*)\nu}= 
            \hat \psi_{(*)n,\nu}(\psi_*,\psi,\psi_{*\nu},\psi_\nu,\mu_n,\cV_n)
              +L^{3/2}\bbbs_\nu\partial_\nu D^{(n)(*)}\bbbs^{-1}z_{(*)}}}\\
&\hskip0.5in  -(\bbbs\tilde\cE_n)\big(\tilde\Psi_*,\tilde\Psi\big)\Big|
       _{\atop{\Psi_{(*)}=\hat\psi_{(*)n}(\psi_*,\psi,\mu_n,\cV_n)}
          {\Psi_{(*)\nu}= 
            \hat \psi_{(*)n,\nu}(\psi_*,\psi,\psi_{*\nu},\psi_\nu,\mu_n,\cV_n)}}
\end{align*}
with the $\tilde\cE_n$ of \cite[Theorem \thmTHmaintheorem]{PAR1} 
and the $\hat \psi_{(*)n,\nu}$ of \cite[Proposition \propCFpsisoln]{BGE}.
By \cite[Remark \remSCscaling.b]{PAR1} the two equations of part (b) hold.
That $\tilde\cE_{n+1,1}$ contains no scaling/weight relevant monomials
follows from the degree properties of $\hat\psi_{(*)n}$ and 
$\hat \psi_{(*)n,\nu}$ specified in \cite[Proposition \propCFpsisoln]{BGE}.

We set
\begin{equation}\label{eqnOSFlalapsi}
\begin{split}
\la&=\big\{ \GGa_\bg
           +\sfrac{1}{L^9}\|\bbbs D^{(n)}\bbbs^{-1}\|_{2m}\big\}\ \bar\ka
\\
\la'&= \max_{0\le \nu\le 3}\big\{\GGa_\bg 
  +\sfrac{1}{L^{11}}\|\partial_\nu\bbbs D^{(n)}\bbbs^{-1}\|_{2m}\big\}\bar\ka'
\\
\si&=\frac{1}{L^{13}}\frac{\bar\ka'}{\bar\ka_\fl} 
\end{split}
\end{equation}
As $\fv_0$ is being chosen sufficiently small, depending on $L$, 
\begin{equation}\label{eqnOSFsigeone}
\si=\frac{1}{L^{11}}\min\bigg\{
                        \frac{\bar\ka}{\bar\ka_\fl}\ ,\ 
                         \frac{\bar\ka'}{L^2\bar\ka_\fl}\bigg\} \ge 1
\end{equation}
by \cite[Definition \defHTbasicnorm]{PAR1}.
Denote by $\|\ \cdot\ \|_\la$ the (auxiliary) norm with 
mass $2m$ that assigns the weight factors
$\la$ to the fields $\Psi_{(*)}$ and $\la'$ to the fields $\Psi_{\nu(*)}$.
By \cite[Proposition \propCFpsisoln]{BGE}, with $\wf=\bar\ka$, 
$\wf'=\bar\ka'$ and $\wf_\fl=\bar\ka_\fl$, 
\begin{equation}\label{eqnOSFlabnds}
\begin{split}
\lutn \hat\psi_{(*)n}\rutn 
+\si\,\lutn L^{3/2}\bbbs D^{(n)(*)}\bbbs^{-1}z_{(*)}\rutn &\le\la
\\
\lutn \hat \psi_{(*)n,\nu} \rutn 
+\si\,\lutn L^{3/2}\bbbs_\nu\partial_\nu D^{(n)(*)}\bbbs^{-1}z_{(*)}\rutn 
&\le\la',\ 0\le\nu\le 3
\end{split}
\end{equation}
so that, by \cite[Proposition \propSUBsubstitution.a,b]{SUB},
$$
\lun \tilde\cE_{n+1,1}\run\le  \|\bbbs\tilde\cE_n\|_\la    \qquad
\lun \de\tilde\cE_n\run\le \sfrac{1}{\si}\|\bbbs\tilde\cE_n\|_\la
$$
For each monomial $\cM$ of type $\vp$ in $\tilde\cE_n$,
\cite[Lemma \lemSAscaletoscale.b]{PAR1} with
$\fm=2m$, $\check\fm = m$,
$\wf=\la$, $\wf'=\la'$,
$\check\wf=\ka$ and $\check\wf'=\ka'$,
gives
\begin{align*}
\big\|\bbbs\cM\big\|_\la\le L^5\,\Sdf(\cM)\,\|\cM\big\|^{(n)}
\end{align*}
with
\begin{align*}
\Sdf(\cM)
&=\big(\frac{1}{L^{3/2}}\frac{\la}{\ka}\Big)^{p_u}
         \Big(\frac{1}{L^{7/2}}\frac{\la'}{\ka'}\Big)^{p_0}
         \Big(\frac{1}{L^{5/2}}\frac{\la'}{\ka'}\Big)^{p_\sp} 
\le \sdf(\vp;\CC_\fl)
\end{align*}
provided
\begin{align*}
\GGa_\bg   +\sfrac{1}{L^9}\|\bbbs D^{(n)}\bbbs^{-1}\|_{2m} +
          \max_{0\le\nu\le 3}
         \sfrac{1}{L^{11}}\|\partial_\nu\bbbs D^{(n)}\bbbs^{-1}\|_{2m}
\le \GGa_\bg+3e^{2m}\Gam_\op
\le \CC_\fl
\end{align*}
So we have
\begin{equation}\label{eqnOSFstildecE}
\big\|\bbbs\tilde\cE_n\big\|_\la \le  L^5\,\sdf(\CC_\fl)\,
                                             \big\| \tilde\cE_n\big\|^{(n)} 
\end{equation}
and the conclusion follows.

\Item (c)
The first equation holds by \cite[Remark \remSCscaling.b]{PAR1} and the bound
on $\big\|\bbbs\tilde\cR_n^{(\vp)}\big\|_{2m}$ is an immediate
consequence of \cite[Lemma \lemSAscaletoscale.a]{PAR1}.

Set 
\begin{align*}
\de \tilde\cR_n^{(\vp)}(\tilde\psi_*,\tilde\psi,z_*,z)
&\!=\!(\bbbs \tilde\cR_n^{(\vp)})\big(\tilde\Phi_*,\tilde\Phi\big)\Big|
       _{\atop{\Phi_{(*)}=\phi_{(*)n+1}(\psi_*,\psi,L^2\mu_n,\bbbs\cV_n)
                     +\de\hat\phi_{(*)n+1}(\psi_*,\psi,z_*,z)}
              {\atop{\Phi_{(*)\nu}= 
                 \phi_{n+1(*),\nu}(\psi_*,\psi,\psi_{*\nu},\psi_\nu,L^2\mu_n,\bbbs\cV_n)
                 \hskip0.5in}
                    {\hskip1in
         +\de\hat\phi_{(*)n+1,\nu}(\psi_*,\psi,\psi_{*\nu},\psi_\nu,z_*,z)}}}\\
&\hskip0.5in  
  -(\bbbs\tilde\cR_n^{(\vp)})\big(\tilde\Phi_*,\tilde\Phi\big)\Big|
       _{\atop{\Phi_{(*)}=\phi_{(*)n+1}(\psi_*,\psi,L^2\mu_n,\bbbs\cV_n)}
              {\Phi_{(*)\nu}= 
     \phi_{n+1(*),\nu}(\psi_*,\psi,\psi_{*\nu},\psi_\nu,L^2\mu_n,\bbbs\cV_n)}}
\end{align*}
where, by \cite[Proposition \propBGEphivepssoln]{BGE}, 
\begin{align*}
\phi_{n+1(*),\nu}(\psi_*,\psi,\psi_{*\nu},\psi_\nu,L^2\mu_n,\bbbs\cV_n)
&= B^{(\pm)}_{n+1,L^2\mu_n,\nu}\psi_{(*)\nu} \\ &\hskip0.5in
           + \phi_{(*)n+1,\nu}^{(\ge 3)} 
                    (\psi_*,\psi,\psi_{*\nu},\psi_\nu,L^2\mu_n,\bbbs\cV_n)
\end{align*}
The properties of  $\de\hat\phi_{(*)n+1,\nu}$ are given in
\cite[Proposition \propBGEdephisoln.e]{BGE}.
By \cite[Remark \remSCscaling.b]{PAR1} we have 
$\de\cR_n=\sum_{\vp\in\fD}\de \tilde\cR_n^{(\vp)}$.

To bound $\lun \de \tilde\cR_n^{(\vp)}\run$ we proceed as we did in part
(b), but setting
\begin{align*}
\la&=\big\{ \big\|S_{n+1}(L^2\mu_n)^{(*)}Q_{n+1}^* \fQ_{n+1}\big\|_{2m} 
          +\GGa_\bg \La_\phi + \GGa_\bg\big\}\ \bar\ka
\\
\la'&= \max_{0\le \nu\le 3}\big\{ 
 \max_{\si\in\pm}\big\|B_{n+1,L^2\mu_n,\nu}^{(\si)}\big\|_{2m} 
          +\GGa_\bg \La_\phi+ \GGa_\bg\big\} \bar\ka'
%\de\la&= L^{3/2}\GGa_\bg\, \bar\ka_\fl
%\\
%\de\la_\nu&= L_\nu L^{3/2}\GGa_\bg\, \bar\ka_\fl
%\\
\end{align*}
Denote by $\|\ \cdot\ \|_\la$ the (auxiliary) norm with 
mass $2m$ that assigns the weight factors
$\la$ to the fields $\Phi_{(*)}$ and $\la'$ to the fields $\Phi_{\nu(*)}$.
By \cite[Propositions \propBGEphivepssoln\ and \propBGEdephisoln.a,e]{BGE},
with $\wf=\bar\ka$, $\wf'=\bar\ka'$ and $\wf_\fl=\bar\ka_\fl$, 
\begin{alignat*}{3}
\lutn \phi_{(*)n+1}(\psi_*,\psi,L^2\mu_n,\bbbs\cV_n) \rutn 
+\si_n(\vp)\,\lutn \de\hat\phi_{(*)n+1}\rutn &\le\la  &
\quad&\text{if $p_u\ne 0$}
\\
\lutn B^{(\pm)}_{n+1,L^2\mu_n,\nu}\psi_{(*)\nu}+\phi^{(\ge 3)}_{(*) n+1,\nu} \rutn 
+\si_n(\vp)\,\lutn \de\hat\phi_{(*)n+1,\nu}\rutn 
&\le\la'&
&\text{if $p_0+p_\sp\ne 0$}
\end{alignat*}
As in \eqref{eqnOSFsigeone}, $\si_n(\vp)\ge 1$,
so that, by \cite[Proposition \propSUBsubstitution.b]{SUB},
\begin{align*}
\lun \de \tilde\cR_n^{(\vp)}\run
       &\le \sfrac{1}{\si_n(\vp)}\|\bbbs\tilde\cR_n^{(\vp)}\|_\la
% \\      &=\sfrac{1}{\si_n(\vp)}\|\bbbs\tilde\cR_n^{(\vp)}\|_{2m}\ 
%               \la^{p_{*u}+p_u}\prod_{\nu=0}^3{\la'}^{p_{*\nu}+p_\nu}\\ &
       \le \sfrac{1}{\si_n(\vp)}L^{5-\De(\vp)}
            \|\tilde\cR_n^{(\vp)}\|_m\ 
               \la^{p_{*u}+p_u}\prod_{\nu=0}^3{\la'}^{p_{*\nu}+p_\nu}
\end{align*}
\end{proof}

Parts (b) and (c) of Lemma \ref{lemOSFmainlem} 
provide bounds on the constituents $\cE_{n+1,1}$ and $\bbbs\cR_n$
of the contribution, $\cC_n$, from the critical field in \cite[Corollary \corSTmainCor]{PAR1}. We now provide a bound on the fluctuation
integral $\cF_n$.

\begin{proposition}\label{propOSFmainprop}
There is an analytic function  $\tilde\cE_\fl(\tilde\psi_*,\tilde\psi)$ 
   and a constant $\cZ'_n$ such that 
\begin{equation*}
\cF_n(\psi_*,\psi)
=\sfrac{1}{\cZ'_n}e^{\tilde\cE_\fl((\psi_*,\{\partial_\nu\psi_*\})\,,\,
                  (\psi,\{\partial_\nu\psi\}))}
\qquad
\text{and}
\qquad
\lun\tilde\cE_\fl\run  \le   \fe_\fl(n)
\end{equation*}

\end{proposition}

\begin{proof} By \cite[\eqnOSAfluctInt]{PAR1} and  Lemma \ref{lemOSFmainlem}
the fluctuation integral  is
\begin{align*}
\cF_n(\psi_*,\psi)
&=\int d\mu_{r_n}(z^*,z)\ 
e^{\tilde\cD(\tilde\psi_*,\tilde\psi,z_*,z)}
        \bigg|_{\tilde\psi_{(*)}=(\psi_{(*)},\{\partial_\nu\psi_{(*)}\})}
\end{align*}
where
\begin{align*}
\int d\mu_{r_n}(z^*,z)
=\Big[\hskip-3pt
 \prod_{w\in\cX_1^{(n)}}  \int\limits_{|z(w)|\le r_n} \hskip -6pt
\sfrac{dz(w)^*\wedge dz(w)}{2\pi i} e^{-|z(w)|^2}\Big]
\end{align*}
and
\refstepcounter{equation}\label{eqnOSFineqs}
\begin{align*}
\tilde\cD(\tilde\psi_*,\tilde\psi,z_*,z)
&=-\de A^{(2)}_n(\psi_*,\psi,z_*,z) -\de A^{(\ge 3)}_n(\psi_*,\psi,z_*,z) \\
&\hskip0.5in  +\de\tilde\cE_n\big(\tilde\psi_*,\tilde\psi,z_*,z\big)    
  +\sum_{\vp\in\fD}
\de \tilde\cR_n^{(\vp)}\big(\tilde\psi_*,\tilde\psi, z_*,z\big)
\end{align*}
By Remark \ref{remOSFvnmum} and \cite[Definition \defHTbasicnorm,
Lemma \lemPARlnmuzeromustar.d and (\eqnPARestrad.b)]{PAR1},
\begin{align}
2 L^{42} \CC_\fl \,\{\|V_n\|_{2m}(\bar\ka+\bar\ka_\fl)^2+|\mu_n|\} \bar\ka_\fl^2
  &\le 2 L^{42} \CC_\fl \,\{4\sfrac{\fv_0}{L^n}\bar\ka^2
              +2L^{2n}(\mu_0-\mu_*)+\fv_0^{1-\eps}\} \bar\ka_\fl^2
\nonumber\\
      &\le \sfrac{1}{16}\fe_\fl(n)
\tag{\ref{eqnOSFineqs}.a}
\end{align}
By \cite[Definition \defHTbasicnorm]{PAR1}, 
\begin{align}
L^{18} \sfrac{\bar\ka_\fl}{\bar\ka'}\ \sdf(\CC_\fl)\ \fv_0^\eps
                    &\le \sfrac{1}{16}\fe_\fl(n)
\tag{\ref{eqnOSFineqs}.b}
\end{align}
By \cite[Lemma \lemPARestrpnC.a]{PAR1} and \eqref{eqnOSFtable}, 
\begin{align}
 L^5\big(\sfrac{\CC_\fl}{L}\big)^{\De(\vp)}
           \ \sfrac{1}{\si_n(\vp)}\ \bar\ka^\vp\,\fr_\vp(n,\CC_\cR)
       &\le \sfrac{1}{16}\fe_\fl(n)\qquad\text{for each $\vp\in \fD$}
&\tag{\ref{eqnOSFineqs}.c}
\end{align}
provided $\fv_0$ is small enough that the hypothesis 
$\,\eps |\log\fv_0| \ge 2\log (1+\CC_\cR)\,\Pi_0^\infty(\CC_\cR)\,$
of \cite[Lemma \lemPARestrpnC]{PAR1} is satisfied.
By Lemma \ref{lemOSFmainlem} and \eqref{eqnOSFineqs} we have
$
\lun \tilde\cD\run
\le \half \fe_\fl(n)
\le\sfrac{1}{32}
$
by \cite[(\eqnPARestrad.a)]{PAR1}.
\cite[Theorem 3.4]{CPC} yields the existence of an
analytic function $\tilde\cE_\fl\big(\tilde\psi_*,\tilde\psi)$
such that 
\begin{equation*}
\frac{\int d\mu_{r_n}(z^*,z)\ 
e^{\tilde\cD(\tilde\psi_*,\tilde\psi,z_*,z)}}
{\int d\mu_{r_n}(z^*,z)\ 
e^{\tilde\cD(0,0,z_*,z)}}
=e^{\tilde\cE_\fl(\tilde\psi_*,\tilde\psi)}
\end{equation*}
and
$
\lun\tilde\cE_\fl\run \le  \fe_\fl(n)
$.
\end{proof}

In order to renormalize the chemical potential we will need 
more detailed information about the monomial in 
$\tilde\cE_\fl(\tilde\psi_*,\tilde\psi)$ that is of type $\psi_*\psi$.
To extract that information, we need more detailed information
about the part of 
\begin{align*}
\tilde\cD(\tilde\psi_*,\tilde\psi,z_*,z)
&=-\de A_n(\psi_*,\psi,z_*,z)  
  +\de\tilde\cE_n\big(\tilde\psi_*,\tilde\psi,z_*,z\big)    
  +\sum_{\vp\in\fD}
\de \tilde\cR_n^{(\vp)}\big(\tilde\psi_*,\tilde\psi, z_*,z\big)
\end{align*}
that is of
degree at most one in each of $\psi_*$ and $\psi$ and is of degree zero
in $\psi_{(*)\nu}$. So we define, on the space of field maps
$\tilde\cG(\tilde\psi_*,\tilde\psi,z_*,z)$ with $\lun\tilde\cG\run<\infty$,
the projections
\begin{itemize}[leftmargin=*, topsep=2pt, itemsep=2pt, parsep=0pt]
\item
$P^\psi_2$ 
which extracts the part which is of degree exactly one in each of $\psi_*$ 
and $\psi$, of degree zero in the $\psi_{(*)\nu}$'s and of arbitrary degree
in $z_{(*)}$ and
\item
$P^\psi_1$ 
which extracts the part which is of degree exactly one in $\psi_{(*)}$, 
of degree zero in the $\psi_{(*)\nu}$'s and of arbitrary degree
in $z_{(*)}$  and
\item
$P^\psi_0$  which extracts the part which is of degree zero in $\psi_{(*)}$ and the $\psi_{(*)\nu}$'s and of arbitrary degree
in $z_{(*)}$. 
\end{itemize}

\begin{lemma}\label{lemOSFcDtwo}
There is a constant $\LLa_1$, depending only on $L$, $\Gam_\op$, 
$\GGa_\bg$ and $\rrho_\bg$, such that the following holds. 
\begin{enumerate}[label=(\alph*), leftmargin=*]
\item
If $n=0$,
$
\ \lun P^\psi_0\de A_0 \run
    \le \LLa_1\big(|\mu_0|+\fv_0\big)\bar\ka_\fl^4\ 
$
,
$
\ \lun P^\psi_1\de A_0 \run
    \le \LLa_1 \fv_0\bar\ka\bar\ka_\fl^3 \ 
$
\newline and
$
\ P^\psi_2\,\de A_0 =-\cM_0 \ 
$
where
\begin{align*}
\cM_0(\psi_*,\psi,z_*,z)&=  - \sfrac{2}{L^3} 
\int_{\cX_0}\! dx_1 \cdots dx_4 \int_{\cX_0}\! dx'_1 dx'_4
\int_{\cX_0^{(1)}}\!\! dx'_2 dx'_3\  V_0(x_1,x_2,x_3,x_4)\\
\noalign{\vskip-0.05in}
&\hskip2.1in  D^{(0)*}(x_1,x'_1)\ z_*(\bbbl^{-1}x'_1)\\
&\hskip2.1in  
  (S_1(L^2\mu_0)Q_1^* \fQ_1)(\bbbl^{-1}x_2,x'_2)\ \psi(x'_2)\\
&\hskip2.1in  
  (S_1(L^2\mu_0)^*Q_1^* \fQ_1)(\bbbl^{-1}x_3,x'_3)\ \psi_*(x'_3)\\
&\hskip2.1in  
  D^{(0)}(x_4,x'_4)\ z(\bbbl^{-1}x'_4)
\end{align*}
If $n\ge 1$,
$
\ \lun P^\psi_0\de A_n \run
    \le \LLa_1\big(|\mu_n|+\fv_{n+1}\big)\bar\ka_\fl^4\ 
$
,
$
\ \lun P^\psi_1\de A_n \run
    \le \LLa_1 \fv_{n+1}\bar\ka\bar\ka_\fl^3 \ 
$\newline
and 
$
\ \lun P^\psi_2\,\de A_n +\cM_n\run
   \le \LLa_1\fv_{n+1}^2\bar\ka^2\bar\ka_\fl^4\ 
$
where
\begin{align*}
\cM_n&=  -\sfrac{2}{L^3}\int_{\cX_n} du_1 \cdots du_4 \int_{\cX_0^{(n+1)}} dx_2 dx_3 \int_{\cX_0^{(n)}}
         dx_1 dx_4\  V_n(u_1,u_2,u_3,u_4)\\
\noalign{\vskip-0.05in}
&\hskip2.4in  (D^{(n)}\fQ_n\, Q_n S_n(\mu_n))(x_1,u_1)\ z_*(\bbbl^{-1}x_1)\\
&\hskip2.4in  
  (S_{n+1}(L^2\mu_n)Q_{n+1}^* \fQ_{n+1})(\bbbl^{-1}u_2,x_2)\ \psi(x_2)\\
&\hskip2.4in  
  (S_{n+1}(L^2\mu_n)^*Q_{n+1}^* \fQ_{n+1})(\bbbl^{-1}u_3,x_3)\ \psi_*(x_3)\\
&\hskip2.4in  
  (S_n(\mu_n) Q_n^* \fQ_n\, D^{(n)})(u_4,x_4)\ 
                                               z(\bbbl^{-1}x_4) \\
\end{align*}
fulfills 
$
\ \lun \cM_n\run
   \le \LLa_1\fv_n\bar\ka^2\bar\ka_\fl^2\ 
$.

\item
$\qquad
\lun \big(P^\psi_2+P^\psi_1+P^\psi_0\big)\de\tilde\cE_n\run
  \le L^{31}\ \sdf(\CC_\fl)\ \sfrac{\bar\ka_\fl^2}{\bar\ka'^2}\ 
        \big\| \tilde\cE_n\big\|^{(n)}
$

\item
For each $\vp\in\fD$,
\begin{align*}
\lun
P^\psi_2\de\tilde\cR_n^{(\vp)}
\run
&\le \LLa_1\,\fr_\vp(n,\CC_\cR)
\begin{cases}
 \bar\ka^2\bar\ka_\fl^4 
& if \text{$\vp=(6,0,0)$}\\ \noalign{\vskip0.05in}
 \bar\ka\bar\ka_\fl
& \text{if $\vp=(1,1,0)$}\\ \noalign{\vskip0.05in}
  \bar\ka_\fl^2
& \text{if $\vp=(0,1,1),(0,0,2)$}
\end{cases}
\\\noalign{\vskip0.05in}%%%%%%%%%%%%%%%%
\lun
P^\psi_1\de\tilde\cR_n^{(\vp)}
\run
&\le \LLa_1\,\fr_\vp(n,\CC_\cR)
\begin{cases}
 \bar\ka\bar\ka_\fl^5
& \text{if $\vp=(6,0,0)$}\\ \noalign{\vskip0.05in}
 \bar\ka\bar\ka_\fl
& \text{if $\vp=(1,1,0)$}\\ \noalign{\vskip0.05in}
 \bar\ka_\fl^2
& \text{otherwise}
\end{cases}
\\\noalign{\vskip0.05in}%%%%%%%%%%%%%%%%%
\lun
P^\psi_0\de\tilde\cR_n^{(\vp)}
\run
&\le \LLa_1\,\fr_\vp(n,\CC_\cR)
\begin{cases}
 \bar\ka_\fl^6
& \text{if $\vp=(6,0,0)$}\\ \noalign{\vskip0.05in}
  \bar\ka_\fl^2
& \text{otherwise}
\end{cases}
\end{align*}
\end{enumerate}
\end{lemma}
\begin{proof} In this proof, when we say that a contribution has norm of order $xyz$,
we mean that the $\lun \cdot\ \run$ norm of the contribution is bounded by a
constant, depending only on $L$, $\Gam_\op$, $\GGa_\bg$ and $\rrho_\bg$, 
times $xyz$.

\Item (a) We first consider the case  $n=0$. By 
            \cite[(\eqnOSAdeAndef.b),   %%% def of \de A_0   
                  Remark \remCFpsisolnZero\
                  and Proposition \propBGEphivepssoln.a]{BGE}  %%% \hat\psi_{(*)0}
            and \cite[Definition \defSCscaling]{PAR1}, %%% scaling
\begin{align*}
&\de A_0(\psi_*,\psi,z_*,z)%\\&\hskip0.5in
=  \!\!\int_0^1\hskip-4pt  (1 - t)\,\sfrac{d^2\hfill}{dt^2} 
   (\bbbs\cV_0)\big(\hat\psi_*+t\de\psi_*\,,\,
               \hat\psi+t\de\psi\big)  \, dt
         \bigg|_{\atop{\hat\psi_{(*)}=\hat\psi_{0(*)}(\psi_*,\psi,\mu_0,\cV_0)}
                  {\de\psi_{(*)}=L^{3/2}\bbbs D^{(0)(*)}\bbbs^{-1}z_{(*)}}}\\
&\hskip2in +\mu_0 L^5\big<z_*,\,
                       \bbbs C^{(0)}\bbbs^{-1}z\big>_1
\displaybreak[0]\\
&\hskip0.5in=  \half\cdot 2\cdot 2\int_0^1\hskip-4pt  (1 - t)\,\sfrac{d^2\hfill}{dt^2} 
\int_{\cX_0}\!\! dx_1 \cdots dx_4 \int_{\cX_0}\! dx'_1 dx'_4
\int_{\cX_0^{(1)}}\!\! dx'_2 dx'_3\  V_0(x_1,x_2,x_3,x_4)\\
\noalign{\vskip-0.05in}
&\hskip3.0in  D^{(0)*}(x_1,x'_1)\ tz_*(\bbbl^{-1}x'_1)\\
&\hskip3.0in  
  L^{-3/2}(S_1(L^2\mu_0)Q_1^* \fQ_1)(\bbbl^{-1}x_2,x'_2)\ \psi(x'_2)\\
&\hskip3.0in  
  L^{-3/2}(S_1(L^2\mu_0)^*Q_1^* \fQ_1)(\bbbl^{-1}x_3,x'_3)\,\psi_*(x'_3)\\
&\hskip3.0in  
  D^{(0)}(x_4,x'_4)\ tz(\bbbl^{-1}x'_4)+\ho
\displaybreak[0]\\
&\hskip0.5in=   \sfrac{2}{L^3} 
\int_{\cX_0}\! dx_1 \cdots dx_4 \int_{\cX_0}\! dx'_1 dx'_4
\int_{\cX_0^{(1)}}\!\! dx'_2 dx'_3\  V_0(x_1,x_2,x_3,x_4)\\
\noalign{\vskip-0.05in}
&\hskip3.0in  D^{(0)*}(x_1,x'_1)\ z_*(\bbbl^{-1}x'_1)\\
&\hskip3.0in  
  (S_1(L^2\mu_0)Q_1^* \fQ_1)(\bbbl^{-1}x_2,x'_2)\ \psi(x'_2)\\
&\hskip3.0in  
  (S_1(L^2\mu_0)^*Q_1^* \fQ_1)(\bbbl^{-1}x_3,x'_3)\ \psi_*(x'_3)\\
&\hskip3.0in  
  D^{(0)}(x_4,x'_4)\ z(\bbbl^{-1}x'_4)+\ho
\end{align*}
with the contributions in $\ho$ being either
\begin{itemize}[leftmargin=*, topsep=2pt, itemsep=2pt, parsep=0pt]
\item
independent of $\psi_{(*)}$ with norms of order 
                       $(|\mu_0|+\fv_0)\bar\ka_\fl^4$ or
\item
of order precisely one in
  $\psi_{(*)}$ with norms of order $\fv_0\bar\ka\bar\ka_\fl^3$ or
\item
of order at least two in $\psi$ or of order at least two
in $\psi_*$ or of order at least three in $\psi_{(*)}$.
\end{itemize}

\medskip\noindent
We now consider the case $n\ge 1$. By
  \cite[(\eqnOSAdeAndef.a)]{PAR1},   %%% def of \de A_n
\begin{equation}\label{eqnOSFdeAn}
\begin{split}
\de A_n(\psi_*,\psi,z_*,z)
&= -L^{7/2} \int_0^1 \!dt\,  \big< z_*,\bbbs D^{(n)}\fQ_n\, Q_n\bbbs^{-1}\,
\de{\hat\phi}^{(+)}_{n+1}\big(\psi_*,\psi;
                  t\,z_*,t\,z\big) \big>_1
\\
& \hskip0.2in
-L^{7/2}\int_0^1 \!dt\,  \big< \bbbs D^{(n)*}\fQ_n\, Q_n\bbbs^{-1}\,
\de{\hat\phi_{*n+1}}^{(+)}\big(\psi_*,\psi;t\,z_*,t\,z\big)
,\, z \big>_1 
\end{split}
\end{equation}
By \cite[Proposition \propBGEdephisoln.d]{BGE}, using the notation of 
\cite[Definition \defBGAgradV]{PAR1},
\begin{align*}
\de\hat\phi_{(*)n+1}^{(+)}(\psi_*,\psi,z_*,z)
&=L^{3/2}\bbbs [S_n(\mu_n)^{(*)}\!\!-\!S_n^{(*)}]Q_n^*\fQ_n 
                                   D^{(n)(*)}\bbbs^{-1}z_{(*)}
\\&\hskip0.5in
  - L^{\frac{3}{2}}\bbbl_*^{-1}{S_n(\mu_n)^{(*)}}\cV'_{(*)}(\varphi_{(*)}, \varphi_{(\bar*)},\varphi_{(*)}) 
      \Big|^{\atop{\varphi_*=\phi_*+\de\phi_*}
                  {\varphi=\phi +\de\phi}}
     _{\atop{\varphi_*=\phi_*}
            {\varphi=\phi}}
     + \de\hat\phi_{(*)}^{(\ho)}
\end{align*}
with the substitutions
\begin{equation}\label{eqndemuSub}
\begin{split}
\phi_{(*)}
    &=\bbbs^{-1}S_{n+1}(L^2\mu_n)^{(*)}Q_{n+1}^* \fQ_{n+1}\,\psi_{(*)}\\
\de\phi_{(*)}
   &=S_n(\mu_n)^{(*)} Q_n^* \fQ_n\,L^{3/2} D^{(n)(*)}\bbbs^{-1} z_{(*)}
\end{split}
\end{equation}
and with the contributions in $\de\hat\phi_{(*)}^{(\ho)}$ being 
of order at least five in $(\psi_{(*)},z_{(*)})$ and obeying
\begin{align*}
  \lutn P^\psi_j\de\hat\phi_{(*)}^{(\ho)} \rutn   
        &\le L^{47}\GGa_\bg\, \fv_{n+1}^2\bar\ka^j\bar\ka_\fl^{5-j} \qquad
  \text{for $j=0,1,2$}
\end{align*}
Here $(\bar *)$ means the opposite of $(*)$ --- i.e nothing when $(*)=*$
and $*$ when $(*)=$nothing.

  So the two terms on the right hand side of \eqref{eqnOSFdeAn} are 
(minus)
\begin{align}
&L^{7/2} \int_0^1 \!dt\,  \big< z_*,\bbbs D^{(n)}\fQ_n\, Q_n\bbbs^{-1}\,
\de{\hat\phi}^{(+)}_{n+1}\big(\psi_*,\psi;
                  t\,z_*,t\,z\big) \big>_1\nonumber\\
&\hskip0.5in=-2L^{7/2} \int_0^1 \!dt\,  
\big< z_*,\bbbs D^{(n)}\fQ_n\, Q_n
 S_n(\mu_n)\cV'(\phi, \phi_*,t\de\phi) \big>_1+\ho\nonumber\\
&\hskip0.5in=-L^{3/2}   
\big< \bbbs^{-1} z_*, D^{(n)}\fQ_n\, Q_n
 S_n(\mu_n)\cV'(\phi, \phi_*,\de\phi) \big>_0+\ho\nonumber\\
&\hskip0.5in
=-L^{-3}\int_{\cX_n} du_1 \cdots du_4 \int_{\cX_0^{(n+1)}} dx_2 dx_3 \int_{\cX_0^{(n)}}
         dx_1 dx_4\  V_n(u_1,u_2,u_3,u_4)\label{eqnOSFintA}\\
\noalign{\vskip-0.05in}
&\hskip2.5in  (D^{(n)}\fQ_n\, Q_n S_n(\mu_n))(x_1,u_1)\ z_*(\bbbl^{-1}x_1)\nonumber\\
&\hskip2.5in  
  (S_{n+1}(L^2\mu_n)Q_{n+1}^* \fQ_{n+1})(\bbbl^{-1}u_2,x_2)\ \psi(x_2)\nonumber\\
&\hskip2.5in  
  (S_{n+1}(L^2\mu_n)^*Q_{n+1}^* \fQ_{n+1})(\bbbl^{-1}u_3,x_3)\ \psi_*(x_3)\nonumber\\
&\hskip2.5in  
  (S_n(\mu_n) Q_n^* \fQ_n\, D^{(n)})(u_4,x_4)\ 
                                               z(\bbbl^{-1}x_4)+\ho
\nonumber
\end{align}
and, similarly,
\begin{align}
&L^{7/2}\int_0^1 \!dt\,  \big< \bbbs D^{(n)*}\fQ_n\, Q_n\bbbs^{-1}\,
\de{\hat\phi_{*n+1}}^{(+)}\big(\psi_*,\psi;t\,z_*,t\,z\big)
,\, z \big>_1
\nonumber\\
%&\hskip0.5in=-2L^{7/2}\int_0^1 \!dt\,  \big< \bbbs D^{(n)*}\fQ_n\, Q_n 
%  S_n(\mu_n)^*\cV'_*(\phi_*, \phi,t\de\phi_*) 
%,\, z \big>_1+\ho\nonumber\\
%&\hskip0.5in=-L^{3/2}  \big<  D^{(n)*}\fQ_n\, Q_n 
%  S_n(\mu_n)^*\cV'_*(\phi_*, \phi,\de\phi_*) 
%,\, \bbbs^{-1}z \big>_0+\ho\nonumber\\
&\hskip0.5in
=-L^{-3}\int_{\cX_n} du_1 \cdots du_4 \int_{\cX_0^{(n+1)}} dx_1 dx_2\int_{\cX_0^{(n)}}
         dx_3 dx_4\ V_n(u_1,u_2,u_3,u_4)\label{eqnOSFintB}\\
\noalign{\vskip-0.05in}
&\hskip2.5in  (D^{(n)*}\fQ_n\, Q_n S_n(\mu_n)^*)(x_4,u_4)\,  
                                                z(\bbbl^{-1}x_4)
\nonumber\\
&\hskip2.5in  
  (S_{n+1}(L^2\mu_n)^*Q_{n+1}^* \fQ_{n+1})(\bbbl^{-1}u_1,x_1)\,\psi_*(x_1)
\nonumber\\
&\hskip2.5in  
  (S_{n+1}(L^2\mu_n)Q_{n+1}^* \fQ_{n+1})(\bbbl^{-1}u_2,x_2)\,\psi(x_2)\\
&\hskip2.5in  
  (S_n(\mu_n)^* Q_n^* \fQ_n\, D^{(n)*})(u_3,x_3)\,
                                               z_*(\bbbl^{-1}x_3)+\ho
\nonumber
\end{align}
with the contributions in $\ho$ being either
\begin{itemize}[leftmargin=*, topsep=2pt, itemsep=2pt, parsep=0pt]
\item
independent of $\psi_{(*)}$ with norm of order 
                       $\big(|\mu_n|+\fv_{n+1}\big)\bar\ka_\fl^4$ or
\item
of order precisely one in
  $\psi_{(*)}$ with norm of order $\fv_{n+1}\bar\ka\bar\ka_\fl^3$ or
\item
of order precisely two in
  $\psi_{(*)}$ with norm of order $\fv_{n+1}^2\bar\ka^2\bar\ka_\fl^4$ or
\item
of order at least two in $\psi$ or of order at least two
in $\psi_*$ or of order at least three in $\psi_{(*)}$.
\end{itemize}
Finally, observe that the two integrals of the right hand sides
of \eqref{eqnOSFintA} and \eqref{eqnOSFintB}
are equal --- just make the change of variables $x_1\leftrightarrow x_3$,
$u_1\leftrightarrow u_3$ in one of them.

\Item (b)
The bound on $\de\tilde\cE_n$ given in Lemma \ref{lemOSFmainlem}.b is
\begin{align*}
\lun \de\tilde\cE_n\run
  &\le L^{18}\ \sfrac{\bar\ka_\fl}{\bar\ka'}
\ \sdf(\CC_\fl)\ \big\| \tilde\cE_n\big\|^{(n)}
%\\  &\le L^{18}\ L^{-(\eta'-\eps/2)(n+1)}\fv_0^{1/3-3\eps/2}
%\ \sdf(\CC_\fl)\ \fv_0^\eps\\ &
  \le L^{13}\ L^{-(\eta'-\eps/2)(n+1)}\fv_0^{1/3-\eps/2}
\end{align*}
which implies that the kernel of the monomial of type $\psi_*\psi$ in 
$\de\tilde\cE_n$ has $L^1$--$L^\infty$ norm bounded by
\begin{equation*}
L^{13}\ L^{-(\eta'-\eps/2)(n+1)}\fv_0^{1/3-\eps/2}\ \sfrac{1}{\bar\ka^2}
=L^{13}\ L^{-(2\eta+\eta'-\eps/2)(n+1)}\fv_0^{1-5\eps/2}
\end{equation*}
This is not adequate for our present purposes. We want a power of $\fv_0$
that is strictly greater than one  --- the contributions from $\de\tilde\cE_n$
should be of higher order than the dominant contributions coming from $\de
A_n$, which are of order exactly one in the coupling constant. We can get 
that by exploiting the fact that we only care about the part of $\de\tilde\cE_n(\tilde\psi_*,\tilde\psi,z_*,z)$
that is of degree at most $2$ in $\psi_{(*)}$ and of degree zero in the
$\psi_{(*),\nu}$'s.

Recall that
\begin{align*}
\de\tilde\cE_n(\tilde\psi_*,\tilde\psi,z_*,z)
&=(\bbbs \tilde\cE_n)\big(\tilde\Psi_*,\tilde\Psi\big)\Big|
       _{\atop{\Psi_{(*)}=\hat\psi_{(*)n}(\psi_*,\psi,\mu_n,\cV_n)
                        +L^{3/2}\bbbs D^{(n)(*)}\bbbs^{-1}z_{(*)}}
               {\Psi_{(*)\nu}= 
            \hat \psi_{(*)n,\nu}(\psi_*,\psi,\psi_{*\nu},\psi_\nu,\mu_n,\cV_n)
              +L^{3/2}\bbbs_\nu\partial_\nu D^{(n)(*)}\bbbs^{-1}z_{(*)}}}\\
&\hskip0.5in  -(\bbbs\tilde\cE_n)\big(\tilde\Psi_*,\tilde\Psi\big)\Big|
       _{\atop{\Psi_{(*)}=\hat\psi_{(*)n}(\psi_*,\psi,\mu_n,\cV_n)}
              {\Psi_{(*)\nu}= 
            \hat \psi_{(*)n,\nu}(\psi_*,\psi,\psi_{*\nu},\psi_\nu,\mu_n,\cV_n)}}
\end{align*}
with the $\tilde\cE_n$ of \cite[Theorem \thmTHmaintheorem]{PAR1}
and the $\hat \psi_{(*)n,\nu}$ of \cite[Proposition \propCFpsisoln]{BGE}.
Let us denote by $\de\tilde \cE_n^{(\le 2)}
          (\tilde\Psi_*,\tilde\Psi,\de\tilde\Psi_*,\de\tilde\Psi_*)$ 
the part of 
\begin{equation*}
(\bbbs \tilde\cE_n)                      
          \big(\tilde\Psi_*+\de\tilde\Psi_*,\tilde\Psi+\de\tilde\Psi\big)
 -(\bbbs\tilde\cE_n)\big(\tilde\Psi_*,\tilde\Psi\big)
\end{equation*}
that is of degree at most two in $\Psi_{(*)}$ and of degree
zero in $\Psi_{(*)\nu}$. Set 
$
\de\tilde\cE_n^{(\le 2)}(\tilde\psi_*,\tilde\psi,z_*,z)
$
to be $\de\tilde E_n^{(\le 2)}      
           (\tilde\Psi_*,\tilde\Psi,\de\tilde\Psi_*,\de\tilde\Psi_*)$
evaluated at
\begin{align*}
\Psi_{(*)}&=\hat\psi_{(*)n}(\psi_*,\psi,\mu_n,\cV_n)\\
\Psi_{(*)\nu}&=\hat \psi_{(*)n,\nu}(\psi_*,\psi,\psi_{*\nu},\psi_\nu,\mu_n,\cV_n)\\
\de\Psi_{(*)}&=L^{3/2}\bbbs D^{(n)(*)}\bbbs^{-1}z_{(*)}\\
\de\Psi_{(*)\nu}&=L^{3/2}\bbbs_\nu\partial_\nu D^{(n)(*)}\bbbs^{-1}z_{(*)}\\
\end{align*}
Since $\hat\psi_{(*)n}$ is of degree at least one in $\psi_{(*)}$ and
$\hat \psi_{(*)n,\nu}$ is of degree at least one in $\psi_{(*)\nu}$, every
monomial in $\de\tilde\cE_n-\de\tilde\cE_n^{(\le 2)}$ is either of degree
at least one in the $\psi_{(*)\nu}$'s or of degree at least three in 
$\psi_{(*)}$. So
\begin{equation*}
\big(P^\psi_2+P^\psi_1+P^\psi_0\big)\de\tilde\cE_n
=\big(P^\psi_2+P^\psi_1+P^\psi_0\big)
               \de\tilde\cE_n^{(\le 2)}
\end{equation*}
 Note further that, since every monomial in $\tilde\cE_n$
that is of degree at least one in $\psi_{(*)}$ is actually of degree at
least $4$ in $(\psi_{(*)},\psi_{(*)\nu})$, and every monomial in
$\tilde\cE_n$ is of degree at least $2$, we have that every monomial
in $\de\tilde\cE_n^{(\le 2)}$ is of degree at least $2$ in $z_{(*)}$.

We define $\la$, $\la'$ and $\si$ by \eqref{eqnOSFlalapsi}.
As in the proof of Lemma \ref{lemOSFmainlem}.b, denote by $\|\ \cdot\ \|_\la$ 
the (auxiliary) norm with mass $2m$ that assigns the weight factors
$\la$ to the fields $\Psi_{(*)}$ and $\la'$ to the fields $\Psi_{\nu(*)}$.
Then, by \cite[Lemma \lemSUBdiff.a and Proposition \propSUBsubstitution.a]{SUB},
\eqref{eqnOSFsigeone}, \eqref{eqnOSFlabnds} and \eqref{eqnOSFstildecE},
\begin{align*}
\lun \de\tilde\cE_n^{(\le 2)}\run
   &\le \sfrac{1}{\si^2}\|\bbbs\tilde\cE_n\|_\la 
     \le \sfrac{1}{\si^2} L^5\,\sdf(\CC_\fl)\,\big\| \tilde\cE_n\big\|^{(n)}
  \le L^{31}\ \sdf(\CC_\fl)\ \sfrac{\bar\ka_\fl^2}{\bar\ka'^2}\ 
          \big\| \tilde\cE_n\big\|^{(n)}
\end{align*}

\Item (c) 
Recall from Lemma \ref{lemOSFmainlem}.c that 
\begin{align*}
\de \tilde\cR_n^{(\vp)}(\tilde\psi_*,\tilde\psi,z_*,z)
&\!=\!(\bbbs \tilde\cR_n^{(\vp)})\big(\tilde\Phi_*,\tilde\Phi\big)\Big|
       _{\atop{\Phi_{(*)}=\phi_{(*)n+1}(\psi_*,\psi,L^2\mu_n,\bbbs\cV_n)
                     +\de\hat\phi_{(*)n+1}(\psi_*,\psi,z_*,z)}
               {\atop{\Phi_{(*)\nu}= \phi_{n+1(*),\nu}(\psi_*,\psi,\psi_{*\nu},\psi_\nu,L^2\mu_n,\bbbs\cV_n)
                             \hskip0.5in} 
                     {\hskip1in
         +\de\hat\phi_{(*)n+1,\nu}(\psi_*,\psi,\psi_{*\nu},\psi_\nu,z_*,z)}}}\\
&\hskip0.5in  
  -(\bbbs\tilde\cR_n^{(\vp)})\big(\tilde\Phi_*,\tilde\Phi\big)\Big|
       _{\atop{\Phi_{(*)}=\phi_{(*)n+1}(\psi_*,\psi,L^2\mu_n,\bbbs\cV_n)}
              {\Phi_{(*)\nu}= 
     \phi_{n+1(*),\nu}(\psi_*,\psi,\psi_{*\nu},\psi_\nu,L^2\mu_n,\bbbs\cV_n)}}
\end{align*}
The claims follow from the observations that
\begin{itemize}[leftmargin=*, topsep=2pt, itemsep=2pt, parsep=0pt]
\item
$\phi_{(*)n+1}$ is of degree at least one in $\psi_{(*)}$
and $\phi_{(*)n+1,\nu}$ is of degree one in $\psi_{(*),\nu}$,
and, 
\item
by \cite[Propositions \propBGEphivepssoln.a and \propBGEdephisoln.a,e]{BGE},
\begin{align*}
\lutn \phi_{(*)n+1}(\psi_*,\psi,L^2\mu_n,\bbbs\cV_n) \rutn
&\le \big\{\big\|S_{n+1}^{(*)}Q_{n+1}^* \fQ_{n+1}\big\|_{2m} 
          +\GGa_\bg \rrho_\bg\big\}\ \bar\ka \\
\lutn \de\hat\phi_{(*)n+1}\rutn 
&\le L^{11}\GGa_\bg\, \bar\ka_\fl \\
\lutn \de\hat\phi_{(*)n+1,\nu}\rutn 
&\le L_\nu L^{11}\GGa_\bg\, \bar\ka_\fl
\end{align*}
and, 
\item
by \cite[Lemma \lemSAscaletoscale.a]{PAR1},
\begin{equation*}
\|\bbbs\tilde\cR_n^{(\vp)}\|_{2m}
       \le  L^{5-\De(\vp)} \|\tilde\cR_n^{(\vp)}\|_m
       \le  L^{5-\De(\vp)} \fr_\vp(n,\CC_\cR)       
\end{equation*}
\end{itemize}
\end{proof}

\begin{proposition}\label{propOSFcEtwo}
There is a constant $\LLa_2$, depending only on $L$, $\Gam_\op$, 
$\GGa_\bg$ and $\rrho_\bg$ such that the following holds.
\begin{enumerate}[label=(\alph*), leftmargin=*]
\item 
Denote by $P_{\psi_*\psi}$ the projection, on the space of 
analytic functions $\tilde\cG(\tilde\psi_*,\tilde\psi)$, 
which extracts the part which is of 
degree exactly one in each of $\psi_*$ and $\psi$, of degree zero 
in the $\psi_{(*)\nu}$'s. Also set, for $n\ge 0$,
\begin{equation*}
M'_n(\psi_*,\psi) 
   = \frac{\int d\mu_{r_n}(z^*,z)\ \cM_n(\psi_*,\psi,z^*,z)}
          {\int d\mu_{r_n}(z^*,z)}
\end{equation*}
Then 
\begin{align*}
\| P_{\psi_*\psi}\tilde\cE_\fl-M'_n(\psi_*,\psi)\|_{2m}
\le \LLa_2\,\fv_0\,\big(\fv_0^{\sfrac{1}{3}-5\eps}+ |\mu_n|\big)\bar\ka_\fl^6
\end{align*}

\item 
Set 
\begin{align*}
M_0(\psi_*,\psi)
&=-\sfrac{2}{L^3} 
\int_{\cX_0}\! dx_1 \cdots dx_4 \int_{\cX_0^{(1)}}\! dx'_2 dx'_3\  V_0(x_1,x_2,x_3,x_4)\\
\noalign{\vskip-0.05in}
&\hskip1.9in  
  (S_1(L^2\mu_0)Q_1^* \fQ_1)(\bbbl^{-1}x_2,x'_2)\ \psi(x'_2)\\
&\hskip1.9in  
  (S_1(L^2\mu_0)^*Q_1^* \fQ_1)(\bbbl^{-1}x_3,x'_3)\ \psi_*(x'_3)\\
&\hskip1.9in  
  C^{(0)}(x_4,x_1)
\end{align*}
and, for $n\ge 1$,
\begin{align*}
M_n(\psi_*,\psi)
&=-\sfrac{2}{L^3}\int_{\cX_n} du_1 \cdots du_4 \int_{\cX_0^{(n+1)}} dx_2 dx_3\  
           V_n(u_1,u_2,u_3,u_4)\\
\noalign{\vskip-0.1in}
&\hskip2in  (S_n(\mu_n) Q_n^* \fQ_n\, C^{(n)}\fQ_n\, Q_n S_n(\mu_n))(u_4,u_1)\\
&\hskip2in  
  (S_{n+1}(L^2\mu_n)Q_{n+1}^* \fQ_{n+1})(\bbbl^{-1}u_2,x_2)\ \psi(x_2)\\
&\hskip2in  
  (S_{n+1}(L^2\mu_n)^*Q_{n+1}^* \fQ_{n+1})(\bbbl^{-1}u_3,x_3)\ \psi_*(x_3)
\end{align*}
Then 
\begin{align*}
\| M'_n(\psi_*,\psi)-M_n(\psi_*,\psi)\|_{2m}
\le \LLa_2\,\fv_n\, r_n^2e^{-r_n^2}
\end{align*}
\end{enumerate}
\end{proposition}
\begin{proof} (a) Set
\begin{align*}
\tilde\cD_2(\tilde\psi_*,\tilde\psi,z_*,z)
&=\big(P^\psi_2+P^\psi_1+P^\psi_0\big)\ 
   \Big\{-\de A_n 
  +\de\tilde\cE_n   
  +\smsum_{\vp\in\fD} \de \tilde\cR_n^{(\vp)}\Big\}
\end{align*}
We are to bound
\begin{align*}
P_{\psi_*\psi}\tilde\cE_\fl-M'_n(\psi_*,\psi)
&=P_{\psi_*\psi}
        \ln\bigg[\frac{\int d\mu_{r_n}(z^*,z)\ 
                           e^{\tilde\cD(\tilde\psi_*,\tilde\psi,z_*,z)}}
                      {\int d\mu_{r_n}(z^*,z)\ 
                           e^{\tilde\cD(0,0,z_*,z)}}\bigg]
   -M'_n(\psi_*,\psi)\\
&=P_{\psi_*\psi}
\ln\bigg[\frac{\int d\mu_{r_n}(z^*,z)\ 
                e^{\tilde\cD_2(\tilde\psi_*,\tilde\psi,z_*,z)}}
{\int d\mu_{r_n}(z^*,z)\ 
              e^{\tilde\cD_2(0,0,z_*,z)}}\bigg]-M'_n(\psi_*,\psi)
\end{align*}
We use \cite[Corollary 3.5]{CPC} with $n=1$,
$d\mu$ being the normalized measure 
\begin{equation*}
d\mu_{r_n}/\int d\mu_{r_n}(z^*,z)
\end{equation*}
with $f=\tilde\cD_2$, and with the norm $\|\ \cdot\ \|_w$ of mass $2m$ which 
assigns a weight $w$ (that we shall choose shortly) to $\psi_*,\psi$ and 
the weight $\bar\ka_\fl$ to the fields $z_*$, $z$. We have, by
Lemma \ref{lemOSFcDtwo} and \cite[Theorem \thmTHmaintheorem,
Definition \defHTbasicnorm\ and Lemma \lemPARestrpnC.b]{PAR1},
\begin{equation}\label{eqnOSFpreppZcDtwo}
\begin{split}
\big\|P^\psi_2\tilde\cD_2-\cM_n\big\|_w \sfrac{1}{w^2}
&\le \LLa_1\fv_{n+1}^2 \bar\ka_\fl^4
    + L^{31}\ \sfrac{\bar\ka_\fl^2}{\bar\ka^2\bar\ka'^2}\ 
        \ \sdf(\CC_\fl)\ \big\| \tilde\cE_n\big\|^{(n)}\\
 &\hskip0.4in +\LLa_1\sum_{\vp\in\fD}\fr_\vp(n,\CC_\cR)
       \begin{cases}
             \bar\ka_\fl/\bar\ka
                & \text{if $\vp=(1,1,0)$}\\ \noalign{\vskip0.05in}
            \bar\ka_\fl^2/\bar\ka^2
                & \text{if $\vp=(0,1,1),\ (0,0,2)$}\\ \noalign{\vskip0.05in}
              \bar\ka_\fl^4  
                & \text{if $\vp=(6,0,0)$}
\end{cases}\\
%%%%%
&\le \LLa_1\fv_n^2 \bar\ka_\fl^4   
   +L^{-2(\eta+\eta')(n+1)}\fv_0^{4/3-5\eps} 
                   \fv_0^\eps\bar\ka_\fl^2
    + \fv_0^{\sfrac{\eps}{2}}\bar\ka_\fl^2
    \min\big\{\fv_0^{\sfrac{4}{3}-7\eps},\sfrac{\fv_0}{L^n}\big\}   \\
%%%%%
&\le \bar\ka_\fl^2
   \min\big\{\fv_0^{\sfrac{4}{3}-7\eps},\fv_n\big\}   
\end{split}
\end{equation}
and
\begin{align*}
\big\|P^\psi_2\tilde\cD_2\big\|_w \sfrac{1}{w^2}
&\le \big\|P^\psi_2\tilde\cD_2-\cM_n\big\|_w \sfrac{1}{w^2}
  +\lun\cM_n\run \sfrac{1}{\bar\ka^2}
%%%%%
\le \LLa'_2\fv_n\bar\ka_\fl^2
\end{align*}
Also 
\begin{align*}
\big\|P^\psi_1\tilde\cD_2\big\|_w \sfrac{1}{w}
 & \le \LLa_1\fv_n \bar\ka_\fl^3
     + L^{31}\ \sfrac{\bar\ka_\fl^2}{\bar\ka\bar\ka'^2}\ 
        \ \sdf(\CC_\fl)\ \big\| \tilde\cE_n\big\|^{(n)}\\
 &\hskip0.5in +\LLa_1\sum_{\vp\in\fD}\fr_\vp(n,\CC_\cR)
       \begin{cases}  \bar\ka_\fl
                  & \text{if $\vp=(1,1,0)$}\\ \noalign{\vskip0.05in}
               \bar\ka_\fl^2/\bar\ka
                 & \text{if $\vp=(0,1,1),\ (0,0,2)$}\\ \noalign{\vskip0.05in}
            \bar\ka_\fl^5 
                  & \text{if $\vp=(6,0,0)$}
              \end{cases}\\
&\le \LLa_1\sfrac{\fv_0}{L^n} \bar\ka_\fl^3
   +L^{31}\sdf(\CC_\fl)\ L^{-(\eta+2\eta')n}\fv_0^{1-2\eps}\ \bar\ka_\fl^2 
   +\fv_0^{1-\sfrac{9}{2}\eps}\bar\ka_\fl^3\\
&\le \fv_0^{1-5\eps}\bar\ka_\fl^3
\end{align*}
if $\fv_0$ is small enough,
by \cite[Definition \defHTbasicnorm, Lemma \lemPARcompradan.d]{PAR1}, in 
the case (1,1,0) and \cite[Lemma \lemPARestrpnC.b]{PAR1} in the other cases.
Furthermore
\begin{align*}
\big\|P^\psi_0\tilde\cD_2\big\|_w 
&\le \LLa_1\big(|\mu_n|+\fv_n\big)\bar\ka_\fl^4
    + L^{31}\ \sfrac{\bar\ka_\fl^2}{\bar\ka'^2}\ 
         \sdf(\CC_\fl)\ \big\| \tilde\cE_n\big\|^{(n)}\\
 &\hskip0.5in +\LLa_1\sum_{\vp\in\fD}\fr_\vp(n,\CC_\cR)
       \begin{cases} \bar\ka_\fl^6 
                  & \text{if $\vp=(6,0,0)$}\\ \noalign{\vskip0.05in}
               \bar\ka_\fl^2
               & \text{otherwise}
\end{cases}\\
&\le \LLa_1\big(|\mu_n|+\fv_n\big)\bar\ka_\fl^4
   + L^{31}\sdf(\CC_\fl)\ L^{-2\eta'n}\fv_0^{2/3-\eps}\ \bar\ka_\fl^2
   +\fv_0^{\frac{2}{3}}\bar\ka_\fl^4 \\
&\le \LLa'_2 
    \big(|\mu_n|+\fv_0^{\frac{2}{3}-\eps}\big)\bar\ka_\fl^4
\end{align*}
by \cite[Definition \defHTbasicnorm\ and Lemma \lemPARcompradan.d]{PAR1}.
We now set
\begin{equation*}
w=\bar\ka_\fl\sqrt{\frac{|\mu_n|+\fv_0^{\frac{2}{3}-\eps}}{\fv_0}}
\end{equation*}
so that
\begin{align*}
\big\|\tilde\cD_2\big\|_w 
&\le 
\LLa'_2\bigg\{
   L^{-n} \big(|\mu_n|+\fv_0^{\frac{2}{3}-\eps}\big)
 + \fv_0^{1-5\eps}
      \sqrt{\frac{|\mu_n|+\fv_0^{\frac{2}{3}-\eps}}{\fv_0}}
  + \big(|\mu_n|+\fv_0^{\frac{2}{3}-\eps}\big)\bigg\}
  \bar\ka_\fl^4 \\
&\le 3 \LLa'_2 \big(|\mu_n|+\fv_0^{\frac{2}{3}-\eps}\big)
  \bar\ka_\fl^4 
\end{align*}
and, by \cite[Corollary 3.5]{CPC} with $d\mu(z^*,z)$ being the normalized
$d\mu_{r_n}(z^*,z)$,
\begin{align*}
\Big\|P_{\psi_*\psi}\Big(\tilde\cE_\fl
   -{\tst\int} d\mu\ \tilde\cD_2(\psi_*,\psi,z^*,z)\Big)\Big\|_w
&\le 40^2 \big\|\tilde\cD_2\big\|_w^2
\le \big(120 \LLa'_2\big)^2 \big(|\mu_n|+\fv_0^{\frac{2}{3}-\eps}\big)^2
  \bar\ka_\fl^8
\end{align*}
Hence
\begin{align*}
\Big\|P_{\psi_*\psi}\Big(\tilde\cE_\fl
   -{\tst\int} d\mu\ \tilde\cD_2(\psi_*,\psi,z^*,z)\Big)\Big\|_{2m}
&\le \sfrac{1}{w^2}\Big\|P_{\psi_*\psi}\Big(\tilde\cE_\fl
   -{\tst\int} d\mu\ \tilde\cD_2(\psi_*,\psi,z^*,z)\Big)\Big\|_w \\
&\le \big(120 \LLa'_2\big)^2 \fv_0
    \big(|\mu_n|+\fv_0^{\frac{2}{3}-\eps}\big)
    \bar\ka_\fl^6
\end{align*}
It now suffices to observe that, by \eqref{eqnOSFpreppZcDtwo},
\begin{align*}
&\Big\|P_{\psi_*\psi}\Big(M'_n(\psi_*,\psi)
   -{\tst\int} d\mu\ \tilde\cD_2(\psi_*,\psi,z^*,z)\Big)\Big\|_{2m} \\
&\hskip1in
 \le \sfrac{1}{w^2}\big\|P^\psi_2\big(\cM_n(\psi_*,\psi,z_*,z)
   -\tilde\cD_2(\psi_*,\psi,z_*,z)\big)\big\|_w\\
&\hskip1in
 \le \fv_0^{\sfrac{4}{3}-7\eps}\bar\ka_\fl^2
\end{align*}

\Item (b) This follows from the observations that, for $r_n\ge 1$,
\begin{align*}
\bigg|\frac{\int d\mu_{r_n}(z^*,z)\ z^*(\bbbl^{-1}x_1)z(\bbbl^{-1}x_2)}
          {\int d\mu_{r_n}(z^*,z)}\ -\ \de_{x_1,x_2}\bigg|
\le  2r_n^2e^{-r_n^2}
\end{align*}
and, recalling that ${D^{(0)}}^*$ is the transpose of $D^{(0)}$,
\begin{align*}
\int_{\cX_0} dx'\ {D^{(0)}}^*(x_1,x)\,D^{(0)}(x_4,x) &= C^{(0)}(x_4,x_1) \\
\int_{\cX_0^{(n)}} dx'\ D^{(n)}(x_1,x)\,D^{(n)}(x,x_2) &= C^{(n)}(x_1,x_2) 
\end{align*}

\end{proof}

\newpage
%%%%%%%%%%%%%%%%%%%%%%%%%%%%%%%%%%
\section{One Block Spin Transformation --- Renormalization and Conclusion of the Induction}\label{chapRen}
%%%%%%%%%%%%%%%%%%%%%%%%%%%%%%%%%%

Lemma \ref{lemOSFmainlem} and Proposition \ref{propOSFmainprop}  provide, 
for each $0\le n\le \np$, an integral free representation for the
$e^{\cC_n} \cF_n=\tilde N^{(n)}_\bbbt\cZ_n
\Big(\!(\bbbs \bbbt_n^{(SF)}) \circ \cdots \circ (\bbbs \bbbt_0^{(SF)})\!\Big) 
\big(\!e^{\cA_0 }\! \big) (\psi_*,\psi)$ of
\cite[Corollary \corSTmainCor]{PAR1}. It is
\begin{equation}\label{eqnREBsofar}
\begin{split}
&\log e^{\cC_n} \cF_n +\log\cZ'_n\\
&= \Big[- A_{n+1}(\psi_*,\psi,\phi_*,\phi,L^2\mu_n,\bbbs \cV_n)
+(\bbbs \cR_n)(\phi_*,\phi)\Big]_
{\phi_{(*)}=\phi_{(*)n+1}(\psi_*,\psi,L^2\mu_n,\bbbs\cV_n)}\\\noalign{\vskip-0.05in}
&\hskip2in +\Big[\tilde\cE_{n+1,1}(\tilde\psi_*,\tilde\psi)
+\tilde\cE_\fl(\tilde\psi_*,\tilde\psi)\Big]_
    {\tilde\psi_{(*)}=(\psi_{(*)},\{\partial_\nu\psi_{(*)}\})}
\end{split}
\end{equation}
To convert this representation into the form specified in 
\cite[Theorem \thmTHmaintheorem\  and Remark \remHTpreciseinduction]{PAR1}
(with $n$ replaced by $n+1$) we shall
\begin{itemize}[leftmargin=*, topsep=2pt, itemsep=2pt, parsep=0pt]
\item
 move the scaling/weight relevant part of $\tilde\cE_\fl$
               into $\cR_{n+1}$ and
\item renormalize the chemical potential and the interaction.
\end{itemize}
We again fix any $0\le n\le \np$ and assume that, if $n\ge 1$, 
the conclusions of \cite[Theorem \thmTHmaintheorem\ 
and Remark \remHTpreciseinduction]{PAR1} hold. In the case of $n=0$, we use the 
data of \cite[\S\sectINTstartPoint]{PAR1}. 
In this section, we will show that these assumptions imply 
\cite[Theorem \thmTHmaintheorem\ and Remark \remHTpreciseinduction]{PAR1}
with $n$ replaced by $n+1$, thus concluding the induction step.

We shall construct a chemical potential $\mu_{n+1}$, an interaction $\cV_{n+1}$
and a polynomial $\cR_{n+1}$ that fulfil the conclusions of 
\cite[Theorem \thmTHmaintheorem\ and Remark \remHTpreciseinduction]{PAR1}, 
and an analytic function  
$\,
\tilde\cE_{n+1,2}(\tilde\psi_*,\tilde\psi)
\,$
whose power series expansion does not
contain scaling/weight relevant monomials 
such that 
\begin{equation}\label{eqnRENrearrange}
\begin{split}
&\Big[- A_{n+1}(\psi_*,\psi,\phi_*,\phi,L^2\mu_n,\bbbs \cV_n)
+(\bbbs \cR_n)(\phi_*,\phi)\Big]_
{\phi_{(*)}=\phi_{(*)n+1}(\psi_*,\psi,L^2\mu_n,\bbbs\cV_n)}\\ \noalign{\vskip-0.05in}
&\hskip2.5in
+\tilde\cE_\fl(\tilde\psi_*,\tilde\psi)\Big|_
    {\tilde\psi_{(*)}=(\psi_{(*)},\{\partial_\nu\psi_{(*)}\})}\\
&=\Big[- A_{n+1}(\psi_*,\psi,\phi_*,\phi,\mu_{n+1},\cV_{n+1})
         +\cR_{n+1}(\phi_*,\phi)\Big]_
{\phi_{(*)}=\phi_{(*)n+1}(\psi_*,\psi,\mu_{n+1},\cV_{n+1})}\\ \noalign{\vskip-0.05in}
&\hskip2.5in
+\tilde\cE_{n+1,2}(\tilde\psi_*,\tilde\psi)\Big|_
    {\tilde\psi_{(*)}=(\psi_{(*)},\{\partial_\nu\psi_{(*)}\})}
\end{split}
\end{equation}
With
$ \tilde\cE_{n+1}
=\tilde\cE_{n+1,1}
       +\tilde\cE_{n+1,2}$
we will get the desired representation for 
\begin{equation*}
\Big(\!(\bbbs \bbbt_n^{(SF)}) \circ \cdots \circ (\bbbs \bbbt_0^{(SF)})\!\Big) 
\Big(\!e^{\cA_0 }\! \Big) (\psi_*,\psi)
\end{equation*}

The first term on the left hand side of  \eqref{eqnRENrearrange} is written
in terms of the background fields $\,\phi_{(*)n+1}(\psi_*,\psi,L^2\mu_n,\bbbs\cV_n)\,$,
and the second term in terms of the fields $\,\psi_{(*)}\,$ themselves.
For the proof of \eqref{eqnRENrearrange} we will reshuffle this arrangement
--- i.e. write background fields in terms of $\,\psi_{(*)}\,$ fields 
and conversely.
To take care of the special degree properties of the relevant monomials, 
the chemical potential and the interaction, we have to keep track 
of the degrees of the monomials arising in the conversion process. 
The background fields are defined in terms of the $\,\psi_{(*)}\,$ fields
(see  \cite[Proposition \propHTexistencebackgroundfields]{PAR1}), so conversion from $\phi$ fields to $\psi$ is
in principle ``easy''. The converse is taken care of in Lemma \ref{lemRENpsitophi}.

We set up and solve (in Lemma \ref{lemRENrenormchem})
the equation for $\de\mu_n=\mu_{n+1}-L^2\mu_n$.
Then, we derive and  solve the equation for  $\de\cV_n=\cV_{n+1} -\bbbs\cV_n$. 
See Lemma \ref{lemRENreninteraction}.
The polynomial $\cR_{n+1}$ and the function $\tilde\cE_{n+1,2}$ of
\eqref{eqnRENrearrange} are constructed in Lemma \ref{lemRENcRcE}. 

We again use the abbreviations \eqref{eqnOSFabbrevwt}
and the norms $\lun\ \cdot\ \run$ (for analytic functions) and
$\lutn\ \cdot\ \rutn$ (for field maps) with mass $2m$ and 
weight factors $\,\bar\ka\,$, $\,\bar\ka'\,$ defined at the 
beginning of \S \ref{chapOSFfluct}.
For the output of the renormalization procedure, we use the norm
$\|\ \cdot\ \|^{(n+1)}\,$
with mass $m$ which associates the same weight $\,\bar\ka=\ka(n+1)\,$ 
to the fields $\,\psi_*, \,\psi\,$, and the same weight
$\,\bar\ka'=\ka'(n+1)\,$ to the fields $\,\psi_{\nu *},\psi_\nu\,$, 
$\,\nu=0,\cdots,3\,$. We abbreviate $\lln\ \cdot\ \rln
=\|\ \cdot\ \|^{(n+1)}\,$ and use $\lltn\ \cdot\ \rltn$
to denote the corresponding norm for field maps. Recall also, 
from \cite[Definition \defHTbasicnorm]{PAR1},
that $\|\ \cdot\ \|_{2m}$ is the norm with mass $2m$ which associates the
weight $1$ to all fields.

\smallskip
To keep track of relevant monomials, we use
\begin{definition}\label{defRENspaces}
For a vector $\vp=(p_u,p_0,p_\sp)$ of nonnegative integers denote by 
$\fP_\vp$, respectively $\fR_\vp$,
the space of $\fS$ invariant, particle number preserving,  
polynomials in the fields $\tilde\psi_*$, $\tilde\psi$, 
respectively $\tilde\phi_*$, $\tilde\phi$, of type $\vp$, as in 
\cite[Definition \defINTmonomialtype]{PAR1}.
For any analytic function $\tilde\cF(\tilde\psi_*,\tilde\psi)$, denote 
by $\cM_\vp(\tilde\cF)$
the part of $\cF$ that is in $\fP_\vp$. Let
\begin{itemize}[leftmargin=*, topsep=2pt, itemsep=2pt, parsep=0pt]
\item
$\fP_\fD$, respectively $\fR_\fD$,  denote the space of 
$\fS$ invariant, particle number preserving,  
polynomials in the fields $\tilde\psi_*$,$\tilde\psi$, respectively $\tilde\phi_*$,$\tilde\phi$, 
that contain only monomials of type $\vp\in\fD$ as in \cite[(\eqnINTfDdef)]{PAR1}.
\item
$\fP_\rel$  denote the space of $\fS$ invariant, particle number preserving,  
polynomials in the fields $\tilde\psi_*$, $\tilde\psi$
that contain only monomials of type $\vp\in\fD_\rel$ as in 
\cite[Definition \defINTrelmonomial]{PAR1}.
\item
$\fP_\irr$ denote the space of $\fS$ invariant, particle number preserving,  
analytic functions of the fields $\tilde\psi_*$, $\tilde\psi$, 
that contain only scaling/weight irrelevant monomials, i.e. 
of type $\vp\notin\fD_\rel$.
\end{itemize}
$\fP_\fD$ and $\fP_\rel$ are direct sums of $\fP_\vp$'s with
$\vp$ running over the vectors specified in 
\cite[(\eqnINTfDdef) and Definition \defINTrelmonomial]{PAR1}, respectively.
By construction, $\fP_\fD$ is a subspace of $\fP_\rel$.
\end{definition}

\noindent We also use the projections $\cL_4$, $\cL_\fD$ and $\cI$ 
and mass extraction operator $\ell$ of Proposition \ref{propLOCproj}
and Definition  \ref{defLOCrel}.
\vskip .5cm

By Corollary \ref{corLOCproj} and Proposition \ref{propOSFmainprop},
\begin{equation}\label{eqnRENfluctParts}
\begin{split}
\tilde\cE_\fl\big((\psi_*,\{\partial_\nu\psi_*\})\,,\,
                        (\psi,\{\partial_\nu\psi\})\big)
&=\ell(\tilde\cE_\fl)\hskip-2.5pt\int\hskip-3pt dx \ \psi_*(x) \psi(x)
+ \cL_4(\tilde\cE_\fl) (\psi_*,\psi)
\\
&\hskip1in+\cL_\fD(\tilde\cE_\fl) \big((\psi_*,\{\partial_\nu\psi_*\})\,,\,
                        (\psi,\{\partial_\nu\psi\})\big)\\
&\hskip1in  +\cI(\tilde\cE_\fl)\big((\psi_*,\{\partial_\nu\psi_*\})\,,\,
                        (\psi,\{\partial_\nu\psi\})\big)\\
|\ell(\tilde\cE_\fl)|&\le  \sfrac{1}{\bar\ka^2}\fe_\fl(n)
\\
\lun \cL_4(\tilde\cE_\fl) \run\ ,\ \lln \cL_\fD(\tilde\cE_\fl) \rln\ ,\ \lln \cI(\tilde\cE_\fl)\rln\ &\le\ 
\Big[1+18\cc_\loc\sfrac{\bar\ka'}{\bar\ka}\Big]^2\,\fe_\fl(n)
\end{split}
\end{equation}
Set
\begin{equation}\label{eqnRENdefAvar}
\begin{split}
&A^{\var}(\psi_*,\psi,\de\mu,\de\cV) 
\\
&\hskip 1cm= A_{n+1}(\psi_*,\psi,\phi_*,\phi,L^2\mu_n+\de\mu, \bbbs\cV_n)\Big|_
       {\phi_{(*)}=\phi_{(*)n+1}(\psi_*,\psi,L^2\mu_n+\de\mu, \bbbs\cV_n+\de\cV)}\\
&\hskip 4cm
- A_{n+1}(\psi_*,\psi,\phi_*,\phi,L^2\mu_n, \bbbs\cV_n)\Big|_
      {\phi_{(*)}=\phi_{(*)n+1}(\psi_*,\psi,L^2\mu_n, \bbbs\cV_n)}
\end{split}
\end{equation}
Note that the last argument of the first $A_{n+1}$ on the right hand side
is $\bbbs\cV_n$, rather than $\bbbs\cV_n+\de\cV$. See the definition of 
$\tilde A^{\var}$  before \eqref{eqnRENgoodproptildeAvar}.

\begin{lemma}[Renormalization of the Chemical Potential]\label{lemRENrenormchem}
There is a unique $\de\mu_n$ in 
$\big[-\rrho_\bg,\rrho_\bg\big]$ such that for all 
$\de\cV \in\fP_{(4,0,0)}$
\begin{equation*}
\ell\big(A^\var(\,\cdot\,,\,\cdot\,,\de\mu_n,\de\cV)\big)+\ell(\tilde\cE_\fl)=0
\end{equation*}
Furthermore $\de\mu_n$ has the same sign as $\ell(\tilde\cE_\fl)$ and
\begin{equation*}
\sfrac{1}{4}|\ell(\tilde\cE_\fl)|\le
|\de\mu_n|\le \sfrac{9}{4} |\ell(\tilde\cE_\fl)| 
\le \sfrac{3}{\bar\ka^2}\fe_\fl(n)\end{equation*}
\end{lemma}
\begin{proof} 
The part of 
$\,
\phi_{(*)}=\phi_{(*)n+1}(\psi_*,\psi,L^2\mu_n+\de\mu, \bbbs\cV_n+\de\cV)
\,$
that is linear in $\psi_*,\psi$ is
$\,
\bbbB_{(*)} \psi_{(*)} + \de\mu\, \De \bbbB_{(*)}(\de\mu)\, \psi_{(*)}
\,$
with 
\begin{equation}\label{eqnRENdefLDeL}
\begin{split}
\bbbB_{(*)}&= S_{n+1}(L^2\mu_n)^{(*)}Q_{n+1}^*\fQ_{n+1}
\qquad\qquad
\De \bbbB _{(*)}(\de\mu) =  \tilde S(\de\mu)^{(*)} \,\bbbB _{(*)}
\\
\tilde S(\de\mu)^{(*)} 
       &= S_{n+1}^{(*)} \big[ 1-(L^2\mu_n+\de\mu) S_{n+1}^{(*)} \big]^{-1}
\end{split}
\end{equation}
See \cite[Propositions \propBGEphivepssoln\ and \propBGEdephidemu]{BGE}.
In particular it is independent of $\de\cV$. Similarly, the part of 
$\phi_D=D_{n+1}\phi_{(*)n+1}(\psi_*,\psi,L^2\mu_n+\de\mu, \bbbs\cV_n+\de\cV)$
that is linear in $\psi$ is
$\, \bbbB _D\, \psi + \de\mu\, \De \bbbB _D \,\psi
\,$
where
$$
 \bbbB _D = B_{n+1,L^2\mu_n,D}^{(-)}
\qquad
\De \bbbB _D =B_{n+1,L^2\mu_n,D}
$$
Therefore, the part of $A^{\var}$ that is quadratic in the fields $\psi_*,\psi$ is
\begin{equation}\label{eqnRENdefAtwo}
\begin{split}
A^\var_2(\psi_*,\psi,\de\mu)
&=\bigg[\<\psi_*-Q_{n+1}\phi_*,\fQ_{n+1}  \big(\psi-Q_{n+1}\phi\big)\>_0 \\
& \hskip 2cm
 +\< \phi_*,\,\phi_D\>_{n+1} 
-L^2\mu_n\< \phi_*,\,\phi\>_{n+1}\bigg]
    ^{\atop{\phi_{(*)}= (\bbbB _{(*)} +\de\mu\De \bbbB _{(*)})\psi_{(*)}
    \hfill}
           {\phi_D= (\bbbB _D +\de\mu\De \bbbB _D)\psi}}
    _{\atop{\phi_{(*)}=\bbbB _{(*)}\psi_{(*)}\hfill}
           {\phi_D=\bbbB _D\psi}}
\\ 
&\hskip 1cm
-\de\mu\< \phi_*,\phi\>_{n+1}\Big|_ {\phi_{(*)}= (\bbbB _{(*)} +\de\mu\De \bbbB _{(*)})\psi_{(*)}}
\end{split}
\end{equation}
In particular
\begin{equation}\label{eqnRENavarbndsa}
\lun A^\var_2 \run 
   \le \cc_A\big[ 1  +L^2|\mu_n|\big]|\de\mu|\bar\ka^2
\end{equation}

Denote by $1$ and $1_\fin$ the functions on $\cX_0^{(n+1)}$ and 
$\cX_{n+1}$, respectively, which always take the value 1.
By Remark \ref{remMXconsteigen},
\begin{align*}
\bbbB _{(*)}1 &= \sfrac{a_{n+1}}{a_{n+1}-L^2\mu_n}1_\fin 
\qquad \qquad
\De \bbbB _{(*)} 1=   \sfrac{a_{n+1}}{[a_{n+1}-L^2\mu_n-\de\mu][a_{n+1}-L^2\mu_n]}1_\fin
\\
%%%%%%
(\bbbB _{(*)} +\de\mu\,\De \bbbB _{(*)})1
&=\sfrac{a_{n+1}}{a_{n+1}-L^2\mu_n-\de\mu}1_\fin
\\
%%%%%%%
\bbbB _D 1  & = \De \bbbB _D 1 = 0
\end{align*} 
where
\begin{equation}\label{eqnOSRan}
a_{n+1}  =a\big(1 +\smsum_{j=1}^n\sfrac{1}{L^{2j}}\big)^{-1}
\end{equation}
Therefore, by Corollary \ref{corLOCproj}.b.
\begin{equation}\label{eqnRENlavar}
\begin{split}
&\ell\big(A^\var(\,\cdot\,,\,\cdot\,,\de\mu,\de\cV)\big)
=\ell\big(A^\var_2(\,\cdot\,,\,\cdot\,,\de\mu)\big)
=\sfrac{1}{\<1,1\>_0}A^\var_2(1\,,\,1\,,\de\mu)\\
&\hskip0.1in=a_{n+1}\Big[
            \big(\sfrac{L^2\mu_n+\de\mu}{a_{n+1}-L^2\mu_n-\de\mu}\big)^2-
            \big(\sfrac{L^2\mu_n}{a_{n+1}-L^2\mu_n}\big)^2\Big]
-L^2\mu_n\Big[\big(\sfrac{a_{n+1}}{a_{n+1}-L^2\mu_n-\de\mu}\big)^2
             -\big(\sfrac{a_{n+1}}{a_{n+1}-L^2\mu_n}\big)^2\Big]\\
&\hskip1in-\de\mu\Big[\sfrac{a_{n+1}}{a_{n+1}-L^2\mu_n-\de\mu} \Big]^2\\
&\hskip0.1in=-\sfrac{a_{n+1}(L^2\mu_n+\de\mu)}{a_{n+1}-L^2\mu_n-\de\mu}
   +\sfrac{a_{n+1}\,L^2\mu_n}{a_{n+1}-L^2\mu_n}\\
&\hskip0.1in=-\sfrac{a_{n+1}^2}{a_{n+1}-L^2\mu_n-\de\mu}
  +\sfrac{a_{n+1}^2}{a_{n+1}-L^2\mu_n}\\
\end{split}
\end{equation}
This function vanishes when $\de\mu=0$ and has first derivative, with respect
to $\de\mu$, given by
\begin{align*}
\sfrac{\partial\hfill}{\partial\de\mu}
    \ell\big(A_2^\var(\,\cdot\,,\,\cdot\,,\de\mu)\big)
=-\frac{a_{n+1}^2}{[a_{n+1}-(L^2\mu_n+\de\mu)]^2}
=-\frac{1}{\big[1-\sfrac{L^2\mu_n+\de\mu}{a_{n+1}}\big]^2}
\end{align*}
For $|\de\mu|\le\sfrac{1}{8}$, this derivative is between $-\sfrac{4}{9}$
and $-4$. So, as $\de\mu$ runs from $-\sfrac{1}{8}$ to $+\sfrac{1}{8}$,
$\ell\big(A^\var(\,\cdot\,,\,\cdot\,,\de\mu,\de\cV)\big)$ decreases 
monotonically over an interval that contains $[-\sfrac{1}{18},\sfrac{1}{18}]$.
As $\ell(\tilde\cE_\fl)$ is a constant, independent of $\de\mu$, the claims
follow by \eqref{eqnRENfluctParts}.
\end{proof}

\noindent Set $\mu_{n+1}=L^2\mu_n+\de\mu_n$.

\begin{lemma}[$\psi$ to $\phi$ conversion]\label{lemRENpsitophi}
There exists a constant  $\cc_\Om$, depending only on $\Gam_\op$ and
$\GGa_\bg$, and there are maps
\begin{align*}
&\Om:\fP_\fD \rightarrow\fR_\fD  \qquad
\Om_4:\fP_{(4,0,0)}\rightarrow\fR_{(4,0,0)}\qquad
\Om_6:\fP_{(4,0,0)}\times \fR_{(4,0,0)}\rightarrow\fR_{(6,0,0)}\\
&\Om_\irr: (\fP_\fD \oplus \fP_{(4,0,0)} ) \times\fR_{(4,0,0)}  \rightarrow \fP_\irr
\end{align*}
with $\Om$, $\Om_4$ and $\Om_6$ being linear and
with $\Om_\irr$ being linear in the first variable,
such that the following holds for all $\de\cV\in \fR_{(4,0,0)}$.
\begin{itemize}[leftmargin=*, topsep=2pt, itemsep=2pt, parsep=0pt]
%%%%%%
\item
For all $\cP\in\fP_\fD$,
\begin{align*}
\cP\big((\psi_*,\{\partial_\nu\psi_*\})\,,\,
                  (\psi,\{\partial_\nu\psi\})\big)
&= \Om(\cP)(\tilde\phi_*,\tilde\phi)\Big|_
  {\atop{\phi_{(*)}=\phi_{(*),n+1}(\psi_*,\psi,\mu_{n+1},\bbbs\cV_n+\de\cV)}
   {\phi_{(*),\nu}=\partial_\nu\phi_{(*),n+1}(\psi_*,\psi,\mu_{n+1},\bbbs\cV_n
                                                                   +\de\cV)}}\\
&\hskip0.5in
  +\Om_\irr(\cP,\de\cV)\big((\psi_*,\{\partial_\nu\psi_*\})\,,\,
                  (\psi,\{\partial_\nu\psi\})\big)
\end{align*}
and for all $\cP\in\fP_{(4,0,0)}$,
\begin{align*}
\cP(\psi_*, \psi)
&= \Big[ \Om_4(\cP)(\phi_*,\phi) +\Om_6(\cP,\de\cV)(\phi_*,\phi)
\Big]_ {\phi_{(*)}=\phi_{(*),n+1}(\psi_*,\psi,\mu_{n+1},\bbbs\cV_n+\de\cV)}\\
&\hskip0.5in
  +\Om_\irr(\cP,\de\cV)\big((\psi_*,\{\partial_\nu\psi_*\})\,,\,
                  (\psi,\{\partial_\nu\psi\})\big)
\end{align*}

%%%%%%%%%%%%%%
\item 
If $\cP\in\fP_\vp$, for some $\vp\in\fD$ then $\Om(\cP)\in\fR_\vp$ and
\begin{align*}
\|\Om(\cP)\|_m\le \sfrac{\cc_\Om}{\bar\ka^\vp}  \lln\cP\rln\qquad
\|\Om(\cP)\|_{2m}\le \sfrac{\cc_\Om}{\bar\ka^\vp}  \lun\cP\run
\end{align*}

%%%%%%%%%%%%%%%%%%%
\item
If $\cP\in\fP_{(4,0,0)}$, then
\begin{align*}
\|\Om_4(\cP)\|_{2m}\le \sfrac{\cc_\Om}{\bar\ka^4}  \lun\cP\run\qquad
\|\Om_6(\cP,\cV)\|_{2m}
  \le \sfrac{\cc_\Om}{\bar\ka^4}\|\bbbs\cV_n +\de\cV\|_{2m} \,\lun\cP\run
\end{align*}

%%%%%%%%%%%%%%%%
\item
For all $\cP\in\fP_\fD\cup\fP_{(4,0,0)}$,
\begin{align*}
\lln\Om_\irr(\cP,\de\cV)\rln 
&\le \cc_\Om \|\bbbs\cV_n +\de\cV\|_m\,\bar\ka^2\ \lln\cP\rln \\
\lun\Om_\irr(\cP,\de\cV)\run 
&\le \cc_\Om \|\bbbs\cV_n +\de\cV\|_{2m}\,\bar\ka^2\ \lun\cP\run
\end{align*}
\end{itemize}
\end{lemma}

\begin{proof}
Let $B_{(*)}=S_{n+1,\mu_{n+1}}^{(*)}Q_{n+1}^*\fQ_{n+1}$ be the operators
of \cite[Proposition \propBGEphivepssoln.a]{BGE}\footnote{
The hypothesis of this Proposition is fulfilled by \eqref{eqnOSFweightineqsA}.
} 
and $B^{(\pm)}_{n,\mu_{n+1},\nu}$ be the operators of 
\cite[Proposition \propBGEphivepssoln.b]{BGE}.
By \cite[Lemma \lemPOGrightinverse]{POA} the operators $B^*B$, $B_*^*B_*$ and 
$\big(B^{(\pm)}_{n+1,\mu_{n+1},\nu}\big)^* B^{(\pm)}_{n+1,\mu_{n+1},\nu}$
all have bounded inverses. Consequently the operators
\begin{align*}
R_{(*)}&=\big[B_{(*)}^*B_{(*)}\big]^{-1}B_{(*)}^*\\
R_\nu^{(\pm)}
  &=\big[\big(B^{(\pm)}_{n+1,\mu_{n+1},\nu}\big)^* B^{(\pm)}_{n+1,\mu_{n+1},\nu}\big]^{-1}
         \big(B^{(\pm)}_{n+1,\mu_{n+1},\nu}\big)^*\\
\end{align*}
are left inverses of $B_{(*)}$ and $B^{(\pm)}_{n+1,\mu_{n+1},\nu}$,
respectively. All have uniformly bounded $\|\ \cdot\ \|_{2m}$ norms.

For $\cP\in\fP_\vp\,$ and 
$\,\vp \in  \fD\cup \{(4,0,0)\}\,$, set
\begin{equation*}
\Om'(\cP)(\tilde\phi_*,\tilde\phi)
    =\cP\big((R_*\phi_*,\{R_\nu^{(+)}\phi_{*\nu}\})\,,\,
                            (R\phi,\{R_\nu^{(-)}\phi_\nu\})\big)
\end{equation*}
Then
$
\|\Om'(\cP)\|_m\le \sfrac{\cc_\Om}{\bar\ka^\vp}  \lln\cP\rln
$
and
$
\|\Om'(\cP)\|_{2m}\le \sfrac{\cc_\Om}{\bar\ka^\vp}  \lun\cP\run
$.
\begin{itemize}[leftmargin=*, topsep=2pt, itemsep=2pt, parsep=0pt]
\item
If $\vp\in\fD$ and $\cP\in\fP_\vp$, 
we set $\Om(\cP)=\Om'(\cP)$. In this case, by  
\cite[Proposition \propBGEphivepssoln.a,b]{BGE} and
\cite[Corollary \corSUBsubstitution]{SUB},
\begin{align*}
&\Om(\cP)(\tilde\phi_*,\tilde\phi)\Big|_
  {\atop{\phi_{(*)}=\phi_{(*),n+1}(\psi_*,\psi,\mu_{n+1},\bbbs\cV_n+\de\cV)}
      {\phi_{(*),\nu}=\partial_\nu\phi_{(*),n+1}(\psi_*,\psi,\mu_{n+1},  
                                                     \bbbs\cV_n+\de\cV)}}
\\
&\hskip0.5in=\cP\big((\psi_*,\{\partial_\nu\psi_*\})\,,\,
                  (\psi,\{\partial_\nu\psi\})\big)
-\Om_\irr(\cP,\de\cV)\big((\psi_*,\{\partial_\nu\psi_*\})\,,\,
                  (\psi,\{\partial_\nu\psi\})\big)
\end{align*}
with an $\Om_\irr(\cP,\de\cV)\in\fP_\irr$ satisfying 
$\lln\Om_\irr(\cP)\rln 
      \le \cc_\Om \|\bbbs\cV_n +\de\cV\|_m\,\bar\ka^2\ \lln\cP\rln\,$
and
$\lun\Om_\irr(\cP)\run 
      \le \cc_\Om \|\bbbs\cV_n +\de\cV\|_{2m}\,\bar\ka^2\ \lun\cP\run\,$.

\item
If $\cP\in\fP_{(4,0,0)}$, 
then there are $\cP_6\in\fP_{( 6,0,0)}$ and $\Om'_\irr(\cP,\de\cV)\in\fP_\irr$,
fulfilling 
\begin{align*}
 \lln\Om'_\irr(\cP,\de\cV)\rln 
   &\le \cc_\Om \|\bbbs\cV_n +\de\cV\|_m\,\bar\ka^2\ \lln\cP\rln\\
\lun\cP_6\run\ ,\ \lun\Om'_\irr(\cP,\de\cV)\run 
   &\le \cc_\Om \|\bbbs\cV_n +\de\cV\|_{2m}\,\bar\ka^2\ \lun\cP\run
\end{align*}
 such that
\begin{align*}
&\Om'(\cP)(\tilde\phi_*,\tilde\phi)\Big|_
  {\phi_{(*)}=\phi_{(*),n+1}(\psi_*,\psi,\mu_{n+1}, \bbbs\cV_n+\de\cV)}\\
&\hskip0.5in=\cP(\psi_*,\psi)
-\cP_6(\psi_*,\psi)
-\Om'_\irr(\cP,\de\cV)(\psi_*,\psi)
\end{align*}
Set 
\begin{equation*}
\Om_4(\cP)=\Om'(\cP)\quad
\Om_6(\cP,\de\cV)=\Om(\cP_6)\quad
\Om_\irr(\cP,\de\cV)=\Om'_\irr(\cP,\de\cV)+\Om_\irr(\cP_6,\de\cV)
\end{equation*}
\end{itemize}
\end{proof}

\begin{lemma}[Renormalization of the Interaction]\label{lemRENreninteraction}
\ 
\begin{enumerate}[label=(\alph*), leftmargin=*]
\item 
There exists a constant  $\cc_{\de\cV}$, depending only 
on $\Gam_\op$, $\GGa_\bg$, $\rrho_\bg$  and $m$, and a unique 
$\de\cV_n \in \fP_{(4,0,0)}$ such that
\begin{equation*}
\de\cV_n(\phi_*,\phi) + \Om_4\Big(
\cL_4 \big(A^\var(\,\cdot\,,\,\cdot\,,\de\mu_n,\de\cV_n)+\tilde\cE_\fl\big)
\Big)
=0
\end{equation*}
It fulfills the estimate
$
\|\de\cV_n\|_{2m} \le    \cc_{\de\cV}\sfrac{\fe_\fl(n)}{\bar\ka^4}
$.

\item 
Set
\begin{equation*}
\cV_{n+1} = \bbbs\cV_n + \de\cV_n\qquad\text{and}\qquad
\CC_{\de\cV} = \cc_{\de\cV}
\end{equation*}
Then
\begin{equation*}
\big\|\cV_{n+1}-\cV^{(u)}_{n+1}\big\|_{2m}
   \le \sfrac{\CC_{\de\cV}}{L^{n+1}} 
    \smsum_{\ell=1}^{n+1} \sfrac{L^\ell }{\ka(\ell)^4}\fe_\fl(\ell-1) 
%\qquad \|\cV_{n+1}\|_{2m}\le\fv_{n+1}
\end{equation*}
\end{enumerate}
\end{lemma}
\noindent
Part (b) provides our choice for the $\CC_{\de\cV}$ of 
\cite[Remark \remHTpreciseinduction]{PAR1}.
By \cite[Corollary \corPARmunvn]{PAR1},
\begin{equation*}
\big\|\cV_{n+1}\big\|_{2m}\le\fv_{n+1}
\end{equation*}

\begin{proof} (a)
By  \cite[Propositions \propBGEphivepssoln\ and \propBGEdephidemu\ 
and (\eqnBGEDephi)]{BGE}
\begin{equation}\label{eqndiffphiwithmun}
\begin{split}
\phi_{(*)n+1}(\psi_*,\psi,L^2\mu_n,\bbbs\cV_n) 
&= \Phi_{(*)}^{(1)} + \Phi_{(*)}^{(3)}+ \Phi_{(*)}^{(\ge 5)}
\\
\phi_{(*)n+1}(\psi_*,\psi,\mu_{n+1},\bbbs\cV_n+\de\cV)
       -\phi_{(*)n+1}(\psi_*,\psi,L^2\mu_n,\bbbs\cV_n)
&= \De\Phi_{(*)}^{(1)}\!+\! \De\Phi_{(*)}^{(3)}\!+\! \De\Phi_{(*)}^{(\ge 5)}
\end{split}
\end{equation}
where
\begin{equation*}
\Phi_{(*)}^{(1)} = \bbbB _{(*)}\, \psi_*     \qquad\qquad
\De\Phi_{(*)}^{(1)} = \de\mu_n\,\De \bbbB _{(*)}(\de\mu_n)\, \psi_{(*)}   
= \de\mu_n \tilde S(\de\mu_n)^{(*)} \Phi^{(1)}_{(*)}
\end{equation*}
with 
\begin{itemize}[leftmargin=*, topsep=2pt, itemsep=2pt, parsep=0pt]
\item
$\,\bbbB _{(*)}\, $, $\,\De \bbbB _{(*)}(\de\mu_n)\,$ and $\,\tilde S(\de\mu_n) \,$ as in \eqref{eqnRENdefLDeL},
\item
$\,\Phi_{(*)}^{(3)}\,$ is the part of $\,\phi_{(*)n+1}(\psi_*,\psi,L^2\mu_n,\bbbs\cV_n) \,$
that is of degree exactly three in the fields $\psi_*,\psi$,
\item
$
\De\Phi_{(*)}^{(3)} =\varphi_{(*)}^{(c)} +\varphi_{(*)}^{(l)}(\de\cV)
$
with, using the notation of\cite[Definition \defBGAgradV]{PAR1},
\begin{align*}
\varphi_*^{(c)} 
&=\de\mu_n\,\tilde S(\de\mu_n)^*\,\Phi_*^{(3)}
   -\tilde S(\de\mu_n)^*\,(\bbbs\cV_n)'_*(\Phi_*,\Phi,\Phi_*)
    \Big|^{\Phi_{(*)}=[\bbbone+ \de\mu_n \tilde S(\de\mu_n)^{(*)}]
                                                     \Phi_{(*)}^{(1)}}
          _{\Phi_{(*)}=\Phi_{(*)}^{(1)}}
\\
\varphi^{(c)} 
&=\de\mu_n\,\tilde S(\de\mu_n)\,\Phi^{(3)}
   -\tilde S(\de\mu_n)\,(\bbbs\cV_n)'(\Phi,\Phi_*,\Phi)
         \Big|^{\Phi_{(*)}=[\bbbone+ \de\mu_n \tilde S(\de\mu_n)^{(*)}]
                                                     \Phi_{(*)}^{(1)}}
          _{\Phi_{(*)}=\Phi_{(*)}^{(1)}}
\\
\varphi_*^{(l)}(\de\cV) 
&=-\tilde S(\de\mu_n)^*\,\de\cV'_*\big(
               \Phi_*^{(1)}+\De\Phi_*^{(1)},
               \Phi^{(1)}+\De\Phi^{(1)},
               \Phi_*^{(1)}+\De\Phi_*^{(1)}\big)
\\
\varphi^{(l)}(\de\cV) &=
   -\tilde S(\de\mu_n)\,\de\cV'\big(
              \Phi^{(1)}+\De\Phi^{(1)},
              \Phi_*^{(1)}+\De\Phi_*^{(1)},
              \Phi^{(1)}+\De\Phi^{(1)}\big)
\end{align*}

\item
and 
$\,\Phi_{(*)}^{(\ge 5)},\,\De\Phi_{(*)}^{(\ge 5)}\,$ being of degree at least five in the fields $\psi_*,\psi$.
\end{itemize}
Observe that $\,\Phi_{(*)}^{(1)}\,$, $\,\Phi_{(*)}^{(3)}\,$, $\,\Phi_{(*)}^{(\ge5)}\,$,
 $\,\De\Phi_{(*)}^{(1)}\,$ and  $\,\varphi_{(*)}^{(c)} \,$
 are independent of $\de\cV$ and that
$\,\varphi_{(*)}^{(l)}(\de\cV) \,$ 
is linear in  $\de\cV$. By 
\cite[Propositions \propBGEphivepssoln\ and \propBGEdephidemu]{BGE} and
\cite{POA}, there is a constant $\GGa_\Phi$, depending only on $\Gam_\op$ 
and $\GGa_\bg$, such that
\begin{equation}\label{eqnRENPhiDeltaPhiestimates}
\begin{split}
\lutn \Phi_{(*)}^{(1)}\rutn\ ,\ \lutn D_{n+1}^{(*)} \Phi_{(*)}^{(1)}\rutn  
            &\le \GGa_\Phi \bar\ka \qquad
\lutn \Phi_{(*)}^{( 3)}\rutn
            \le \GGa_\Phi \sfrac{\fv_{n}}{L} \bar\ka^3
 \\
\lutn \De\Phi_{(*)}^{(1)}\rutn\ ,\ \lutn D_{n+1}^{(*)} \De \Phi_{(*)}^{(1)}\rutn 
            &\le \GGa_\Phi |\de\mu_n| \bar\ka 
\end{split}
\end{equation}
and
\begin{equation}\label{eqnRENnormsvarphi}
\lutn\varphi_{(*)}^{(c)} \rutn 
        \le \GGa_\Phi\, |\de\mu_n |  \sfrac{\fv_{n}}{L} \,\bar\ka^3 
\qquad\qquad
\lutn\varphi_{(*)}^{(l)}(\de\cV) \rutn \le \GGa_\Phi\,\|\de V\|_{2m}\ \bar\ka^3
\end{equation}
By inspection
\begin{align*}
&\cL_4 \big(A^\var(\psi_*,\psi,\de\mu_n,\de\cV)\big) 
\\
&\hskip0.2in
= -\big<\big(\psi_*-Q_{n+1}\Phi^{(1)}_*-Q_{n+1}\De\Phi^{(1)}_*\big)\ ,\ 
                      \fQ_{n+1} Q_{n+1}\De\Phi^{(3)} \big>_0 
\\ &\hskip0.8in
   -\big<\De\Phi_*^{(3)}\ ,\ Q_{n+1}^*\fQ_{n+1}
         \big(\psi -Q_{n+1}\Phi^{(1)}-Q_{n+1}\De\Phi^{(1)}\big)\big>_{n+1}
\\ &\hskip0.8in
+\big<Q_{n+1}\De\Phi^{(1)}_*\ ,\ \fQ_{n+1} Q_{n+1}\Phi^{(3)} \big>_0
+\big<\Phi_*^{(3)}\ ,\ Q_{n+1}^*\fQ_{n+1}Q_{n+1}\De\Phi^{(1)}\big>_{n+1} 
\\ &\hskip0.4in
 + \bbbs\cV_n(\phi_*,\phi)  \Big|
    ^{\phi_{(*)}=\Phi_{(*)}^{(1)}+\De\Phi_{(*)}^{(1)}}   
    _{\phi_{(*)}=\Phi_{(*)}^{(1)}}
\\ &\hskip0.4in
- L^2\mu_n  \Big( 
    \big< \Phi_*^{(3)}+\De\Phi_*^{(3)}\ , \ \Phi^{(1)}+\De\Phi^{(1)}\big>_{n+1}
    -\big< \Phi_*^{(3)}\ , \ \Phi^{(1)}\big>_{n+1}\\ &\hskip0.8in
 +  \big< \Phi_*^{(1)}+\De\Phi_*^{(1)}\ , \ \Phi^{(3)}+\De\Phi^{(3)}\big>_{n+1} 
    -\big< \Phi_*^{(1)}\ ,\ \Phi^{(3)}\big>_{n+1}   \Big)
\\ &\hskip0.4in
-\de\mu_n  \Big( 
   \big< \Phi_*^{(3)}+\De\Phi_*^{(3)}\ , \ \Phi^{(1)}+\De\Phi^{(1)}\big>_{n+1}  
+  \big< \Phi_*^{(1)}+\De\Phi_*^{(1)}\ , \ \Phi^{(3)}+\De\Phi^{(3)}\big>_{n+1}  
           \Big)
\\  
 &\hskip0.4in 
+\big< \De\Phi_*^{(3)}\ ,\ D_{n+1}\Phi^{(1)}\big>_{n+1}
+\big< \Phi_*^{(3)} + \De\Phi_*^{(3)} \ ,\ D_{n+1}\De\Phi^{(1)} \big>_{n+1}
         \\  &\hskip0.8in
+\big<  D_{n+1}^*\Phi_*^{(1)}\ ,\  \De\Phi^{(3)} \big>_{n+1}
+ \big< \De\Phi_*^{(1)}\ ,\ D_{n+1}(\Phi^{(3)}+\De\Phi^{(3)})\big>_{n+1}  
\end{align*}
By  \cite[Proposition \propBGEphivepssoln.c]{BGE},
\begin{align*}
D_{n+1}^{(*)}\Phi^{(1)}_{(*)}
&= Q_{n+1}^*\fQ_{n+1}\psi_{(*)} -(Q_{n+1}^*\fQ_{n+1}Q_{n+1}-L^2\mu_n) \bbbB _{(*)}\psi_*
\\
&=Q_{n+1}^*\fQ_{n+1}\big(\psi_{(*)}-Q_{n+1}\Phi^{(1)}_{(*)}\big)+L^2\mu_n\Phi^{(1)}_{(*)}
\end{align*}
This leads to a cancellation between lines 1,2,5,6 and the last 
two lines in the formula for $\,\cL_4 \big(A^\var)\,$. 
Inserting the decomposition
$\,
\De\Phi_{(*)}^{(3)} =\varphi_{(*)}^{(c)} +\varphi_{(*)}^{(l)}(\de\cV)
\, $
we get
\begin{equation*}
\cL_4 \big(A^\var(\psi_*,\psi,\de\mu_n,\de\cV)\big) 
= A_{\de\cV}(\psi_*,\psi,\de\cV)-B(\psi_*,\psi)
\end{equation*}
with
\begin{align*}
A_{\de\cV}(\psi_*,\psi,\de\cV)
&= \big<\De\Phi^{(1)}_*,
        Q_{n+1}^*\fQ_{n+1} Q_{n+1}\varphi^{(l)}(\de\cV) \big>_0 
\\&\hskip0.2in   +\big<\varphi_*^{(l)}(\de\cV),  
            Q_{n+1}^*\fQ_{n+1}Q_{n+1}\De\Phi^{(1)}\big>_0
\\ &\hskip0.2in
- L^2\mu_n \Big( \big< \varphi_*^{(l)}(\de\cV), \ \De\Phi^{(1)}\big>_{n+1}
+  \big< \De\Phi_*^{(1)}, \ \varphi^{(l)}(\de\cV)\big>_{n+1} \Big)
\\ &\hskip0.2in
-\de\mu_n \Big(  \big<\varphi_*^{(l)}(\de\cV), \ \Phi^{(1)}+\De\Phi^{(1)}\big>_{n+1}\! 
+\!\big< \Phi_*^{(1)}+\De\Phi_*^{(1)},\ \varphi^{(l)}(\de\cV)\big>_{n+1}  \Big)
\\  &\hskip0.2in
+\big< \varphi_*^{(l)}(\de\cV), D_{n+1}\De\Phi^{(1)} \big>_{n+1}
+ \big< D_{n+1}^* \De\Phi_*^{(1)},\varphi^{(l)}(\de\cV)\big>_{n+1}  
\end{align*}
linear in $\de\cV$ and
\begin{align*}
B(\psi_*,\psi)
&=-\big<\De\Phi^{(1)}_*\,,\,
        Q_{n+1}^*\fQ_{n+1} Q_{n+1}\varphi^{(c)} \big>_0 
   -\big<\varphi_*^{(c)}\,,\,  
            Q_{n+1}^*\fQ_{n+1}Q_{n+1}\De\Phi^{(1)}\big>_0
\\ &\hskip0.3in
-\big<Q_{n+1}\De\Phi^{(1)}_*\ ,\ \fQ_{n+1} Q_{n+1}\Phi^{(3)} \big>_0
-\big<\Phi_*^{(3)}\ ,\ Q_{n+1}^*\fQ_{n+1}Q_{n+1}\De\Phi^{(1)}\big>_{n+1} 
\\ &\hskip0.3in
-\bbbs\cV_n(\phi_*,\phi)  \Big|
    ^{\phi_{(*)}=\Phi_{(*)}^{(1)}+\De\Phi_{(*)}^{(1)}}   
    _{\phi_{(*)}=\Phi_{(*)}^{(1)}}
\\ &\hskip0.3in
+ L^2\mu_n \Big(\big< \Phi_*^{(3)}+\varphi_*^{(c)}, \ 
                              \De\Phi^{(1)}\big>_{n+1}
+ \big< \De\Phi_*^{(1)}, \ 
                              \Phi^{(3)}+\varphi^{(c)}\big>_{n+1} \Big)
\\ &\hskip0.3in
+\de\mu_n \Big(\!
   \big< \Phi_*^{(3)}+\varphi_*^{(c)}, \ 
                       \Phi^{(1)}+\De\Phi^{(1)}\big>_{n+1}  
\!+\!\big< \Phi_*^{(1)}+\De\Phi_*^{(1)}, \ 
                \Phi^{(3)}+\varphi^{(c)}\big>_{n+1}  \Big)
\\  &\hskip0.3in
-\big< \Phi_*^{(3)}+\varphi_*^{(c)}, D_{n+1}\De\Phi^{(1)} \big>_{n+1}
- \big<  D_{n+1}^* \De\Phi_*^{(1)}, (\Phi^{(3)}+\varphi^{(c)})\big>_{n+1}
\end{align*}
independent of $\de\cV$. By \eqref{eqnRENPhiDeltaPhiestimates} and
\eqref{eqnRENnormsvarphi}, there is a constant $\cc_1$ such that
\begin{equation*}
\lun A_{\de\cV} \run
   \le \cc_1 |\de\mu_n|\,\|\de V\|_{2m}\,\bar\ka^4
\qquad\qquad   
\lun B\run
   \le \cc_1 |\de\mu_n|\sfrac{\fv_{n}}{L}\,\bar\ka^4
\end{equation*}
Therefore,  by Lemma \ref{lemRENpsitophi}, \eqref{eqnRENfluctParts} and the estimate
on $\de\mu_n$ in Lemma \ref{lemRENrenormchem}
\begin{align*}
\| \Om_4(A_{\de\cV}) \|_{2m}& 
         \le  \cc_2  \|\de V\|_{2m}\, \sfrac{\fe_\fl(n)}{\bar\ka^2}
\qquad \qquad
\|\Om_4(B)\|_{2m} 
      \le \cc_2  \sfrac{\fv_{n}}{L}\sfrac{\fe_\fl(n)}{\bar\ka^2}
\\
\|\Om_4( \cL_4(\tilde\cE_\fl) ))\|_{2m} 
          & \le  \cc_2 \sfrac{\fe_\fl(n)}{\bar\ka^4}
\end{align*}
Assuming that $\fv_0$ is small enough,  the linear operator 
$\,
\de\cV \mapsto \Om_4(A_{\de\cV}) 
\,$
has operator norm at most $\half$ with respect to the norm $\,\|\,\cdot\,\|_{2m}\,$.  Therefore the operator
\begin{equation*}
\de\cV \mapsto \de\cV+\Om_4(A_{\de\cV}) 
\end{equation*}
has an inverse $I_{\de V}$ whose operator norm is bounded by $2$.
Set
\begin{equation*}
\de\cV_n =  I_{\de V} \Big( \Om_4\big(B - \cL_4(\tilde\cE_\fl) \big)\Big)
\end{equation*}
By \cite[(\eqnPARestrad.b)]{PAR1},
\begin{equation*}
\|\de\cV_n\|_{2m}
 \le  2\,\cc_2\sfrac{\fe_\fl(n)}{\bar\ka^2}
\big( \sfrac{\fv_n}{L} + \sfrac{1}{\bar\ka^2} \big)
 \le  2(1+2\rrho_\bg)\,\cc_2\sfrac{\fe_\fl(n)}{\bar\ka^4}
\end{equation*}

\Item (b)
By \cite[(\eqnSAkernNorm), Remark \remHTpreciseinduction]{PAR1} and part (a),
\begin{align*}
\|\cV_{n+1}-\cV^{(u)}_{n+1}\|_{2m}
&\le \big\|\bbbs\big(\cV_n-\cV^{(u)}_n\big)\big\|_{2m} +\|\de\cV_n\|_{2m} \\
&\le \sfrac{1}{L}\big\|\cV_n-\cV^{(u)}_n\big\|_{2m} +\|\de\cV_n\|_{2m} \\
&\le \sfrac{\CC_{\de\cV}}{L^{n+1}} 
    \smsum_{\ell=1}^n \sfrac{L^\ell }{\ka(\ell)^4}\fe_\fl(\ell-1)  
    +\cc_{\de\cV}\sfrac{\fe_\fl(n)}{\ka(n+1)^4}  
\end{align*}
\end{proof}
 
\begin{lemma}[Garbage Collection from $\cR$]\label{lemRENrngarbage}
There is a constant $\cc_\gar$, depending only on $\Gam_\op$,  
$\GGa_\bg$ and $\rrho_\bg$, such that the following holds. There are
\begin{align*} 
&\cP_\cR=\sum_{\vp\in\fD}\cP_\cR^\vp\qquad\text{with}\quad
   \cP_\cR^\vp\in\fP_\vp\text{ for each }\vp\in\fD \\
&\cI_\cR\in\fP_\irr
\end{align*}
such that
\begin{equation*}
(\bbbs \cR_n)(\phi_*,\phi)\Big|
     ^{\phi_{(*)}=\phi_{(*)n+1}(\psi_*,\psi,\mu_{n+1},\cV_{n+1})}
     _{\phi_{(*)}=\phi_{(*)n+1}(\psi_*,\psi,L^2\mu_n,\bbbs\cV_n)}
=\Big[\cP_R\big(\tilde\psi_*,\tilde\psi\big)
     +\cI_R\big(\tilde\psi_*,\tilde\psi\big)
    \Big]_{\tilde\psi_{(*)}=(\psi_{(*)},\{\partial_\nu\psi_{(*)}\})}
\end{equation*}
Furthermore
\begin{align*}
\lln \cP_\cR^\vp\rln
&\le \cc_\gar|\de\mu_n| \big\|\bbbs\tilde\cR_n^\vp\big\|_m 
           \bar\ka^\vp 
\\
\lln \cI_\cR\rln 
&\le\ \cc_\gar\,\big(|\de\mu_n| \fv_{n+1}+\|\de\cV_n\|_{2m})\,
                              \bar\ka^2\,\fr(n,\CC_\cR)
\end{align*}
where  
\begin{equation*}
\fr(n,C)=\sum_{\vp\in\fD}
                     L^{5-\De(\vp)}\,\bar\ka^\vp\,\fr_\vp(n,C)
\end{equation*}
\end{lemma}

\begin{proof}  Similar to \eqref{eqndiffphiwithmun} we write
\begin{align*}
\Phi_{(*)}&=\phi_{(*)n+1}(\psi_*,\psi,L^2\mu_n,\bbbs\cV_n)\\
\De\Phi_{(*)}&=\phi_{(*)n+1}(\psi_*,\psi,\mu_{n+1},\cV_{n+1})
                 -\phi_{(*)n+1}(\psi_*,\psi,L^2\mu_n,\bbbs\cV_n)
\end{align*}
and, for each $\vp\in\fD$,
\begin{align*}
\cR_\var^\vp(\tilde\psi_*,\tilde\psi)
=\big(\bbbs\tilde\cR_n^\vp\big)(\tilde\phi_*,\tilde\phi)\Big| 
      ^{\phi_{(*)}=\Phi_{(*)}+\De\Phi_{(*)},\ \ 
        \phi_{(*)\nu}=\partial_\nu\Phi_{(*)}+\partial_\nu\De\Phi_{(*)}}
     _{\phi_{(*)}=\Phi_{(*)},\ \ 
        \phi_{(*)\nu}=\partial_\nu\Phi_{(*)}}
\end{align*}
As in the proof of Lemma \ref{lemRENreninteraction}  we decompose
\begin{align*}
 \Phi_{(*)}(\psi_*,\psi)
     &=\Phi_{(*)}^{(1)}(\psi_*,\psi) + \Phi_{(*)}^{(\ge 3)}(\psi_*,\psi)\\
 \De\Phi_{(*)}(\psi_*,\psi)
     &=\De\Phi_{(*)}^{(1)}(\psi_*,\psi) 
    + \De\Phi_{(*)}^{(\ge 3)}(\psi_*,\psi)\\
 \partial_\nu\Phi_{(*)}(\psi_*,\psi)
     &=\partial_\nu\Phi_{(*)}^{(1)}(\psi_*,\psi) 
    + \partial_\nu\Phi_{(*)}^{(\ge 3)}(\psi_*,\psi)\\
 \partial_\nu\De\Phi_{(*)}(\psi_*,\psi)
     &=\partial_\nu\De\Phi_{(*)}^{(1)}(\psi_*,\psi) 
    + \partial_\nu\De\Phi_{(*)}^{(\ge 3)}(\psi_*,\psi)
\end{align*}
where the superscript ``$(1)$'' signifies the part that is of degree precisely
one in $\psi_{(*)(\nu)}$ and the superscript ``$(\ge 3)$'' signifies the part 
that is of degree at least three in $\psi_{(*)(\nu)}$.
By \cite[Propositions \propBGEphivepssoln\ and \propBGEdephidemu]{BGE},
\begin{equation}\label{eqnRENPhiDeltaPhiestimatesB}
\begin{aligned}
\luTN \Phi_{(*)}^{(1)}\ruTN 
            &\le \GGa_\Phi \bar\ka &
\luTN \Phi_{(*)}^{(\ge 3)}\ruTN
            &\le \GGa_\Phi\fv_{n+1} \bar\ka^3 \\
\luTN \De\Phi_{(*)}^{(1)}\ruTN
            &\le \GGa_\Phi |\de\mu_n| \bar\ka &
\luTN \De\Phi_{(*)}^{(\ge 3)}\ruTN
  &\le \GGa_\Phi \big(|\de\mu_n| \fv_{n+1}+\|\de\cV_n\|_{2m})\bar\ka^3 \\
\luTN \partial_\nu\Phi_{(*)}^{(1)}\ruTN 
            &\le \GGa_\Phi \bar\ka' &
\luTN \partial_\nu\Phi_{(*)}^{(\ge 3)}\ruTN
       &\le \GGa_\Phi\fv_{n+1} \bar\ka^2\bar\ka' \\
\luTN \partial_\nu\De\Phi_{(*)}^{(1)}\ruTN
            &\le \GGa_\Phi |\de\mu_n| \bar\ka'\quad &
\luTN \partial_\nu\De\Phi_{(*)}^{(\ge 3)}\ruTN
            &\le \GGa_\Phi \big(|\de\mu_n| \fv_{n+1}+\|\de\cV_n\|_{2m})
                       \bar\ka^2\bar\ka' 
\end{aligned}
\end{equation}
We correspondingly decompose
\begin{align*}
\cR_\var^\vp(\tilde\psi_*,\tilde\psi)
=\cR_\lo^\vp(\tilde\psi_*,\tilde\psi)
 + \cR_\ho^\vp(\tilde\psi_*,\tilde\psi)
\end{align*}
where
\begin{align*}
\cR_\lo^\vp(\tilde\psi_*,\tilde\psi)
&=\big(\bbbs\tilde\cR_n^\vp\big)(\tilde\phi_*,\tilde\phi)\Big| 
      ^{\phi_{(*)}=\Phi_{(*)}^{(1)}+\De\Phi_{(*)}^{(1)},\ \ 
     \phi_{(*)\nu}=\partial_\nu\Phi_{(*)}^{(1)}+\partial_\nu\De\Phi_{(*)}^{(1)}}
     _{\phi_{(*)}=\Phi_{(*)}^{(1)},\ \ 
        \phi_{(*)\nu}=\partial_\nu\Phi_{(*)}^{(1)}}
\end{align*}
Clearly $\cR_\lo^\vp\in\fP_\vp$ and
\begin{align*}
\lln \cR_\lo^\vp \rln 
 &\le \cc_3\,\|\bbbs\tilde\cR_n^\vp\|_m   
                             \,|\de\mu_n|\,\bar\ka^\vp\\
% \le \cc_3\,\|\tilde\cR_n^\vp\|_m   \,|\de\mu_n|\,
%                            L^{5-\De(\vp)} \bar\ka^\vp
\lln \cR_\ho^\vp \rln 
 &\le \cc_3\,\|\bbbs\tilde\cR_n^\vp\|_m   
     \,\big(|\de\mu_n| \fv_{n+1}+\|\de\cV_n\|_{2m})\,\bar\ka^2\ \bar\ka^\vp
\end{align*}
Set
\begin{equation*}
\cP_\cR^\vp=\cR^\vp_\lo 
\qquad\qquad
\cI_\cR=\sum_{\vp\in\fD}\cR^\vp_\ho 
\end{equation*}
The estimates follow by Lemma \ref{lemOSFmainlem}.c and
the bound $\big\|\tilde\cR_n^{(\vp)}\big\|_m \le \fr_\vp(n,\CC_\cR)$ 
of \cite[Remark \remHTpreciseinduction]{PAR1}.
\end{proof}

\begin{lemma}\label{lemRENcRcE}
There exist 
\begin{itemize}[leftmargin=*, topsep=2pt, itemsep=2pt, parsep=0pt]
\item 
a polynomial 
$\,\tilde\cR_{n+1}(\tilde\phi_*,\tilde\phi)=\sum_{\vp\in\fD}
               \tilde\cR_{n+1}^{(\vp)}(\tilde\phi_*,\tilde\phi)\,$
on  $\,\tilde\cH_{n+1}^{(0)}\times \tilde\cH_{n+1}^{(0)}\,$, with
each $\tilde\cR_{n+1}^{(\vp)}$ being an $\fS$ invariant polynomial of 
type $\vp$, and
\item
an $\fS$ invariant analytic function  
$\,
\tilde\cE_{n+1,2}(\tilde\psi_*,\tilde\psi)
\,$
on a neighbourhood of the origin in 
$\tilde\cH_0^{(n+1)}\times \tilde\cH_0^{(n+1)}$ with
$\tilde\cE_{n+1,2}(0,0)=0$, whose power series expansion does not
contain scaling/weight relevant monomials
\end{itemize}
such that 
\begin{equation}\label{eqnRENrearrangeBis}
\begin{split}
&\Big[- A_{n+1}(\psi_*,\psi,\phi_*,\phi,L^2\mu_n,\bbbs \cV_n)
+(\bbbs \cR_n)(\phi_*,\phi)\Big]_
{\phi_{(*)}=\phi_{(*)n+1}(\psi_*,\psi,L^2\mu_n,\bbbs\cV_n)}\\\noalign{\vskip-0.05in}
&\hskip2.5in
+\tilde\cE_\fl(\tilde\psi_*,\tilde\psi)\Big|_
    {\tilde\psi_{(*)}=(\psi_{(*)},\{\partial_\nu\psi_{(*)}\})}\\
&=\Big[- A_{n+1}(\psi_*,\psi,\phi_*,\phi,\mu_{n+1},\cV_{n+1})
         +\cR_{n+1}(\phi_*,\phi)\Big]_
{\phi_{(*)}=\phi_{(*)n+1}(\psi_*,\psi,\mu_{n+1},\cV_{n+1})}\\   
     \noalign{\vskip-0.05in}
&\hskip2.5in
+\tilde\cE_{n+1,2}(\tilde\psi_*,\tilde\psi)\Big|_
    {\tilde\psi_{(*)}=(\psi_{(*)},\{\partial_\nu\psi_{(*)}\})}
\end{split}
\end{equation}
where
\begin{equation*}
\cR_{n+1}(\phi_*,\phi )= \tilde\cR_{n+1}\big((\phi_*,\{\partial_\nu\phi_*\}),(\phi,\{\partial_\nu\phi\})
          \big)
\end{equation*}
Furthermore
\ 
\begin{enumerate}[label=(\alph*), leftmargin=*]
\item 
there exists a constant $\CC_\cR$, depending only on 
$\Gam_\op$, $\GGa_\bg$, $\rrho_\bg$ and $m$,
such that if \cite[(\eqnHTinductiveRestimates)]{PAR1} holds for $n$, then 
\begin{align*}
\big\|\tilde\cR_{n+1}^{(\vp)}\big\|_m
 &\le \fr_\vp(n+1,\CC_\cR)
\end{align*}
\item 
there exists a constant $\CC_\ren$, depending only on 
$\Gam_\op$, $\GGa_\bg$, $\rrho_\bg$ and $m$, such that
\begin{equation*}
\lln\tilde\cE_{n+1,2}\rln 
\le \CC_\ren\,\fe_\fl(n)
\end{equation*}
\end{enumerate}
\end{lemma}
\noindent
Part (a) provides our choice for the $\CC_\cR$ of 
\cite[Remark \remHTpreciseinduction]{PAR1}.

\begin{proof}
Set
\begin{align*}
\tilde A^{\var}(\psi_*,\psi)  
&=A_{n+1}(\psi_*,\psi,\phi_*,\phi,\mu_{n+1},\cV_{n+1})\Big|_
       {\phi_{(*)}=\phi_{(*)n+1}(\psi_*,\psi,\mu_{n+1},\cV_{n+1})}\\
&\hskip 3cm
- A_{n+1}(\psi_*,\psi,\phi_*,\phi,L^2\mu_n, \bbbs\cV_n)\Big|_
      {\phi_{(*)}=\phi_{(*)n+1}(\psi_*,\psi,L^2\mu_n, \bbbs\cV_n)}\\
&=A^{\var}(\psi_*,\psi,\de\mu_n,\de\cV_n) 
+ \de\cV_n(\phi_*,\phi)\Big|_{
             \phi_{(*)}=\phi_{(*)n+1}(\psi_*,\psi,\mu_{n+1},\cV_{n+1})}
\end{align*}
and
$$
\cP_A=\cL_\fD(\tilde A^\var)
\qquad\qquad
\cI_A=\cI(\tilde A^\var)
$$
By Corollary \ref{corLOCproj}.a and Lemmas \ref{lemRENrenormchem},
\ref{lemRENpsitophi} and \ref{lemRENreninteraction},
\begin{equation}\label{eqnRENgoodproptildeAvar}
\begin{split}
\tilde A^\var(\psi_*,\psi)
&=-\ell(\tilde\cE_\fl)\hskip-2.5pt\int\hskip-3pt dx \ \psi_*(x) \psi(x)
- \cL_4(\tilde\cE_\fl) \\
&\hskip1in+\Big[\cP_A\big(\tilde\psi_*,\tilde\psi\big)
     +\cI_A\big(\tilde\psi_*,\tilde\psi\big)
    \Big]_{\tilde\psi_{(*)}=(\psi_{(*)},\{\partial_\nu\psi_{(*)}\})}
\end{split}
\end{equation}

By \eqref{eqnRENavarbndsa} and \cite[Propositions \propBGEphivepssoln\ and \propBGEdephidemu]{BGE}
\begin{equation}\label{eqnRENavarbnds}
\begin{split}
\lun A^\var_2 \run 
   &\le \cc_A\big[ 1  +L^2|\mu_n|\big]|\de\mu_n|\bar\ka^2\\
\lun \tilde A^\var -  A_2^\var -  \cL_4( \tilde A^\var) \run 
   &\le \cc_A (\fv_{n+1}|\de\mu_n|+\|\de\cV_n\|_{2m})  \bar\ka^4\\
%\lun A^\var_\ho \run 
%   &\le \cc_A (\fv_{n+1}\fv_{n+1}\fv_{n+1}\fv_{n+1}+\|\de\cV\|_{2m}) |\de\mu|\bar\ka^4
\end{split}
\end{equation}
To prove the second bound, use \eqref{eqnRENPhiDeltaPhiestimates}, 
\eqref{eqnRENPhiDeltaPhiestimatesB} and
\begin{equation*}
\luTN D_{n+1}^{(*)}\Phi_{(*)}^{(\ge 3)}\ruTN
            \le \GGa_\Phi\fv_{n+1} \bar\ka^3 \qquad
\luTN D_{n+1}^{(*)} \De\Phi_{(*)}^{(\ge 3)}\ruTN 
     \le \GGa_5\ \big(|\de\mu_n|\fv_{n+1}+\|\de\cV_n\|_{2m}\big)\bar\ka^3
\end{equation*}
and also the observation that if one substitutes 
\begin{equation*}
\phi_{(*)n+1}(\psi_*,\psi,\mu_{n+1},\cV_{n+1})=
\Phi_{(*)}^{(1)}+\Phi_{(*)}^{(\ge 3)}
                    +\De\Phi_{(*)}^{(1)}+\De\Phi_{(*)}^{(\ge 3)}
\end{equation*}
and 
$\phi_{(*)n+1}(\psi_*,\psi,L^2\mu_n,\bbbs\cV_n)=
\Phi_{(*)}^{(1)}+\Phi_{(*)}^{(\ge 3)}$
into 
$\tilde A^\var -  A_2^\var -  \cL_4( \tilde A^\var)$
and expands out, then
\begin{itemize}[leftmargin=*, topsep=2pt, itemsep=2pt, parsep=0pt]
\item
every surviving term must contain at least one
$\Phi_{(*)}^{(\ge 3)}$ or $\De\Phi_{(*)}^{(\ge 3)}$
and 
\item
every surviving term, except for those coming from 
$\de\mu_n\< \phi_*,\phi\>_{n+1}$
and $\de\cV_n(\phi_*,\phi)$, must contain at least one
$\De\Phi_{(*)}^{(1)}$ or $\De\Phi_{(*)}^{(\ge 3)}$
\end{itemize}

\noindent 
So, if we write
 $\cP_A=\sum_{\vp\in\fD}\cP_A^\vp$ with $\cP_A^\vp\in\fP_\vp$,
then, by Proposition \ref{propLOCproj},
\begin{equation}\label{eqnRENestPAIA}
\begin{split}
\lln \cP_A^\vp\rln &\le  18\cc_\loc\cc_A|\de\mu_n|\bar\ka^\vp 
\qquad  \text{if }\vp \ne (6,0,0)
 \\
 \lln \cP_A^{(6,0,0)}\rln \,, \ \lln \cI_A\rln
 &\le \cc_A (\fv_{n+1}|\de\mu_n|+\|\de\cV_n\|_{2m})  \bar\ka^4
\end{split}
\end{equation}
Set
\begin{align*}
\tilde\cR_{n+1}^{(\vp)}
&=\bbbs \tilde\cR_n^{(\vp)}
  +\Om_\vp\big(\cL_\fD(\tilde\cE_\fl) +\cP_A-\cP_\cR\big)\qquad\text{for $\vp\in\fD$}\\
\tilde\cE_{n+1,2}
&=\cI(\tilde\cE_\fl)+\cI_A-\cI_\cR+\Om_\irr\big(\cL_\fD(\tilde\cE_\fl) +\cP_A-\cP_\cR,\de\cV_n\big)
\end{align*}
Here $\Om_\vp(\cP)$ is the part of $\Om(\cP)$ in $\fR_\vp$.
The identity \eqref{eqnRENrearrangeBis} now follows from \eqref{eqnRENfluctParts}, 
\eqref{eqnRENdefAvar}, \eqref{eqnRENgoodproptildeAvar} and 
Lemmas \ref{lemRENpsitophi} and \ref{lemRENrngarbage}.

Write
$
\ \cL_\fD(\tilde\cE_\fl) =\sum_{\vp\in\fD} \cP_\cE^\vp\ 
$
with each $\cP_\cE^\vp\in\fP^\vp$. 
By \eqref{eqnRENfluctParts}, \eqref{eqnRENestPAIA}, Lemma \ref{lemRENrngarbage},
and the estimate $\,|\de\mu_n|\le\sfrac{3}{\bar\ka^2}\fe_\fl(n)\,$ 
of  Lemma \ref{lemRENrenormchem},
\begin{align*}
\llN\cP_\cE^\vp+\cP^\vp_A-\cP^\vp_\cR\rlN
&\le \Big[1+\cc_4\sfrac{\bar\ka'}{\bar\ka}\Big]^2
        \,\fe_\fl(n)
 + 3\,\cc_\gar \sfrac{\bar\ka^\vp}{\bar\ka^2}
       \|\bbbs\tilde\cR_n^\vp\|_m\fe_\fl(n)
\\
& \hskip 1cm + \cc_A\de_{\vp,\,(6,0,0)}
 (3\fv_{n+1}\bar\ka^2 \fe_\fl(n)+\|\de\cV_n\|_{2m} \bar\ka^4)  
\\
%%%
&\le \Big\{  \Big[1+\cc_5\big( \sfrac{\bar\ka'}{\bar\ka}
               +\fv_{n+1}\bar\ka^2\big)\Big]^2
  + \cc_5  \sfrac{\bar\ka^\vp}{\bar\ka^2} \|\bbbs\tilde\cR_n^\vp\|_m 
     \Big\}\,\fe_\fl(n)
%& \hskip 1cm + \cst{''}{A}\de_{\vp,\,(6,0,0)}
%\big(1+ \fv_{n+1}\bar\ka^2\big)\,\fe_\fl(n)
\end{align*}
with a constant $\cc_5$ depends only on $\cc_4$, $\cc_\gar$, $\cc_A$ 
and $\cc_{\de\cV}$.
For the second inequality we used Lemma \ref{lemRENreninteraction}.
Using  Lemma \ref{lemOSFmainlem}.c and 
the bound $\big\|\tilde\cR_n^{(\vp)}\big\|_m \le \fr_\vp(n,\CC_\cR)$ 
of \cite[(\eqnHTinductiveRestimates)]{PAR1},
\begin{equation}\label{eqnRENmessA}
\llN\cL_\fD(\tilde\cE_\fl) +\cP_A-\cP_\cR\rlN
\le \Big\{ 4\Big[1+\cc_5\big( \sfrac{\bar\ka'}{\bar\ka}
               +\fv_{n+1}\bar\ka^2\big)\Big]^2
         + \sfrac{\cc_5}{\bar\ka^2}\fr(n,\CC_\cR)
\Big\}\fe_\fl(n)
\end{equation}

\Item (a) 
By Lemmas \ref{lemRENpsitophi} and \ref{lemOSFmainlem}.c
and \cite[(\eqnHTinductiveRestimates)]{PAR1},
\begin{align*}
&\big\|\tilde\cR_{n+1}^{(\vp)}\big\|_m
\le \big\|\bbbs\tilde\cR_n^{(\vp)}\big\|_m
     + \sfrac{\cc_\Om}{\bar\ka^\vp}
            \llN\cP_\cE^\vp+\cP^\vp_A-\cP^\vp_\cR\rlN\\
&\hskip0.2in\le \big(1
   +\sfrac{\cc_\Om \cc_5}{\bar\ka^2}\fe_\fl(n)\big)
   \big\|\bbbs \tilde\cR_n^{(\vp)}\big\|_m
   + \sfrac{\cc_\Om}{\bar\ka^\vp}
   \Big[1+\cc_5\big( \sfrac{\bar\ka'}{\bar\ka}     
               +\fv_{n+1}\bar\ka^2\big)\Big]^2
        \,\fe_\fl(n)\\
&\hskip0.2in\le \big(1
   +\sfrac{\cc_\Om \cc_5}{\bar\ka^2}\fe_\fl(n)\big)
       L^{5-\De(\vp)}\fr_\vp(n,\CC_\cR)
   + \sfrac{\cc_\Om}{\bar\ka^\vp}
   \Big[1+\cc_5\big( \sfrac{\bar\ka'}{\bar\ka}     
               +\fv_{n+1}\bar\ka^2\big)\Big]^2
        \,\fe_\fl(n)
\end{align*}
If $\CC_\cR$ is large enough, depending only on 
$\cc_\Om$, $\cc_5$ and $\rrho_\bg$, then
by \cite[(\eqnPARformforrpnC)]{PAR1},
\begin{align*}
&\big\|\tilde\cR_{n+1}^{(\vp)}\big\|_m\le \big(1
   +\sfrac{\CC_\cR}{\bar\ka^2}\fe_\fl(n)\big)
   L^{5-\De(\vp)}\fr_\vp(n,\CC_\cR)
   + \CC_\cR\sfrac{\fe_\fl(n)}{\ka^\vp(n+1)}\\
&\hskip0.2in\le 
   \fr_\vp(0)  L^{(5-\De(\vp))(n+1)} \,\Pi_0^{n+1}(\CC_\cR)
+\CC_\cR \sum_{\ell=1}^{n+1} L^{(5-\De(\vp))(n+1-\ell)}
    \sfrac{ \fe_\fl(\ell-1) }{\ka^\vp(\ell)} \,\Pi_\ell^{n+1}(\CC_\cR)
   \\
&\hskip0.2in= \fr_\vp(n+1,\CC_\cR)
\end{align*}

\Item (b) We have, by \eqref{eqnRENfluctParts}, \eqref{eqnRENestPAIA}, 
Lemmas \ref{lemRENrngarbage} and \ref{lemRENpsitophi} and \eqref{eqnRENmessA},
\begin{align*}
\lln\tilde\cE_{n+1,2}\rln
&\le \Big[1+18\cc_\loc\sfrac{\bar\ka'}{\bar\ka} \Big]^2\,\fe_\fl(n)
    + \cc_A\,(\fv_{n+1}|\de\mu_n|+\|\de\cV_n\|_{2m})\bar\ka^4
             \\ &\hskip0.5in
    + \cc_\gar(\fv_{n+1}|\de\mu_n|+\|\de\cV_n\|_{2m})\bar\ka^2\,\fr(n,\CC_\cR)
             \\ &\hskip0.5in
    +\cc_\Om \fv_{n+1}\bar\ka^2 
       \Big\{ 4\Big[1+\cc_5\big( \sfrac{\bar\ka'}{\bar\ka}
               +\fv_{n+1}\bar\ka^2\big)\Big]^2
         + \sfrac{\cc_5}{\bar\ka^2}\fr(n,\CC_\cR)
\Big\}\fe_\fl(n) \\
&\le \half \CC_\ren\Big\{1+(\fv_{n+1}+\sfrac{1}{\bar\ka^2}\big)\fr(n,\CC_\cR)
       \Big\}\,\fe_\fl(n) \\
&\le \CC_\ren\,\fe_\fl(n) 
\end{align*}
by \cite[(\eqnPARestrad.b) and Lemma \lemPARestrpnC.a]{PAR1},
provided $\fv_0$ is small enough that the hypothesis 
\begin{equation*}
\eps |\log\fv_0| \ge 2\log (1+\CC_\cR)\,\Pi_0^\infty(\CC_\cR)
\end{equation*}
of \cite[Lemma \lemPARestrpnC]{PAR1} is satisfied.
\end{proof}

\begin{lemma}[Properties of $\de\mu_n$]\label{lemRENmunppties}
There is a constant $\LLa_{\de\mu}$, depending only on $L$, $\Gam_\op$, 
$\GGa_\bg$, $\rrho_\bg$ and $m$, such that
the following holds. Set $\mu^*_0=\mu_0$ and, for $n\ge 0$, $\de\mu^*_n=\mu^*_{n+1}-L^2\mu^*_n$.
Then, for $n\ge 0$,
\begin{enumerate}[label=(\alph*), leftmargin=*]
\item 
$
\big|\ell\big(\tilde\cE_\fl\big)-\de\mu^*_n\big|
\le  \LLa_{\de\mu}\ \fv_0 \big(\fv_0^{\sfrac{1}{3}-5\eps}
               + |\mu_n|
               + r_n^2e^{-r_n^2}
               \big)\bar\ka_\fl^6
$

\item  $|\de\mu_n-\de\mu^*_n|
\le  \LLa_{\de\mu}\ \fv_0^{1-7\eps}  \big(\fv_0^{\sfrac{1}{3}-5\eps}
               + L^{2n}(\mu_0-\mu_*)
               \big)L^{3\eps(n+1)}
$

\item 
$|\mu_{n+1}-\mu^*_{n+1}|
\le L^{2(n+1)}\,\fv_0^{1-8\eps}
     \smsum_{\ell=1}^{n+1} \sfrac{1}{L^{(2-3\eps)\ell}}
        \big[\fv_0^{\sfrac{1}{3}-5\eps} +L^{2\ell}(\mu_0-\mu_*)\big]$
\end{enumerate}
\end{lemma}
\noindent
The bound on $|\mu_{n+1}-\mu_{n+1}^*|$ in part (c) is
exactly the bound of \cite[Remark \remHTpreciseinduction]{PAR1} with $n$ replaced by
$n+1$.

\begin{proof} (a)
The monomials $M_n(\psi_*,\psi)$ of Proposition \ref{propOSFcEtwo}.b
are translation invariant with respect
to $\cX_0^{(n+1)}$, despite the fact that $C^{(n)}$ is only translation 
invariant with respect to $\cX_{-1}^{(n+1)}$. 
Using Corollary \ref{corLOCproj}.b and Remark \ref{remMXconsteigen}\ 
and using $1_0^{(1)}$ to denote the constant function on $\cX_0^{(1)}$
that always takes the value $1$, and using $M_0(x'_2,x'_3)$ to denote
the kernel of $M_0$,
\begin{align*}
\ell(M_0)
&=\sfrac{1}{|\cX_0^{(1)}|}\int_{\cX_0^{(1)}} dx'_2 dx'_3\ M_0(x'_2,x'_3)\\
&=-\sfrac{2}{L^3|\cX_0^{(1)}|}\int_{\cX_0} dx_1 \cdots dx_4\  
V_0(x_1,x_2,x_3,x_4)\ 
%\\
%&\hskip1.9in  
  (S_1(L^2\mu_0)Q_1^* \fQ_1 1_0^{(1)})(\bbbl^{-1}x_2)\\
\noalign{\vskip-0.05in}
&\hskip1.9in  
  ({S_1(L^2\mu_0)}^*Q_1^* \fQ_1 1_0^{(1)})(\bbbl^{-1}x_3)\ 
%\\
%&\hskip1.9in  
  C^{(0)}(x_4,x_1)\\
&=-\sfrac{2}{L^3|\cX_0^{(1)}|}\big(\sfrac{a_1}{a_1-L^2\mu_0}\big)^2
   \int_{\cX_0} dx_1 \cdots dx_4\  V_0(x_1,x_2,x_3,x_4)\ C^{(0)}(x_4,x_1)\\
&=-\sfrac{2}{L^3|\cX_0^{(1)}|}
   \int_{\cX_0} dx_1 \cdots dx_4\  V_0(x_1,x_2,x_3,x_4)\ C^{(0)}(x_4,x_1)
   +O\big(|\mu_0|\fv_0\big)
\end{align*}
Recalling that
\begin{equation*}
 \quad
S_1=(D_1+Q_1^*\fQ_1 Q_1)^{-1} \quad
\fQ_1=a\bbbone\quad
D_1 = L^2\ \bbbl_*^{-1} D_0 \bbbl_*\quad
Q_1 = \bbbl_*^{-1}\,Q \,\bbbl_*
\end{equation*}
we have
\begin{align*}
\sfrac{1}{L^2}C^{(0)}
=(a Q^*Q+L^2 D_0)^{-1}
={(\bbbl_*Q^*_1\fQ_1Q_1\bbbl_*^{-1}+ \bbbl_* D_1\bbbl_*^{-1})}^{-1}
=\bbbl_*S_1\bbbl_*^{-1}
\end{align*}
or, in terms of kernels,
\begin{equation*}
C^{(0)}(x_4,x_1)
=\sfrac{1}{L^3}S_1(\bbbl^{-1} x_4,\bbbl^{-1} x_1)
\end{equation*}
by \cite[Lemma \lemPoPscaling.a]{POA}. So
\begin{align*}
\ell(M_0)
&=-\sfrac{2}{L^6|\cX_0^{(1)}|}
   \int_{\cX_0} dx_1 \cdots dx_4\  V_0(x_1,x_2,x_3,x_4)\ 
   S_1(\bbbl^{-1} x_4,\bbbl^{-1} x_1) +O\big(|\mu_0|\fv_0\big)\\
%&=-\sfrac{2L^{14}}{|\cX_0^{(1)}|}
%   \int_{\cX_1} du_1 \cdots du_4\  
%         V_0(\bbbl u_1,\bbbl u_2,\bbbl u_3,\bbbl u_4)\ 
%   S_1(u_4,u_1) +O\big(|\mu_0|\fv_0\big)\\
&=-\sfrac{2}{|\cX_0^{(1)}|}
   \int_{\cX_1} du_1 \cdots du_4\  
         V^{(u)}_1(u_1,u_2,u_3,u_4)\ 
   S_1(u_4,u_1) +O\big(|\mu_0|\fv_0\big)\\
&=\de\mu^*_0  +O\big(|\mu_0|\fv_0\big)
\end{align*}

Similarly, for $n\ge 1$, using $1_0^{(n+1)}$ to denote the constant function 
on $\cX_0^{(n+1)}$ that always takes the value $1$, and using 
$M_n(x_2,x_3)$ to denote the kernel of $M_n$,
\begin{align*}
\ell(M_n)
&=\sfrac{1}{|\cX_0^{(n+1)}|}\int_{\cX_0^{(n+1)}} dx_2 dx_3\ M_n(x_2,x_3)\\
&=-\sfrac{2}{L^3|\cX_0^{(n+1)}|}\int_{\cX_n} du_1 \cdots du_4\  V_n(u_1,u_2,u_3,u_4)\\
\noalign{\vskip-0.1in}
&\hskip1.5in \big(S_n(\mu_n) Q_n^* \fQ_n\,C^{(n)}\fQ_n\,Q_n 
                        S_n(\mu_n)\big)(u_4,u_1)\\
&\hskip1.5in  
  \big(S_{n+1}(L^2\mu_n)Q_{n+1}^* \fQ_{n+1}1_0^{(n+1)}\big)(\bbbl^{-1}u_2)\\
&\hskip1.5in  
  \big(S_{n+1}(L^2\mu_n)^*Q_{n+1}^* \fQ_{n+1}1_0^{(n+1)}\big)(\bbbl^{-1}u_3)\\
&=-\sfrac{2}{L^3|\cX_0^{(n+1)}|}\big(\sfrac{a_{n+1}}{a_{n+1}-L^2\mu_n}\big)^2
   \int_{\cX_n} du_1 \cdots du_4\  V_n(u_1,u_2,u_3,u_4)\\
\noalign{\vskip-0.05in}
&\hskip2in \big(S_n(\mu_n) Q_n^* \fQ_n\,C^{(n)}\fQ_n\,Q_n 
                        S_n(\mu_n)\big)(u_4,u_1)\\
&=-\sfrac{2}{L^3|\cX_0^{(n+1)}|}
   \int_{\cX_n} du_1 \cdots du_4\  V_n(u_1,u_2,u_3,u_4)
    \big(S_n Q_n^* \fQ_n\,C^{(n)}\fQ_n\,Q_n 
                        S_n\big)(u_4,u_1)
\\ \noalign{\vskip-0.05in} &\hskip2in 
     +O\big(|\mu_n|\fv_n\big)\\
&=-\sfrac{2}{L^3|\cX_0^{(n+1)}|}
   \int_{\cX_n}\!\! du_1 \cdots du_4\  V_n^{(u)}(u_1,u_2,u_3,u_4) 
           (S_n Q_n^* \fQ_n\,C^{(n)}\fQ_n\,Q_n S_n)(u_4,u_1)
\\ \noalign{\vskip-0.05in} &\hskip2in 
  +O\big(\big[|\mu_n|+\fv_0^{\frac{2}{3}-6\eps}\big]\fv_n\big)
\end{align*}
since $\sfrac{a_{n+1}}{a_{n+1}-L^2\mu_n}=1+O(|\mu_n|)$
and $\|S_n(\mu_n)-S_n\|_{2m} \le \Gam_\op |\mu_n|$, by
\cite[Proposition \POGmainpos]{POA},
and $\|V_n-V_n^{(u)}\|_{2m} \le \CC_{\de\cV} \fv_0^{\frac{2}{3}-6\eps}\fv_n$, 
by \cite[Remark \remHTpreciseinduction\ and Lemma \lemPARmunvn.b]{PAR1}.
By \cite[Remark \remBSedA.c]{BlockSpin}
\begin{equation*}
 S_nQ_n^*\fQ_n C^{(n)}\fQ_nQ_nS_n
= L^2\,\bbbs^{-1} S_{n+1}\bbbs -S_n
\end{equation*}
In terms of kernels, by \cite[Lemma \lemPoPscaling.a]{POA},
\begin{equation*}
\big(S_nQ_n^*\fQ_n C^{(n)}\fQ_nQ_nS_n\big)(u_4,u_1)
= \sfrac{1}{L^3} S_{n+1}(\bbbl^{-1}u_4,\bbbl^{-1}u_1) -S_n(u_4,u_1)
\end{equation*}
Now, by \cite[Definition \defHTbackgrounddomaction.a,b]{PAR1},
\begin{align*}
&\sfrac{2}{L^3|\cX_0^{(n+1)}|}
   \int_{\cX_n} du_1 \cdots du_4\  V_n^{(u)}(u_1,u_2,u_3,u_4)\ 
     \sfrac{1}{L^3} S_{n+1}(\bbbl^{-1}u_4,\bbbl^{-1}u_1)\\
%&\hskip0.5in = \sfrac{2L^{20}}{L^6|\cX_0^{(n+1)}|}
%\int_{\cX_{n+1}} dv_1 \cdots dv_4\  
%    V_n^{(u)}(\bbbl v_1,\bbbl v_2,\bbbl v_3,\bbbl v_4)\   S_{n+1}(v_4,v_1)\\
&\hskip0.5in = \sfrac{2}{|\cX_0^{(n+1)}|}
\int_{\cX_{n+1}} dv_1 \cdots dv_4\  
    V_{n+1}^{(u)}(v_1,v_2,v_3,v_4)\   S_{n+1}(v_4,v_1)
\end{align*}
so that
\begin{align*}
\ell(M_n)
&=-\sfrac{2}{|\cX_0^{(n+1)}|}
\int_{\cX_{n+1}} dv_1 \cdots dv_4\  
    V_{n+1}^{(u)}(v_1,v_2,v_3,v_4)\   S_{n+1}(v_4,v_1)
\\ \noalign{\vskip-0.05in} &\hskip0.2in 
+\sfrac{2}{L^3|\cX_0^{(n+1)}|}
   \int_{\cX_n} du_1 \cdots du_4\  V_n^{(u)}(u_1,u_2,u_3,u_4)\, 
           S_n(u_4,u_1)
\\ \noalign{\vskip-0.05in} &\hskip2in 
  +O\big(\big[|\mu_n|+\fv_0^{\frac{2}{3}-6\eps}\big]\fv_n\big)\\
&=\de\mu^*_n
   +O\big(\big[|\mu_n|+\fv_0^{\frac{2}{3}-6\eps}\big]\fv_n\big)
\end{align*}
since $|\cX_0^{(n+1)}|=\sfrac{1}{L^5}\cX_0^{(n)}$.
Using Corollary \ref{corLOCproj}.b,
Proposition \ref{propOSFcEtwo}.a,b gives the claim.

\Item (b) 
Recall, from Lemma \ref{lemRENrenormchem} and \eqref{eqnRENlavar}, that
$\de\mu_n$ obeys
\begin{equation*}
\frac{a_{n+1}^2}{a_{n+1}-L^2\mu_n-\de\mu_n}
  -\frac{a_{n+1}^2}{a_{n+1}-L^2\mu_n} 
= \ell(\tilde\cE_\fl)
\end{equation*}
As
\begin{align*}
\frac{a_{n+1}^2}{a_{n+1}-L^2\mu_n-\de\mu}
&=\frac{a_{n+1} }{1-\sfrac{L^2\mu_n+\de\mu}{a_{n+1}}}
=a_{n+1}\Bigg[1+\sfrac{L^2\mu_n+\de\mu}{a_{n+1}}
+\frac{ \big(\sfrac{L^2\mu_n+\de\mu}{a_{n+1}}\big)^2}
      {1-\sfrac{L^2\mu_n+\de\mu}{a_{n+1}}}\Bigg]\\
&=a_{n+1}+L^2\mu_n+\de\mu
+\frac{ (L^2\mu_n+\de\mu)^2}
      {a_{n+1}-L^2\mu_n-\de\mu}
\end{align*}
the left hand side
\begin{align*}
\frac{a_{n+1}^2}{a_{n+1}-L^2\mu_n-\de\mu}\bigg|_{\de\mu=0}^{\de\mu=\de\mu_n}
&=\de\mu_n+\frac{ (L^2\mu_n+\de\mu_n)^2}{a_{n+1}-L^2\mu_n-\de\mu_n}
         -\frac{ (L^2\mu_n)^2}{a_{n+1}-L^2\mu_n}\\
&=\de\mu_n+\frac{ \de\mu_n\big[a_{n+1}(2L^2\mu_n+\de\mu_n)
                                -(L^2\mu_n+\de\mu_n)L^2\mu_n\big]}
        {\big(a_{n+1}-L^2\mu_n-\de\mu_n\big)\big(a_{n+1}-L^2\mu_n\big)}
%\\&=\de\mu_n\big[1+ O\big(L^2|\mu_n|+|\de\mu_n|\big)\big]
\end{align*}
As $L^2|\mu_n|,|\de\mu_n|\le\sfrac{1}{4} a_{n+1}$,
\begin{align*}
|\de\mu_n-\ell(\tilde\cE_\fl)|
&=\bigg|\de\mu_n\frac{ a_{n+1}(2L^2\mu_n+\de\mu_n)
                                -(L^2\mu_n+\de\mu_n)L^2\mu_n}
        {\big(a_{n+1}-L^2\mu_n-\de\mu_n\big)\big(a_{n+1}-L^2\mu_n\big)}\bigg|\\
&\le \sfrac{12 L^2}{a_{n+1}} |\de\mu_n|\big(|\mu_n|+|\de\mu_n|\big) \\
&\le \sfrac{12 L^2}{a_{n+1}} \sfrac{3}{\bar\ka^2}\fe_\fl(n)
            \big(|\mu_n|+\sfrac{3}{\bar\ka^2}\fe_\fl(n)\big) 
   \qquad\text{(by Lemma \ref{lemRENrenormchem})}
\end{align*}
Part (a) and
\begin{equation*}
\sfrac{\fe_\fl(n)}{\bar\ka^2} =  L^{-(2\eta-\eta_\fl)n}\fv_0^{1-4\eps}
\le \fv_0^{1-4\eps} 
\end{equation*}
now implies, using Remark \ref{remOSFvnmum}, that
\begin{align*}
|\de\mu_n-\de\mu^*_n|
&\le  \LLa'_{\de\mu}\ \fv_0^{1-4\eps}  \big(\fv_0^{\sfrac{1}{3}-5\eps}
               + |\mu_n|
               + r_n^2e^{-r_n^2}
               \big)\bar\ka_\fl^6 \\
&\le  \LLa_{\de\mu}\ \fv_0^{1-7\eps}  \big(\fv_0^{\sfrac{1}{3}-5\eps}
               + L^{2n}(\mu_0-\mu_*)
               \big)L^{3\eps(n+1)}\\
\end{align*}
with a new $\LLa_{\de\mu}$.

\Item (c)  Since $\mu_{n+1}=L^2\mu_n+\de\mu_n$ and
$\mu^*_{n+1}=L^2\mu^*_n+\de\mu^*_n$, we have
\begin{align*}
\big|\mu_{n+1}-\mu^*_{n+1}\big|
&\le L^2\big|\mu_n-\mu^*_n\big|
    +\big|\de\mu_n-\de\mu^*_n\big| \\
%&\le L^{2(n+1)}\,\fv_0^{1-8\eps}
%     \smsum_{\ell=1}^n \sfrac{1}{L^{(2-3\eps)\ell}}
%        \big[\fv_0^{\sfrac{1}{3}-5\eps} +L^{2\ell}(\mu_0-\mu_*)\big] \\
%&\hskip0.5in  +\LLa_{\de\mu}\ \fv_0^{1-7\eps}  \big[\fv_0^{\sfrac{1}{3}-6\eps}
%               + L^{2n}(\mu_0-\mu_*)
%               \big]L^{3\eps(n+1)} \\
&\le L^{2(n+1)}\,\fv_0^{1-8\eps}
     \smsum_{\ell=1}^{n+1} \sfrac{1}{L^{(2-3\eps)\ell}}
        \big[\fv_0^{\sfrac{1}{3}-5\eps} +L^{2\ell}(\mu_0-\mu_*)\big] 
\end{align*} 
by \cite[Remark \remHTpreciseinduction]{PAR1} and part (b).
\end{proof}
\subsubsection*{Completion of the Inductive Proof of \cite[Theorem \thmTHmaintheorem\  and Remark \remHTpreciseinduction]{PAR1}}

\begin{proof}
Set
\begin{align*}
\cR_{n+1}(\phi_*,\phi )
&= \tilde\cR_{n+1}\big((\phi_*,\{\partial_\nu\phi_*\}),(\phi,\{\partial_\nu\phi\})
          \big)
\\
\tilde\cE_{n+1}(\tilde\psi_*,\tilde\psi)
&=\tilde\cE_{n+1,1}(\tilde\psi_*,\tilde\psi)
       +\tilde\cE_{n+1,2}(\tilde\psi_*,\tilde\psi)\\
\cE_{n+1}(\psi_*,\psi )
&= \tilde\cE_{n+1}\big((\psi_*,\{\partial_\nu\psi_*\}),
                     (\psi,\{\partial_\nu\psi\})   \big)\\
\cZ_{n+1}&=\cZ_n \tilde N^{(n)}_\bbbt \cZ'_n
\end{align*}
where 
  $\tilde\cE_{n+1,1}$ was defined in Lemma \ref{lemOSFmainlem}, 
  $\tilde\cR_{n+1}$ and $\tilde\cE_{n+1,2}$ were defined in  
                          Lemma \ref{lemRENcRcE},
  $\tilde N^{(n)}_\bbbt$  was defined in  
            \cite[Definition \defHTapproximateblockspintr]{PAR1}
  and  $\cZ'_n$ was defined in Proposition \ref{propOSFmainprop}. 
Then, by \cite[Corollary \corSTmainCor]{PAR1}, Lemma \ref{lemOSFmainlem}.b,c, 
Proposition \ref{propOSFmainprop} and  Lemma \ref{lemRENcRcE},
\begin{align*}
&
\Big( (\bbbs \bbbt_n^{(SF)}) \circ(\bbbs \bbbt_{n-1}^{(SF)})\circ
   \cdots \circ (\bbbs \bbbt_0^{(SF)}) \Big)
\Big(e^{\cA_0(\psi^*,\psi) } \Big)
\\
&\hskip 0.5in
= \sfrac{1}{\cZ_{n+1}}\exp\Big\{- A_{n+1}(\psi_*,\psi, \phi_{*n+1}, \phi_{n+1},\,\mu_{n+1},\cV_{n+1})\\\noalign{\vskip-0.1in}
&\hskip2in+\cR_{n+1}(\phi_{*n+1},\phi_{n+1})+\cE_{n+1}(\psi_*,\psi)
\Big\}
\end{align*}
The bounds on $|\mu_{n+1}-\mu_{n+1}^*|$, 
$\big\|\cV_{n+1}-\cV^{(u)}_{n+1}\big\|_{2m}$ and
$\big\|\tilde\cR_{n+1}^{(\vp)}\big\|_m$
required by \cite[Remark \remHTpreciseinduction]{PAR1} were proven in 
Lemmas \ref{lemRENmunppties}, \ref{lemRENreninteraction} and \ref{lemRENcRcE}. 
That these bounds in turn imply the bounds on 
$\big|\mu_{n+1}-L^{2(n+1)}(\mu_0-\mu_*)\big|$,
$\|\cV_{n+1}-\cV_{n+1}^{(u)}\|_{2m}$ and
$\big\|\tilde\cR_{n+1}^{(\vp)}\big\|_m$ specified in \cite[Theorem \thmTHmaintheorem]{PAR1} was pointed out in \cite[Remark \remHTpreciseinduction]{PAR1}.

By Lemma \ref{lemOSFmainlem}.b, Lemma \ref{lemRENcRcE} and
\cite[Theorem \thmTHmaintheorem]{PAR1},
$\tilde\cE_{n+1}$ does not contain any  scaling/weight relevant monomials
and
\begin{equation*}
\|\tilde\cE_{n+1}\|^{(n+1)} 
   \le L^5\,\sdf(\CC_\fl)\,\fv_0^\eps
       +\CC_\ren\,\fe_\fl(n)
  \le\fv_0^\eps
\end{equation*}
by Remark \ref{remOSFsdf} (with $L$ chosen big enough that $L^5\,\sdf(\CC_\fl)\le\half$) 
and \cite[(\eqnPARestrad.a)]{PAR1}. 
\end{proof}

\newpage
\appendix
%&&&&&&&&&&&&&&&&&&&&&&&&&&&&&&&&&
\section{The Limiting Behaviour of $\mu_n^*$}\label{appMustar}
%&&&&&&&&&&&&&&&&&&&&&&&&&&&&&&&&&

In \cite[\S\sectINTstartPoint]{PAR1}, we defined
\begin{align*}
\mu_*&=2\int_{{((\bbbz/L_\tp\bbbz) \times \bbbz^3)}^3} dx_1 \cdots dx_3\  
              \bV_0(0,x_1,x_2,x_3)\ \bD_0^{-1}(x_3,0) \\
\end{align*}
with $\bD_0=\bbbone - e^{-\bh_0} -e^{-\bh_0} \partial_0\,$ and
in \cite[Remark \remHTpreciseinduction]{PAR1}, we defined, for $n\ge 1$,
\begin{equation*}
\mu_n^* = L^{2n}\mu_0
 -\sfrac{2}{|\cX_0^{(n)}|}
\int_{\cX_n^4} du_1 \cdots du_4\  
    V_n^{(u)}(u_1,u_2,u_3,u_4)\ S_n(u_4,u_1)
\end{equation*}
From \cite[(\eqnHTdefcriticalmu)]{PAR1}, we see that there will be a well
developed potential well when $\mu_n$ is sufficiently positive for large
$n$. As $\mu_n\approx\mu_n^*$ (see \cite[(\eqnHTinductiveVmuestimates)]{PAR1}),
the following lemma shows that this is the case if $\mu_0-\mu_*$ is sufficiently
positive. That is why we expect $\mu_*$ to be the critical $\mu$, to leading
order in the coupling constant.

\begin{lemma}\label{lemOSImustar}
There is a constant $\cc_{\mu_*}$, depending only on $\Gam_\op$ and
$m$, such that
\begin{equation*}
\big|L^{2n}\big(\mu_0-\mu_*\big)-\mu_n^*\big|
\le \cc_{\mu_*}\fv_0
\end{equation*}
for all $1\le n\le n_p$.
\end{lemma}

\begin{proof}
In \cite[(\eqnPOGSntransapprox)]{POA} we defined, on $\cX_n$, the operator
\begin{align*}
S_n'&=\big[D_n+a_n\exp\{-\De_n\}\big]^{-1}\qquad\text{where} \\
\De_n&=\partial_0^*\partial_0 
              +\big(\partial_1^*\partial_1+\partial_2^*\partial_2
                     +\partial_3^*\partial_3\big)\quad\text{and}\quad
a_n=a\big(1 +\smsum_{j=1}^{n-1}\sfrac{1}{L^{2j}}\big)^{-1}
\end{align*}
It is fully translation invariant with respect to $\cX_n$,
is exponentially decaying, and has the same local singularity as $S_n$.
Precisely, we proved in \cite[Lemma \lemPOGSnppties.d]{POA} that
\begin{equation*}
\big|S_n(u,u')-S'_n(u,u')\big|\le \Gam_\op e^{-2m |u-u'|}
\end{equation*}
so that
\begin{align*}
&\sfrac{2}{|\cX_0^{(n)}|}
\bigg|\int_{\cX_n^4} du_1 \cdots du_4\  
    V^{(u)}_n(u_1,u_2,u_3,u_4)\   
   \big\{S_n(u_4,u_1)-S'_n(u_4,u_1)\big\}\bigg|\\
&\hskip1in\le \sfrac{2\Gam_\op}{|\cX_0^{(n)}|}
\int_{\cX_n^4} du_1 \cdots du_4\ |V^{(u)}_n(u_1,u_2,u_3,u_4)|
\le  2\Gam_\op\sfrac{\fv_0 }{L^n}
\end{align*}
and, by \cite[Definition \defHTbackgrounddomaction.b]{PAR1},
\begin{align}
L^{2n}\mu_0-\mu_n^*
%&=\sfrac{2}{|\cX_0^{(n)}|}
%\int_{\cX_n^4} du_1 \cdots du_4\  
%    V^{(u)}_n(u_1,u_2,u_3,u_4)\   
%   S_n(u_4,u_1)\nonumber\\
&=\sfrac{2}{|\cX_0^{(n)}|}
\int_{\cX_n^4} du_1 \cdots du_4\  
    V^{(u)}_n(u_1,u_2,u_3,u_4)\   
   S'_n(u_4,u_1)
   +O\big(\sfrac{\fv_0}{L^n}\big)\nonumber\\
&=\sfrac{2L^{14n}}{|\cX_0^{(n)}|}
\int_{\cX_n^4} du_1 \cdots du_4\  
   V_0(\bbbl^n u_1,\bbbl^n u_2,\bbbl^n u_3,\bbbl^n u_4)\   
   S'_n(u_4,u_1)
   +O\big(\sfrac{\fv_0}{L^n}\big)\nonumber\displaybreak[0]\\
&=\sfrac{2}{L^{6n}|\cX_0^{(n)}|}
\int_{\cX_0^4} dx_1 \cdots dx_4\  
   V_0(x_1,x_2,x_3,x_4)\   
   S'_n(\bbbl^{-n} x_4,\bbbl^{-n} x_1)
   +O\big(\sfrac{\fv_0}{L^n}\big)\nonumber\\
&=\sfrac{2}{L^n|\cX_0|}
\int_{\cX_0^4} dx_1 \cdots dx_4\  
   V_0(x_1,x_2,x_3,x_4)\   
   S'_n(\bbbl^{-n} x_4,\bbbl^{-n} x_1)
   +O\big(\sfrac{\fv_0}{L^n}\big)\nonumber\\
&=\sfrac{2}{L^n}
\int_{\cX_0^3} dx_1 \cdots dx_3\  
   V_0(0,x_1,x_2,x_3)\,  
   S'_n(\bbbl^{-n} x_3,0)
   +O\big(\sfrac{\fv_0}{L^n}\big)
\label{eqnOSIa}
\end{align}
The operator $S'_n$ acts on on $L^2(\cX_n)$ with, as in 
\cite[Definition \defHTbackgrounddomaction.a]{PAR1},
\begin{equation*}
\cX_n=\big(\veps_n^2\bbbz/L_\tp\veps_n^2\bbbz\big)\times
         \big(\veps_n\bbbz^3/L_\sp\veps_n\bbbz^3\big)
\qquad\text{where}\quad
\veps_n=\sfrac{1}{L^n}
\end{equation*}
It may be expressed as the spatial periodization of an operator $\bS'_n$ 
on $L^2(\bcX_n)$ where 
\begin{equation*}
\bcX_n=\big(\veps_n^2\bbbz/L_\tp\veps_n^2\bbbz\big)
                            \times\veps_n\bbbz^3
\end{equation*}
We define $\bS'_n$ in terms of its Fourier transform
\begin{align*}
\bS'_n(u,u')
&=\int_{\hat\bcX_n} \hat \bS'_n(p)\ e^{ip\cdot(u-u')}
             \sfrac{d^4p}{(2\pi)^4}
\end{align*}
where the dual space
\begin{equation*}
\hat\bcX_n=\big(\sfrac{2\pi}{L_\tp}L^{2n}\bbbz/2\pi L^{2n}\bbbz\big)
\times \big(\bbbr^3/2\pi L^n\bbbz^3\big)
\end{equation*}
and the integral
\begin{align*}
\int_{\hat\bcX_n} f(p)\ \sfrac{d^4p}{(2\pi)^4}
= \sum_{p_0\in \sfrac{2\pi}{L_\tp}L^{2n}\bbbz/2\pi L^{2n}\bbbz}
    \sfrac{L^{2n}}{L_\tp}\ 
       \int_{\bbbr^3/2\pi L^n\bbbz^3} f(p_0,p_1,p_2,p_3)\ 
   \sfrac{dp_1 dp_2 dp_3}{(2\pi)^3}
\end{align*}
The Fourier transform
\begin{equation*}
\hat\bS'_n(p)=\big[\widehat\bD_n(p)
                    +a_n\exp\{-\bDe_n(p)\}\big]^{-1}
\quad\text{with}\quad
\bDe_n(p_0,\bp)= 
\big[\sfrac{\sin\frac{1}{2}\veps_n^2 p_0}{\frac{1}{2}\veps_n^2}\big]^2 \!
+\smsum_{\nu=1}^3\!\big[\sfrac{\sin\frac{1}{2}\veps_n \bp_\nu}{\frac{1}{2}\veps_n}\big]^2 
\end{equation*}
and
\begin{align*}
\widehat\bD_n(p_0,\bp)
%&=\half\veps_n^2 e^{-\hat\bh(\veps_n\bp)}
%\Big[\sfrac{e^{-i\veps_n^2p_0}-1}{\veps_n^2}\Big]
%\Big[\sfrac{e^{i\veps_n^2p_0}-1}{\veps_n^2}\Big]
%+\sfrac{1}{\veps_n^2}\big[1-e^{-\hat\bh(\veps_n\bp)}\big]\\
%&\hskip2.4in-i\,e^{-\hat\bh(\veps_n\bp)}\sfrac{\sin\veps_n^2p_0}{\veps_n^2}
%\quad p=(p_0,\bp)\in\hat\cX_n\\
&=\half\veps_n^2 e^{-\hat\bh_0(\veps_n\bp)}
\bigg[\frac{\sin\frac{1}{2}\veps_n^2p_0}{\frac{1}{2}\veps_n^2}\bigg]^2
+\frac{1-e^{-\hat\bh_0(\veps_n\bp)}}{\veps_n^2}
-ie^{-\hat\bh_0(\veps_n\bp)}\frac{\sin\veps_n^2p_0}{\veps_n^2}
\end{align*}
Define, for $u_3,u_4\in\cX_n$, $\tilde\bS'_n(u_3,u_4)
=\bS'_n(\tilde u_3,\tilde u_4)$, where $\tilde u_3$ and $\tilde u_4$ are
representatives of $u_3,\ u_4$ in $\bcX_n$ that minimize the magnitude
of each spatial component of $\tilde u_3-\tilde u_4$. 
Thanks to the exponential decay of $\bS'_n$ proven in
\cite[Lemma \lemPOGSnppties.d]{POA}, the 
difference $S'_n(\bbbl^{-n}x_3,\bbbl^{-n}x_4)
-\tilde\bS'_n(\bbbl^{-n}x_3,\bbbl^{-n}x_4)$ is bounded, uniformly in $n$.
Hence
\begin{equation}\label{eqnOSIb}
\sfrac{2}{L^{n}}\bigg|
\int_{\cX_0} dx_1 \cdots dx_3\  
   V_0(0,x_1,x_2,x_3)\,  
   \big\{S'_n(\bbbl^{-n} x_3,0)
            -\tilde\bS'_n(\bbbl^{-n} x_3,0)\big\}\bigg|
   =O\big(\sfrac{\fv_0}{L^{n}}\big)
\end{equation}

So we consider
\begin{align*}
\lim_{n\rightarrow\infty}\!
\sfrac{2}{L^{n}}
\int_{\cX_0} dx_1 \cdots dx_3\  
   V_0(0,x_1,x_2,x_3)\,  
   \tilde\bS'_n(\bbbl^{-n} x_3,0)
\end{align*}
If $\tilde x_3\in\bcX_0$ is the representative of $x_3\in\cX_0$ whose
spatial components have minimum magnitude, then
\begin{align*}
\sfrac{1}{L^n}\tilde\bS'_n(\bbbl^{-n} x_3,0)
&=\sfrac{1}{L^n}\int_{\hat\bcX_n} \hat \bS'_n(p) 
      e^{ip\cdot(\bbbl^{-n} \tilde x_3)}
             \sfrac{d^4p}{(2\pi)^4}\\
&=L^{4n}\int_{\hat\bcX_0} \hat \bS'_n(\bbbl^n k) e^{ik\cdot \tilde x_3}
             \sfrac{d^4k}{(2\pi)^4}\qquad\text{with }k=\bbbl^{-n} p
\end{align*}
Observe that
\begin{align*}
L^{2n}\hat \bS'_n(\bbbl^n k)
=\bigg\{2 e^{-\hat\bh_0(\bk)} \sin^2\half k_0
+\big(1\!-\!e^{-\hat\bh_0(\bk)}\big)
-i\,e^{-\hat\bh_0(\bk)}\sin k_0
+\sfrac{a_n}{L^{2n}}e^{-\bDe_n(\bbbl^n k)}\bigg\}^{-1}
\end{align*}
converges pointwise, as $n\rightarrow\infty$, to
\begin{equation}\label{eqnOSIdinversek}
\widehat\bD_0(k)^{-1}=\Big\{2 e^{-\hat\bh_0(\bk)} \sin^2\half k_0
+\big(1\!-\!e^{-\hat\bh_0(\bk)}\big)
-i\,e^{-\hat\bh_0(\bk)}\sin k_0\Big\}^{-1}
\end{equation}
and is bounded, uniformly in $n$, by
$
\big|\widehat\bD_0(k)\big|^{-1}
\in L^1(\hat\bcX_0)
$.
Hence $\sfrac{1}{L^{3n}}\tilde\bS'_n(\bbbl^{-n} x_3,0)$ is bounded,
uniformly in $n$ and $x_3$ and converges pointwise, as $n\rightarrow\infty$, to
\begin{equation*}
\bD_0^{-1}(\tilde x_3,0)
=\int_{\hat\bcX_0} \widehat\bD_0(k)^{-1}
    e^{ik\cdot\tilde x_3}  \sfrac{d^4k}{(2\pi)^4}
\end{equation*}
Hence, by \eqref{eqnOSIa} and \eqref{eqnOSIb},
\begin{equation}\label{eqnMSalmostDone}
\begin{split}
& \big(\mu_0-\mu_*\big)-\sfrac{1}{L^{2n}}\mu_n^*
=2 \!\int_F \! d\tilde x_1 \cdots d\tilde x_3\  
   \bV_0(0,\tilde x_1,\tilde x_2,\tilde x_3)\,  
   \big[\sfrac{1}{L^{3n}}\bS'_n(\bbbl^{-n} \tilde x_3,0)
     -\bD_0^{-1}(\tilde x_3,0)\big] \\
&\hskip2in
   +O\big(\sfrac{\fv_0}{L^{3n}}\big)+O\big(\fv_0 e^{-m L_\sp}\big) 
\end{split}
\end{equation}
where
$
F=\Big\{\ (\tilde x_1,\tilde x_2,\tilde x_3)^3\in 
          \bcX_0^3
   \ \Big|\ 
        \sfrac{L_\sp}{2}< \tilde x_{i,j}<\sfrac{L_\sp}{2} 
        \text{\ for all $i,j=1,2,3$}
   \ \Big\}
$.
To bound the right hand side, observe that
\begin{align*}
&\big| \sfrac{1}{L^{3n}}\bS'_n(\bbbl^{-n} \tilde x_3,0)
                                -\bD_0^{-1}(\tilde x_3,0)\big|
\le \sfrac{a_n}{L^{2n}}\int_{\hat\bcX_0} 
     \sfrac{1}{|\widehat\bD_0(k)+\sfrac{a_n}{L^{2n}}e^{-\bDe_n(\bbbl^n k)}|\
|\widehat\bD_0(k)|}
     \sfrac{d^4k}{(2\pi)^4}
\end{align*}
As $\hat\cX_0$ is a compact set, both terms 
$2 e^{-\hat\bh_0(\bk)} \sin^2\half k_0$ and $\big(1-e^{-\hat\bh_0(\bk)}\big)$
of the real part of $\widehat\bD_0(k)$
are nonnegative, and $\widehat\bD_0(k)$ is bounded away from zero outside
of any neignbourhood of $k_0=0$, $\bk=0$ we have 
\begin{align*}
\big|\widehat\bD_0(k)\big|
&\ge\const\big| ik_0+\bk^2\big| \\
\big|\widehat\bD_0(k)+\sfrac{a_n}{L^{2n}}e^{-\bDe_n(\bbbl^n k)}\big|
&\ge \const\big| ik_0+\bk^2
          +\sfrac{a_n}{L^{2n}}e^{-\const L^{2n}[k_0^2+\bk^2]}\big|
\end{align*}
For the part of the integral with $k_0\ne 0$,
\begin{align*}
 \sfrac{a_n}{L^{2n}}\int_{\atop{\hat\bcX_0}{k_0\ne 0}} 
     \sfrac{1}{|\widehat\bD_0(k)+\sfrac{a_n}{L^{2n}}e^{-\bDe_n(\bbbl^n k)}|\
|\widehat\bD_0(k)|}
     \sfrac{d^4k}{(2\pi)^4}
&\le  \sfrac{\const}{L^{2n}}\int_{-\pi}^\pi dk_0\ \sfrac{1}{|k_0|^{3/4}}
       \int_{|\bk|\le 2\pi} d^3\bk \ \sfrac{1}{|\bk|^{5/2}} \\
&\le  \sfrac{\const}{L^{2n}}
\end{align*}
For the part of the integral with $k_0= 0$, scaling $\bk=\sfrac{\bp}{L^n}$,
\begin{align*}
 \sfrac{a_n}{L^{2n}}\int_{\atop{\hat\bcX_0}{k_0= 0}} 
     \sfrac{1}{|\widehat\bD_0(k)+\sfrac{a_n}{L^{2n}}e^{-\bDe_n(\bbbl^n k)}|\
|\widehat\bD_0(k)|}
     \sfrac{d^4k}{(2\pi)^4}
&\le  \sfrac{\const}{L^{2n} L_\tp}\int_{|\bk|\le 2\pi} 
     \sfrac{1}{[\bk^2+\sfrac{1}{L^{2n}}e^{-\const L^{2n}\bk^2}]\bk^2}
     d^3\bk  \\
&=  \sfrac{\const}{L^{2n} L_\tp} L^n\int_{|\bp|\le2\pi L^n} 
     \sfrac{1}{[\bp^2+e^{-\const \bp^2}]\bp^2}
     d^3\bp  \\
&\le \sfrac{\const}{L^{2n}} \sfrac{L^n}{L_\tp} \\
&\le \sfrac{\const}{L^{2n}} \qquad\text{ for all }n\le n_p 
\end{align*}
Putting these bounds into \eqref{eqnMSalmostDone},
\begin{align*}
\sfrac{1}{L^{2n}}\big[L^{2n}\big(\mu_0-\mu_*\big)-\mu_n^*\big]
&\le O\big(\sfrac{\fv_0}{L^{2n}}\big) 
      +O\big(\sfrac{\fv_0}{L^{3n}}\big)
      +O\big(\fv_0 e^{-m L_\sp}\big)
\end{align*}

\end{proof}

%&&&&&&&&&&&&&&&&&&&&&&&&&&&&&&&&&
\section{Localization}\label{appLocal}
%&&&&&&&&&&&&&&&&&&&&&&&&&&&&&&&&&

Fix masses $\fm\ge 0$ and $\bar\fm>\fm$.
\begin{lemma}\label{lemLexpandalphiveps}
Let $0\le j\le n$. 
For each point $u$ of the fine lattice $\cX_j^{(n-j)}$, we use  
$X(u)$ to denote the point of the unit lattice $\cX_0^{(n)}$ 
nearest to $u$. There exists a constant $C_{\fm,\bar\fm}$, 
depending only on $\fm$ and $\bar\fm$, such that the following holds.
For each linear tansformation $B:\cH_0^{(n)}\rightarrow\cH_j^{(n-j)}$ 
there are linear maps 
     $B_\nu$, $0\le\nu\le 3$,  
such that
\begin{equation*}
\sum_{x\in\cX^{(n)}_0} 
   B(u,x)\,\big[\psi(x)-\psi\big(X(u)\big)\big]
=\sum_{\nu=0}^3B_\nu(\partial_\nu\psi)(u)
\qquad\text{for all $u\in\cX_j^{(n-j)}$}
\end{equation*}
and
\begin{equation*}
\| B_\nu\|_\fm \le  C_{\fm,\bar\fm}  \|B\|_{\bar\fm}
\qquad 0\le\nu\le 3
\end{equation*}

\end{lemma}
\begin{proof} 
Define
\begin{itemize}[leftmargin=*, topsep=2pt, itemsep=2pt, parsep=0pt]
\item
 $\cB_0^{(n)}$ to be the set of (oriented) bonds on 
 the lattice $\cX_0^{(n)}$. 
%%%
\item
 For any bond $b=\<x_1,x_2\>\in\cB_0^{(n)}$, 
 $\nabla\psi(b)=\sfrac{\psi(x_2)-\psi(x_1)}{|x_2-x_1|}
 =\psi(x_2)-\psi(x_1)$.
%%%
\item
 Given fields $\psi_\nu$, $0\le\nu\le 3$
 on $\cX_0^{(n)}$, we write, for each $\<x_1,x_2\>\in\cB_0^{(n)}$
\begin{equation}\label{eqnLpsiNabla}
\psi_\snabla(\{\psi_\nu\})(\<x_1,x_2\>)
            =\begin{cases}\psi_\nu(x_1)& \text{if $x_2-x_1=|x_2-x_1|e_\nu$}\\
                  -\psi_\nu(x_2)& \text{if $x_1-x_2=|x_2-x_1|e_\nu$}
              \end{cases}
\end{equation}
where $e_\nu$ is the usual unit vector in direction $\nu$.
Observe that $\psi_\snabla(\{\partial_\nu\psi\})(b)=\nabla\psi(b)$.
%%%%
\item
 If $x,x'\in\cX_0^{(n)}$ we select by any reasonable
algorithm a set $\Pi(x,x')\subset \cB_0^{(n)}$ of bonds forming 
a path from $x$, $x'$. This algorithm must be such that no bond
ever appears more than once, even ignoring orientation, in any $\Pi(x,x')$ 
and such that if $z$ is any 
point on a path $\Pi(x,x')$, then $|x-z|,|z-x'|\le|x-x'|$.  
We have $\psi(x')-\psi(x)=\sum_{b\in\Pi(x,x')} \nabla\psi(b)$. 
\end{itemize}
Using this notation,
\begin{align*}
&\sum_{x\in\cX^{(n)}_0} 
   B(u,x)\,\big[\psi(x)-\psi\big(X(u)\big)\big]
=\sum_{x\in\cX^{(n)}_0} \ 
  \sum_{b\in\Pi(X(u),x)}\hskip-5pt
    B(u,x)\,\nabla\psi(b)\\
&\hskip0.5in=\sum_{x\in\cX^{(n)}_0} \ 
  \sum_{b\in\Pi(X(u),x)}
    B(u,x)\,\psi_\snabla(\{\partial_\nu\psi\})(b)
=\sum_{\nu=0}^3 \sum_{z\in\cX_0^{(n)}}
       B_\nu(u;z)\,\partial_\nu\psi(z)
\end{align*}
with
\begin{equation*}
B_\nu(u;z)=
     \hskip-10pt
     \sum_{\atop{x\in\cX^{(n)}_0}
                {\<z,z+ e_\nu\>\in\Pi(X(u),x)}} 
    \hskip-25pt
    B(u,x)\hskip10pt
-\hskip-15pt
    \sum_{\atop{x\in\cX^{(n)}_0}
               {\<z+e_\nu,z\>\in\Pi(X(u),x)}}
    \hskip-25pt
    B(u,x)
\end{equation*}
Recall, from \cite[Definition \defHTkernelnorm]{PAR1}, that 
\begin{align*}
\| B_\nu\|_\fm
&=\max\Big\{
  \sup_{u\in\cX_j^{(n-j)}}\vol_0\hskip-5pt
               \sum_{z\in\cX_0^{(n)}}|B_\nu(u;z)|e^{\fm|u-z|}\ ,\ \\
  \noalign{\vskip-0.2in}
  &\hskip2in
  \sup_{z\in\cX_0^{(n)}}\vol_j\sum_{u\in\cX_j^{(n-j)}}|B_\nu(u;z)|
                                                e^{\fm|u-z|} \Big\}
\end{align*}
Now, writing $m'=\half(\bar\fm-\fm)$,
\begin{align*}
|B_\nu(u;z)|e^{\fm|u-z|}
&\le 2 \hskip-10pt
     \sum_{\atop{x\in\cX^{(n)}_0}
                {z{\rm \,on\,}\Pi(X(u),x)}}
    \hskip-10pt
    \big|B(u,x)\big|
                     e^{\fm|u-z|}\\
&\le 2 e^{2\fm} 
    \hskip-10pt
     \sum_{\atop{x\in\cX^{(n)}_0}
          {z{\rm \,on\,}\Pi(X(u),x)}}
    \hskip-10pt
    \big|B(u,x)\big|
                       e^{\bar\fm|u-x|-(\bar\fm-\fm)|u-x|}\\
&\le 2 e^{2 \bar\fm } 
    \sum_{x\in\cX^{(n)}_0} 
    \big|B(u,x)
               \big|e^{\bar\fm|u-x|-m'|u-z|-m'|z-x|}
\end{align*}
so that 
\begin{equation*}
\| B_\nu\|_\fm\le 2e^{2\bar\fm}  
        \Big(\sup_{u\in\cX_j^{(n-j)}}\sum_{z\in\cX_0^{(n)}}e^{-m'|z-u|}\Big)
                  \|B\|_{\bar\fm}
\end{equation*}
\end{proof}

\begin{corollary}\label{corLprelocalize} 
Let 
\begin{equation*}
\cP(\ga,\psi) = \int dx\,dy\ \ga(x)\, K(x,y)\, \psi(y)
\end{equation*}
be a bilinear form on $\cH_0^{(n)}$ with translation invariant kernel $K$.
Then there exist bilinear forms $\cP_\nu\big(\ga,\psi_\nu\big)$, $0\le\nu\le 3$,  such that
\begin{equation*}
\cP(\ga, \psi)
=\cK \int dx\ \ga(x)\psi(x)
+\sum_{\nu=0}^3\cP_\nu\big(\ga,\partial_\nu\psi\big)
\end{equation*}
where $\cK=\int dy\ K(0,y)$.
Furthermore, for each $0\le\nu\le 3$, the kernel $K_\nu$ of $\cP_\nu$ obeys
$
\|K_\nu\|_{\fm}\le C_{\fm,\bar\fm} \|K\|_{\bar\fm}
$ .

\end{corollary}
\begin{proof} Write
\begin{equation*}
\cP(\ga,\psi) = \int dx\,dy\ \ga(x)\, K(x,y)\, [\psi(y)-\psi(x)]
           +\cK \int dx\ \ga(x)\, \psi(x)
\end{equation*}
where $\cK=\int dy\ K(x,y)$ is independent of $x$.  Lemma \lemLexpandalphiveps,
with $j=0$, and thus $X(x)=x$, gives kernels $K_\nu$, $0\le\nu\le 3$ such that
\begin{equation*}
\int dy\ K(x,y)\, [\psi(y)-\psi(x)]
=\sum_{\nu=0}^3 \int dy\ K_\nu(x,y)\, \partial_\nu\psi(y)
\end{equation*}
and $\|K_\nu\|_\fm\le C_{\fm,\bar\fm}\|K\|_{\bar\fm}$, $0\le\nu\le 3$.
Setting
\begin{align*}
\cP_\nu\big(\ga,\psi_\nu\big)
   =\sum_{x,y\in\cX^{(n)}_0}\ga(x)\, K_\nu(x,y)\,\psi_\nu(y)
\end{align*}
the corollary follows.
\end{proof}

\begin{lemma}\label{lemLlocalize}
There is a constant $\cc_\loc$, depending only on $\fm$ and $\bar\fm$, such that
the following holds.
\begin{enumerate}[label=(\alph*), leftmargin=*]
\item 
Let $1\le\nu\le3$ and let
\begin{equation*}
\cP(\psi_{*\nu},\psi) = \int_{\cX_0^{(n)}} dx\,dy\ \psi_{*\nu}(x)\, K(x,y)\, \psi(y)
\end{equation*}
be invariant under $\fS_\spat$. Then there exists a bilinear form
\begin{equation*}
\cP_{\rm ren}\big(\psi_{*\nu},\big\{\psi_{\nu'}\big\}_{\nu'=0}^3\big)
\end{equation*}
that  is also invariant under $\fS_\spat$, such that 
\begin{equation*}
\cP(\partial_\nu\psi_*, \psi) =\cP_{\rm ren}\big(\partial_\nu\psi_*,
       \big\{\partial_{\nu'}\psi\big\}_{\nu'=0}^3\big) \qquad\text{and}\qquad
\|\cP_{\rm ren}\|_{\fm}\le \cc_\loc \|\cP\|_{\bar\fm}
\end{equation*}

\item
Let $1\le\nu\le3$ and let
\begin{equation*}
\cP(\psi_*,\psi,\psi_\nu) =\int dx_1\cdots dx_4\ 
   K(x_1,x_2,x_3,x_4)\ \psi_*(x_1)\psi(x_2)\psi_*(x_3)\psi_\nu(x_4)
\end{equation*}
be invariant under $\fS_\spat$. Then there exists, in the notation of
Definition \ref{defRENspaces}, 
\begin{equation*}
\cP_{\rm ren}
  \big((\psi_*,\{\psi_{*\nu'}\})\,,\,(\psi,\{,\psi_{\nu'}\})\big)
 \in\fP_{(2,1,1)}\oplus\fP_{(2,0,2)}
\end{equation*}
that
\begin{itemize}[leftmargin=*, topsep=2pt, itemsep=2pt, parsep=0pt]
%\item 
%is also invariant under $\fS_\spat$ and
\item
 is of degree at least one in $\psi_\nu$ and
\item
 obeys $\cP(\psi_*, \psi,\partial_\nu\psi)
=\cP_{\rm ren}\big((\psi_*,\{\partial_{\nu'}\psi_*\})\,,\,
                        (\psi,\{\partial_{\nu'}\psi\})\big)$
and with
\item
 each monomial in $\cP_{\rm ren}$ having $\|\ \cdot\ \|_{\fm}$
norm bounded by  $\cc_\loc \|\cP\|_{\bar\fm}$ 
\end{itemize}

\item
Let
\begin{equation*}
\cP(\psi_*,\psi) = \int_{\cX_0^{(n)}} dx\,dy\ \psi_*(x)\, K(x,y)\, \psi(y)
\end{equation*}
be invariant under $\fS$. Then there exists 
\begin{equation*}
\cP_{\rm ren}\big(\psi_*,\psi,\psi_{*\nu},\psi_{\nu}\big)
\in\fP_{(1,1,0)}\oplus\fP_{(0,1,1)}\oplus\fP_{(0,0,2)}
\end{equation*}
such that each monomial in $\cP_{\rm ren}$ has $\|\ \cdot\ \|_{\fm}$
norm bounded by  $\cc_\loc \|\cP\|_{\bar\fm}$ and
\begin{equation*}
\cP(\psi_*, \psi)
=\de\mu\int dx\ \psi_*(x)\, \psi(x)
+\cP_{\rm ren}\big(\psi_*,\psi,\partial_\nu\psi_*,\partial_\nu\psi\big)
\end{equation*}
where
\begin{equation*}
\de\mu=\int dy\ K(0,y)
\end{equation*}
is real and obeys $|\de\mu|\le \|K\|_{m=0}$.
\end{enumerate}
\end{lemma}
\begin{proof} (a) By Corollary \ref{corLprelocalize}, with $\ga=\psi_{*\nu}$,
\begin{equation*}
\cP(\psi_{*\nu},\psi)=\cK \int dx\ \psi_{*\nu}(x)\psi(x)
+\sum_{\nu'=0}^3\cP_{\nu'}\big(\psi_{*\nu},\partial_{\nu'}\psi\big)
\end{equation*}
We have $K(x,y)=-K(R_\nu x-e_\nu,R_\nu  y)$, by \cite[Lemma \lemSYfullsymmetry]{PAR1}, 
so that
\begin{equation*}
\cK=\int dy\ K(0,y)=-\int dy\ K(-e_\nu,R_\nu y)
=-\int dy\ K(-e_\nu,y) = -\cK
\end{equation*}
yielding $\cK=0$. Set 
\begin{equation*}
\cP'_{\rm ren}\big(\psi_{*\nu},\big\{\psi_{\nu'}\big\}\big)
=\sum_{\nu'=0}^3\cP_{\nu'}\big(\psi_{*\nu},\psi_{\nu'}\big)
\end{equation*}
It has all of the properties required of $\cP_{\rm ren}$, with the 
possible exception of invariance under $\fS_\spat$. To recover invariance 
under $\fS_\spat$ we define $\cP_{\rm ren}$ by averaging over $\fS_\spat$.
$$
\cP_{\rm ren}=\sfrac{1}{|\fS_\spat|}\sum_{g\in\fS_\spat} g\cP'_{\rm ren}
$$
The claim follows by \cite[Remark \remSYsymmetryNorm]{PAR1}.

\Item (b) Write
\begin{equation*}
\cP(\psi_*,\psi,\psi_\nu) =\cK \int dx\ \psi_*(x)\psi(x)\psi_*(x)\psi_\nu(x)
                            +\de\cP(\psi_*,\psi,\psi_\nu)
\end{equation*}
where
\begin{equation*}
\cK=\int dx_1 dx_2 dx_3\ K(x_1,x_2,x_3,0)
\end{equation*}
and
\begin{align*}
&\de\cP(\psi_*,\psi,\psi_\nu)\\
&\hskip0.2in
=\int dx_1\, dx_2\, dx_3\, dx\ 
   K(x_1,x_2,x_3,x)\ \big[\psi_*(x_1)\psi(x_2)\psi_*(x_3)
                       -\psi_*(x)\psi(x)\psi_*(x)\big]\psi_\nu(x)
\end{align*}
As 
$K(x_1,x_2,x_3,x_4)
   =-K(R_\nu x_1,R_\nu x_2,R_\nu x_3,R_\nu  x_4-e_\nu)$, 
by  \cite[Lemma \lemSYfullsymmetry]{PAR1}, we have
\begin{align*}
\cK&=\int dx_1 dx_2 dx_3\ K(x_1,x_2,x_3,0)\\
%&=-\int dx_1 dx_2 dx_3\ 
%       K(R_\nu x_1,R_\nu x_2,R_\nu x_3,- e_\nu) \\
&=-\int dx_1 dx_2 dx_3\ K(x_1,x_2,x_3,- e_\nu) \\
&=-\int dx_1 dx_2 dx_3\ K(x_1+e_\nu,x_2+e_\nu,x_3+e_\nu,0) \\
&= -\cK
\end{align*}
so that $\cK=0$. As in Lemma \ref{lemLexpandalphiveps}, 
\begin{align*}
&\de\cP(\psi_*,\psi,\psi_\nu)\\
&\hskip0.3in
=\ \int dx_1\, dx_2\, dx_3\, dx\ 
   K(x_1,x_2,x_3,x)\ \psi_*(x_1)\psi(x_2)\big[\psi_*(x_3)-\psi_*(x)\big]
                       \psi_\nu(x)\\
&\hskip0.4in
+\!\int dx_1\, dx_2\, dx_3\, dx\ 
   K(x_1,x_2,x_3,x)\ \psi_*(x_1)\big[\psi(x_2)-\psi(x)\big]
                       \psi_*(x)\psi_\nu(x)\\
&\hskip0.4in
+\!\int dx_1\, dx_2\, dx_3\, dx\ 
   K(x_1,x_2,x_3,x)\ \big[\psi_*(x_1)-\psi_*(x)\big]
                       \psi(x)\psi_*(x)\psi_\nu(x)\\
&\hskip0.3in
=\ \cP'_1(\psi_*,\psi,\big\{\partial_{\nu'}\psi_*\big\},\psi_\nu)
 + \cP'_2(\psi_*,\psi,\big\{\partial_{\nu'}\psi\big\},\psi_\nu)
 + \cP'_3(\psi_*,\psi,\big\{\partial_{\nu'}\psi\big\},\psi_\nu)
\end{align*}
where
\begin{align*}
\cP'_1(\psi_*,\psi,\big\{\psi_{*\nu'}\big\},\psi_\nu)
&=\sum_{b\in\Pi(x,x_3)}\int dx_1\, dx_2\, dx_3\, dx\ 
   K(x_1,x_2,x_3,x)\\
\noalign{\vskip-0.1in}
&\hskip2in 
     \psi_*(x_1)\,\psi(x_2)\,\psi_{*\snabla}\big(\{\psi_{*\nu'}\big)(b)\,
        \psi_\nu(x)\\
%%%%
\cP'_2(\psi_*,\psi,\big\{\psi_{*\nu'}\big\},\psi_\nu)
&=\sum_{b\in\Pi(x,x_2)}\int dx_1\, dx_2\, dx_3\, dx\ 
   K(x_1,x_2,x_3,x)\\
\noalign{\vskip-0.1in}
&\hskip2in 
     \psi_*(x_1)\,\psi_\snabla\big(\{\psi_{\nu'}\big)(b)\,
        \psi_*(x)\,\psi_\nu(x)\\
%%%%%
\cP'_3(\psi_*,\psi,\big\{\psi_{*\nu'}\big\},\psi_\nu)
&=\sum_{b\in\Pi(x,x_1)}\int dx_1\, dx_2\, dx_3\, dx\ 
   K(x_1,x_2,x_3,x)\\
\noalign{\vskip-0.1in}
&\hskip2in 
     \psi_{*\snabla}\big(\{\psi_{*\nu'}\big)(b)\, 
        \psi(x)\,\psi_*(x)\,\psi_\nu(x)
\end{align*}
For each $i=1,2,3$, we may write
\begin{equation*}
\cP'_i(\psi_*,\psi,\big\{\psi_{(*)\nu'}\big\},\psi_\nu)
=\sum_{\nu'=0}^3\cP_{i\nu'}(\psi_*,\psi,\psi_{(*)\nu'},\psi_\nu)
\end{equation*}
and bound $\cP_{\nu'}$ just as $\cP_\nu$ was bounded in 
Lemma \ref{lemLexpandalphiveps}. Then it suffices to set
\begin{equation*}
\cP_{\rm ren}
   \big((\psi_*,\{\psi_{*\nu'}\})\,,\,(\psi,\{,\psi_{\nu'}\})\big)
=\sfrac{1}{|\fS_\spat|}\sum_{g\in\fS_\spat}\ 
\sum_{\nu'=0}^3\ \sum_{i=1}^3
g\cP_{i\nu'}(\psi_*,\psi,\psi_{(*)\nu'},\psi_\nu)
\end{equation*}

\Item (c) By Corollary \ref{corLprelocalize}, with $\ga=\psi_*$,
\begin{equation*}
\cP(\psi_*,\psi)=\de\mu \int dx\ \psi_*(x)\psi(x)
+\sum_{\nu=0}^3\cP_{\nu}\big(\psi_*,\partial_\nu\psi\big)
\end{equation*}
with
\begin{equation*}
\|\cP_\nu\|_{(\fm+\bar\fm)/2}\le C'_r\|\cP\|_{\bar\fm}
\end{equation*}
We have $\overline{K(R_0 y,R_0 x)}=K(x,y)$, by 
\cite[Example \exSYfullsymmetry]{PAR1}, 
so that
\begin{equation*}
\de\mu=\int dy\ K(0,y)=\int dy\ \overline{K(R_0 y,0)}
=\int dy\ \overline{K(0,-R_0 y)}
=\int dy\ \overline{K(0,y)}
=\overline{\de\mu}
\end{equation*}
so that $\de\mu$ is real. By averaging as in part (a), we may assume that
each $\cP_{\nu}\big(\psi_*,\psi_\nu\big)$ is
invariant under $\fS_\spat$. It now suffices to apply part (a) to each
$\cP_{\nu}$, $1\le\nu\le3$, and average over $\fS$.
\end{proof}

We fix any $\wf$, $\wf'_0$ and $\wf'_\sp$ 
and use norms $\lun \cF(\tilde \psi_*,\tilde \psi)\run$
and $\lln \cF(\tilde \psi_*,\tilde \psi)\rln$
which associate the weight factor $\wf$ to the fields $\psi_{(*)}$,
the weight factor $\wf'_0$ to the fields $\psi_{(*)0}$,
and the weight factor $\wf'_\sp$ to the fields $\psi_{(*)\nu}$, 
$1\le\nu\le 3$. The norm $\lun \ \cdot\ \run$
has mass $\bar\fm$ and the norm $\lln  \ \cdot\ \rln$
has mass $\fm$.

Let $\fP_\rel$, $\fP_\fD$ and $\fP_\irr$ be the spaces of 
Definition \ref{defRENspaces} and, as in 
Definition \ref{defRENspaces}, denote by 
$\fP_{(4,0,0)}$ the space of quartic monomials in $\psi_*,\psi$ that 
are $\fS$ invariant and particle number preserving.

\begin{proposition}\label{propLOCproj}
There exist linear maps  
\begin{equation*}
\ell:\fP_\rel\rightarrow\bbbc \qquad 
\cL_4:\fP_\rel\rightarrow\fP_{(4,0,0)}\qquad
\cL_\fD:\fP_\rel\rightarrow\fP_\fD \qquad
\cI:\fP_\rel\rightarrow\fP_\irr  
\end{equation*}
such that, for all $\cP\in\fP_\rel$,
\begin{align*}
\cP\big((\psi_*,\{\partial_\nu\psi_*\})\,,\,
                        (\psi,\{\partial_\nu\psi\})\big)
&=\ell(\cP)\hskip-2.5pt\int\hskip-3pt dx \ \psi_*(x) \psi(x)
+\cL_4(\cP)(\psi_*,\psi)   
\\
&\hskip 2cm
+\cL_\fD(\cP)\big((\psi_*,\{\partial_\nu\psi_*\})\,,\,
                        (\psi,\{\partial_\nu\psi\})\big)\\
&\hskip 2cm   +\cI(\cP)\big((\psi_*,\{\partial_\nu\psi_*\})\,,\,
                        (\psi,\{\partial_\nu\psi\})\big)
\end{align*}

\noindent
and 
\begin{itemize}[leftmargin=*, topsep=2pt, itemsep=2pt, parsep=0pt]
\item
 for $\cP\in\fP_\fD$,
\begin{equation*}
\ell(\cP)=0 \qquad \cL_4(\cP)= 0 \qquad \cL_\fD(\cP)=\cP \qquad \cI(\cP)=0 
\end{equation*}

\item
 for $\cP\in\fP_{(4,0,0)}$,
\begin{equation*}
\ell(\cP)=0 \qquad \cL_4(\cP)= \cP \qquad \cL_\fD(\cP)=0 \qquad \cI(\cP)=0 
\end{equation*}

\item
 for 
$
\cP=\int dx\,dx'\ \psi_*(x) K(x,x')\psi(x')\in\fP_{(2,0,0)}
$
\begin{itemize}[leftmargin=*, topsep=2pt, itemsep=2pt, parsep=0pt]
\item
$
\ell(\cP)=\int dx'\ K(0,x')
$ 
\item
$
\cL_4(\cP)=0
$ 
\item
$\cL_\fD(\cP)=\cL_{(1,1,0)}(\cP)+\cL_{(0,1,1)}(\cP)+\cL_{(0,0,2)}(\cP)$ with
\begin{align*}
\cL_{(1,1,0)}(\cP)&\in\fP_{(1,1,0)} &
\lln \cL_{(1,1,0)}(\cP)\rln&\le 2\cc_\loc \sfrac{\wf'_0}{\wf}\lun \cP \run \\
\cL_{(0,1,1)}(\cP)&\in\fP_{(0,1,1)} &
\lln \cL_{(0,1,1)}(\cP)\rln&\le 6\cc_\loc \sfrac{\wf'_0\wf'_\sp}{\wf^2}\lun \cP \run\\
\cL_{(0,0,2)}(\cP)&\in\fP_{(0,0,2)} &
\lln \cL_{(0,0,2)}(\cP)\rln&\le 9\cc_\loc \sfrac{{\wf'_\sp}^2}{\wf^2}\lun \cP \run \\
\end{align*}
\item
$
\cI(\cP)=0
$ 
\end{itemize}
\pagebreak[2]

\item
 for 
$
\cP\in\fP_{(1,0,1)}
$, 
we have
$\ell(\cP)=0\,$, $\cL_4(\cP)=0\,$,
$
\cI(\cP)=0
$ 
and
 $\cL_\fD(\cP)=\cL_{(0,1,1)}(\cP)+\cL_{(0,0,2)}(\cP)$ with
\begin{align*}
\cL_{(0,1,1)}(\cP)&\in\fP_{(0,1,1)} &
\lln \cL_{(0,1,1)}(\cP)\rln&\le \cc_\loc \sfrac{\wf'_0}{\wf}\lun \cP \run\\
\cL_{(0,0,2)}(\cP)&\in\fP_{(0,0,2)} &
\lln \cL_{(0,0,2)}(\cP)\rln&\le 3\cc_\loc \sfrac{\wf'_\sp}{\wf}\lun \cP \run 
\end{align*}

\item
 for 
$
\cP\in\fP_{(3,0,1)}
$, we have
$
\ell(\cP)=0
$,
$\cL_4(\cP)=0\,$,
$
\cL_\fD(\cP)=0
$ 
and
\begin{equation*}
\lln \cI(\cP)\rln\le18 \cc_\loc\big( \sfrac{\wf'_0}{\wf}+ \sfrac{\wf'_\sp}{\wf}\big)\lun \cP \run
\end{equation*}
\end{itemize}
\end{proposition}

\begin{proof} Just apply the previous lemma. 
\end{proof}

\begin{definition}\label{defLOCrel}
Let $\cF(\tilde\psi_*,\tilde\psi)$ be an analytic function of the fields
in a neighbourhood of the origin in $\tilde\cH_0^{(n)}\times \tilde\cH_0^{(n)}$
that obeys $\cF(0,0)=0$. Write $\cF=\cF_\rel+\cF_\irr$ with $\cF_\rel\in\fP_\rel$
and $\cF_\irr\in\fP_\irr$. Define
\begin{equation*}
\ell(\cF)=\ell(\cF_\rel)\quad 
\cL_4(\cF)=\cL_4(\cF_\rel)\quad
\cL_\fD(\cF)=\cL_\fD(\cF_\rel)\quad
\cI(\cF) = \cI(\cF_\rel)+\cF_\irr
\end{equation*}
\end{definition}

\begin{corollary}\label{corLOCproj}
Let $\cF(\tilde\psi_*,\tilde\psi)$ be an analytic function of the fields
in a neighbourhood of the origin in $\tilde\cH_0^{(n)}\times \tilde\cH_0^{(n)}$
that obeys $\cF(0,0)=0$. 
\begin{enumerate}[label=(\alph*), leftmargin=*]
\item 
Then
\begin{align*}
\cF\big((\psi_*,\{\partial_\nu\psi_*\})\,,\,
                        (\psi,\{\partial_\nu\psi\})\big)
&=\ell(\cF)\hskip-2.5pt\int\hskip-3pt dx \ \psi_*(x) \psi(x)
+ \cL_4(\cF)(\psi_*,\psi)\\
&\hskip.5in +\cL_\fD(\cF)\big((\psi_*,\{\partial_\nu\psi_*\})\,,\,
                        (\psi,\{\partial_\nu\psi\})\big)\\
&\hskip.5in  +\cI(\cF)\big((\psi_*,\{\partial_\nu\psi_*\})\,,\,
                        (\psi,\{\partial_\nu\psi\})\big)
\end{align*}

\item 
If the monomial in $\cF$ of type $(2,0,0)$ is 
$\cF_2(\psi_*,\psi)=\int dx\,dx'\ \psi_*(x) K(x,x')\psi(x')$
then
\begin{equation*}
\ell(\cF)=\int dx\ K(0,x)=\frac{\cF_2(1,1)}{\int dx}\qquad\text{and}\qquad
|\ell(\cF)|\le \sfrac{1}{\wf^2}\lln \cF \rln
\end{equation*}

\item
We have
\begin{equation*}
\lun \cL_4(\cF)\run\ ,\ \lln \cL_\fD(\cF)\rln\ ,\ \lln \cI(\cF)\rln\ \le\ 
\Big[1+9\cc_\loc\big( \sfrac{\wf'_0}{\wf}+ \sfrac{\wf'_\sp}{\wf}\big)\Big]^2\,\lun \cF \run
\end{equation*}

\item 
Define a partial ordering\footnote{
``Converting a nonderivative field to a derivative field'' or
``adding a field'', increases $(p_u,p_0,p_\sp)$ under this 
partial ordering.} on the set of vectors $\vp=(p_u,p_0,p_\sp)$ 
by
\begin{equation*}
(p_u,p_0,p_\sp) \lesssim (p_u',p_0',p_\sp')
\,\ \iff\ \  p_0\le p_0',\ \, 
     p_\sp\le p_\sp',\ \ 
     p_u+p_0+p_\sp\le p_u'+p_0'+p_\sp'
\end{equation*}
If $\ \cF\in\fP_\vp\ $ then
$\ \cL_\fD(\cF)\in\bigoplus\limits_{\vp'\gtrsim\vp}\fP_{\vp'}$.
\end{enumerate}
\end{corollary}

\begin{remark}\label{remMXconsteigen}
The following are useful when exploiting Corollary \ref{corLOCproj}.b.
\begin{enumerate}[label=(\alph*), leftmargin=*] 
\item
Denote by $1$,$1_\crs$ and $1_\fin$ the functions on 
$\cX_0^{(n)}$, $\cX_{-1}^{(n+1)}$ and $\cX_n$, respectively,  
which alway take the value 1. Then
\begin{equation*}
Q1 =1_\crs\quad
Q^* 1_\crs=1\quad
Q_n1_\fin =1\quad
Q_n^* 1=1_\fin\quad
\fQ_n 1=a_n1\quad
D_n1_\fin=0
\end{equation*}
where, as in \eqref{eqnOSRan},
\begin{equation*}
a_n=a\big(1 +\smsum_{j=1}^{n-1}\sfrac{1}{L^{2j}}\big)^{-1}
\end{equation*}
\item 
We have
\begin{align*}
& S_n 1_\fin = S_n^* 1_\fin = \sfrac{1}{a_n} 1_\fin &\quad
& S_n(\mu) 1_\fin = S_n(\mu)^* 1_\fin = \sfrac{1}{a_n-\mu} 1_\fin \\
& S_n(\mu)^{(*)} Q_n^* \fQ_n 1 
           = \sfrac{a_n}{a_n-\mu}1_\fin &
&  {(S_n(\mu)^{(*)} Q_n^*\fQ_n)}^* 1_\fin  
          =\sfrac{a_n}{a_n-\mu}1 \\
& B_{(*)n,\mu} 1 = \sfrac{a_n}{[a_n-\mu-\de\mu][a_n-\mu]}1_\fin &
&  B_{(*)n,\mu}^* 1_\fin  = \sfrac{a_n}{[a_n-\mu-\de\mu][a_n-\mu]}1_\fin \\
& B_{(*)n,\mu,D} 1 = 0 &
&  B_{(*)n,\mu,D}^* 1_\fin = 0 \\
& B^{(-)}_{n,\mu,D} 1 = 0 &
&  B^{(-)*}_{n,\mu,D} 1_\fin = 0 
%&\big[\bbbone-Q_nS_n Q_n^* \fQ_n\big]1
%=\big[\bbbone-Q_nS_n^*Q_n^* \fQ_n\big]^*1=0\hidewidth
\end{align*}
where $B_{(*)n,\mu}$ and $B_{(*)n,\mu,D}$ are the operators of
\cite[Proposition \propBGEdephidemu]{BGE} and $B_{n,\mu,D}^{(-)}$
is the operator of \cite[Proposition \propBGEphivepssoln]{BGE}.
\end{enumerate}
\end{remark}

\begin{proof} (a) 
Taking Fourier transforms, both of the equations
\begin{equation*}
Q_n1_\fin =1\quad \text{and}\quad Q_n^* 1=1_\fin
\end{equation*} 
follow from the facts that
the function $u_n(p)$ of \cite[Remark \remPBSqnft.b]{POA}
obeys $u_n(k+\ell)=1$, when $k=\ell=0$ and
$u_n(k+\ell)=0$ when $k=0$ and $0\ne \ell\in\hat\cB_n$.
See \cite[Remark \remPBSqnft.e and Lemma \lemPBSunppties.b,c]{POA}.
Similarly, both of the equations
$Q1 =1_\crs$ and $Q^* 1_\crs=1$ follow from the facts that
the function $u_+(p)$ of \cite[(\eqnPBSuplusdef)]{POA}
obeys $u_+(\fk+\ell)=1$, when $\fk=\ell=0$ and
$u_+(\fk+\ell)=0$ when $\fk=0$ and $0\ne \ell\in\hat\cB^+$.
See \cite[Remark \remPBSqnft.e and Lemma \lemPBSuplusppties.c,d]{POA}.
As $\fQ_n=a\big(\bbbone
             +\smsum_{j=1}^{n-1}\sfrac{1}{L^{2j}}Q_jQ_j^*\big)^{-1}$
the equality $\fQ_n 1=a\big(1
             +\smsum_{j=1}^{n-1}\sfrac{1}{L^{2j}}\big)^{-1}1=a_n1$ follows.
That $D_n 1_\fin=0$
is true is trivial since discrete derivatives annihilate constant functions.

\Item (b) follows from part (a) and the definitions
\begin{align*}
{S_n^{(*)}}^{-1}&=D_n+Q_n^*\fQ_n Q_n\\
{S_n(\mu)^{(*)}}^{-1}&=D_n+Q_n^*\fQ_n Q_n-\mu\\
B_{(*)n,\mu}&=S_n^{(*)}\,
         \big[\bbbone -(\mu+\de\mu)S_n^{(*)}\big]^{-1} 
         S_n(\mu)^{(*)}Q_n^* \fQ_n\\
B_{(*)n,\mu,D}&=S_n(\mu)^{(*)}Q_n^* \fQ_n 
         - \big(Q_n^*\fQ_nQ_n-\mu-\de\mu\big)
          B_{(*)n,\mu}\\
B_{n,\mu,D}^{(-)}1
  &=\big[\bbbone-(Q_n^*\fQ_nQ_n-\mu)S_n(\mu)\big]Q_n^*\fQ_n
\end{align*}
\end{proof}

%&&&&&&&&&&&&&&&&&&&&&&&&&&&&&&&&&
\section{Scaling and Bounds}\label{appSCscaling}
%&&&&&&&&&&&&&&&&&&&&&&&&&&&&&&&&&

Let $n\ge 0$ and $0\le i,j\le n+1$. In this appendix we consider the impact 
of scaling on norms of functions
\begin{equation*}
\cF:\tilde\cH^{(n+1-j)}_{j-1}\times \tilde\cH^{(n+1-j)}_{j-1}
  %\times\cH^{(n+1-i)}_{i-1}\times \cH^{(n+1-i)}_{i-1}
  \rightarrow\bbbc
\end{equation*}
and field maps
\begin{equation*}
A:\tilde\cH^{(n+1-j)}_{j-1}\times \tilde\cH^{(n+1-j)}_{j-1}
  \times\cH^{(n+1-i)}_{i-1}\times \cH^{(n+1-i)}_{i-1}
  \rightarrow \cH^{(0)}_{n}
\end{equation*}
Recall from  \cite[Definition \defSCscaling.b]{PAR1}, that
$(\bbbs \cF)(\tilde \be_*,\tilde \be)
   = \cF\big(\bbbs^{-1}\tilde\be_*,\bbbs^{-1}\tilde\be)$ maps
\begin{equation*}
\bbbs \cF:\tilde\cH^{(n+1-j)}_j\times \tilde\cH^{(n+1-j)}_j
  %\times\cH^{(n+1-i)}_i\times \cH^{(n+1-i)}_i
  \rightarrow\bbbc
\end{equation*}
Similarly, define the scaled field map 
\begin{equation}\label{eqnSAscaledfm}
A^{(s)}(\tilde \be_*,\tilde \be,z_*,z)
   = \bbbl_*^{-1}\big[A\big(\bbbs^{-1}\tilde\be_*,\bbbs^{-1}\tilde\be,
           \bbbs^{-1}z_*,\bbbs^{-1}z)\big]
\end{equation}
with the $\bbbl_*$ of \cite[Definition \defHTbackgrounddomaction.a]{PAR1}.
It maps
\begin{equation*}
A^{(s)}:\tilde\cH^{(n+1-j)}_j\times \tilde\cH^{(n+1-j)}_j
  \times\cH^{(n+1-i)}_i\times \cH^{(n+1-i)}_i
  \rightarrow \cH^{(0)}_{n+1}
\end{equation*}

We fix any $\check\fm,\check\wf,\check\wf',\check\wf_\fl>0$ 
and use the norms $\|\cF(\tilde\al_*,\tilde\al)\|^{\check{}}$,
$\tn A(\tilde\al_*,\tilde\al,\ze_*,\ze)\tn^{\check{}}$ with mass $\check\fm>0$ and 
weight factors $\check\wf,\check\wf',\check\wf_\fl$ 
to measure the unscaled functions and field maps.
See  \cite[Definition \defDEFkrnel]{PAR1}.
The weight factor $\check\wf$ is used for the $\al_{(*)}$'s, 
the weight factor $\check\wf'$ is used for the $\al_{(*)\nu}$'s, 
$0\le\nu\le 3$, and  
the weight factor $\check\wf_\fl$ is used for the $\ze_{(*)}$'s.

Also, fix any $\fm$, $\wf$, $\wf'$, $\wf_\fl>0$ 
and use the norms $\|(\bbbs \cF)(\tilde \be_*,\tilde \be)\|$,
$\tn A^{(s)}(\tilde \be_*,\tilde \be,z_*,z)\tn$ with
mass $\fm>0$ and weight factors $\wf,\wf',\wf_\fl$ to measure 
the scaled functions and field maps. 
The weight factor $\wf$ is used for the $\be_{(*)}$'s, 
the weight factor $\wf'$ is used for the $\be_{(*)\nu}$'s, 
$0\le\nu\le 3$, and  
the weight factor $\wf_\fl$ is used for the $z_{(*)}$'s.

\begin{definition}[Scaling Divergence Factor]\label{defSAdd}
\ \begin{enumerate}[label=(\alph*), leftmargin=*]
\item
Let 
\begin{equation*}
\cM\big(\,(\al_*,\{\al_{*\nu}\})\,,\,
                            (\al,\{\al_\nu\})\,\big)
=\int_{\cX_{j-1}^{(n+1-j)}} dv_1\cdots dv_p\ M(v_1,\cdots,v_p)
             \smprod_{\ell=1}^p\al_{\si_\ell}(v_\ell)
\end{equation*}
be a monomial of degree $p$. Here each $\al_{\si_\ell}$ is one of 
$\al_*,\al, \big\{\al_{*\nu},\al_\nu\big\}_{\nu=0}^3$.
Denote by  
\begin{itemize}[leftmargin=*, topsep=2pt, itemsep=0pt, parsep=0pt]
\item
$p_u$, the number of $\al_{\si_\ell}$'s that is either 
$\al_*$ or $\al$ and 
\item
$p_0$, the number of $\al_{\si_\ell}$'s that is either 
$\al_{*0}$ or $\al_0$ and 
\item 
$p_\sp$, the number of $\al_{\si_\ell}$'s that is one of 
$\big\{\al_{*\nu},\al_\nu\big\}_{\nu=1}^3$.
\end{itemize}
Set
\begin{equation*}
\Sdf(\cM)=\big(\sfrac{1}{L^{3/2}}\sfrac{\wf}{\check\wf}\big)^{p_u}
         \big(\sfrac{1}{L^{7/2}}\sfrac{\wf'}{\check\wf'}\big)^{p_0}
         \big(\sfrac{1}{L^{5/2}}\sfrac{\wf'}{\check\wf'}\big)^{p_\sp}
\end{equation*}

\item
Let $\cF$ be an analytic function on a neighbourhood of the origin
in $\tilde\cH^{(n+1-j)}_{j-1}\times \tilde\cH^{(n+1-j)}_{j-1}$. 
Then $\Sdf(\cF)$ is the supremum of $\Sdf(\cM)$ with $\cM$ running over the
nonzero monomials in the power series representation of $\cF$.
\end{enumerate}
\end{definition}

\begin{lemma}\label{lemSAscaletoscale}
Assume that
$
\fm\le L\check\fm
$.
\begin{enumerate}[label=(\alph*), leftmargin=*]
\item
Let 
\begin{equation*}
\cM\big(\,(\al_*,\{\al_{*\nu}\})\,,\,
                            (\al,\{\al_\nu\})\,\big)
=\int_{\cX_{j-1}^{(n+1-j)}} dv_1\cdots dv_p\ M(v_1,\cdots,v_p)
             \smprod_{\ell=1}^p\al_{\si_\ell}(v_\ell)
\end{equation*}
be a monomial as in Definition \ref{defSAdd}.a. Then the kernel of $\bbbs\cM$ is
\begin{equation}\label{eqnSAkern}
 M^{(s)}(u_1,\cdots,u_p)
=L^{\frac{7}{2}p_u +\frac{3}{2}p_0+\frac{5}{2}p_\sp}  
   M(\bbbl u_1,\cdots,\bbbl u_p)            
\end{equation}
and
\begin{equation}\label{eqnSAkernNorm}
\big\| M^{(s)}\big\|_\fm
\le L^5 L^{-\frac{3}{2}p_u-\frac{7}{2}p_0-\frac{5}{2}p_\sp} 
                  \big\| M\big\|_{\check\fm}
\end{equation}

\item
Let $\cF$ be an analytic function on a neighbourhood of the origin
in $\tilde\cH^{(n+1-j)}_{j-1}\times \tilde\cH^{(n+1-j)}_{j-1}$. Then
\begin{equation*}
\big\|\bbbs\cF\big\|\le L^5\,\Sdf(\cF)\,\|\cF\|^{\check{}}
\end{equation*}
In the event that 
$\wf\le L^{3/2}\check\wf$ and 
$\wf'\le L^{5/2}\check\wf'$, then
$
\big\|\bbbs\cF\big\|\le L^5\,\|\cF\|^{\check{}}
$.

\item
Assume that
$
\wf\le L^{3/2}\check\wf,\
\wf'\le L^{5/2}\check\wf'$ and
$\wf_\fl\le L^{3/2}\check\wf_\fl
$.
Let $A$ be a field map defined on a neighbourhood of the origin
in $\tilde\cH^{(n+1-j)}_{j-1}\times \tilde\cH^{(n+1-j)}_{j-1}
  \times\cH^{(n+1-i)}_{i-1}\times \cH^{(n+1-i)}_{i-1}$ and taking values
in $\cH^{(0)}_{n}$. Then
$
\tn A^{(s)}\tn\le\tn A\tn^{\check{}}
$.
\end{enumerate}
\end{lemma}
\begin{proof} (a)
\cite[Remark \remSCscaling.h]{PAR1} gives \eqref{eqnSAkern}.
Then, introducing the local shorthand notation $\cX=\cX_j^{(n+1-j)}$ and
$\check\cX=\cX_{j-1}^{(n+1-j)}$,
\begin{align*}
&\big\| M^{(s)}\big\|_\fm\\
%&=\max_{1\le i\le p} \max_{u_{i}}
% \int du_1 \cdots du_{i-1}\,du_{i+1} \cdots du_p\ M^{(s)}(u_1,\cdots,u_p)\, e^{\fm\tau(u_1,\cdots,u_p)}\\
&=\max_{1\le i\le p} \max_{u_i}\!
 \int_{\cX^{p-1}}\hskip-12pt du_1 \cdots du_{i-1}\,du_{i+1} \cdots du_p\ 
   L^{\frac{7}{2}p_u +\frac{3}{2}p_0+\frac{5}{2}p_\sp}  M(\bbbl u_1,\cdots,\bbbl u_p)\, e^{\fm\tau(u_1,\cdots,u_p)}\\
&\le\sfrac{L^{\frac{7}{2}p_u +\frac{3}{2}p_0+\frac{5}{2}p_\sp}}{L^{5(p-1)}}\max_{1\le i\le p} \max_{v_i}
 \!\!\int\!\! dv_1 \cdots dv_{i-1}\,dv_{i+1} \cdots dv_p\,
     M(v_1,\cdots,v_p)\, 
   e^{L\check\fm\tau(\bbbl^{-1}v_1,\cdots,\bbbl^{-1}v_p)}\\
&\le\sfrac{L^{\frac{7}{2}p_u +\frac{3}{2}p_0+\frac{5}{2}p_\sp}}{L^{5(p-1)}}\max_{1\le i\le p} \max_{v_i}
 \int dv_1 \cdots dv_{i-1}\,dv_{i+1} \cdots dv_p\ 
     M(v_1,\cdots,v_p)\, 
   e^{\check\fm\tau(v_1,\cdots,v_p)}\\
%&=\sfrac{L^{\frac{7}{2}p_u +\frac{3}{2}p_0+\frac{5}{2}p_\sp}}{L^{5(p-1)}}\big\| M\big\|_{\check\fm}\\
&=L^5 L^{-\frac{3}{2}p_u-\frac{7}{2}p_0-\frac{5}{2}p_\sp} 
                  \big\| M\big\|_{\check\fm}
\end{align*}
since, if $t(v_1,\cdots,v_p)$ is the length of a specific tree $T$ that is
minimal for $\tau(v_1,\cdots,v_p)$ and if $t_L(v_1,\cdots,v_p)$ is the
length of the tree constructed from $T$ by moving the location $v$ of each
vertex of $T$ to $\bbbl^{-1}v$,
\begin{equation*}
L\tau(\bbbl^{-1}v_1,\cdots,\bbbl^{-1}v_p)
\le L t_L(v_1,\cdots,v_p)
\le t(v_1,\cdots,v_p)
=\tau(v_1,\cdots,v_p)
\end{equation*}

\Item (b) 
It suffices to consider the case that $\cF$ is a monomial as in part (a). Then
\begin{align*}
\|\bbbs \cF\| 
& = \big\| M^{(s)}\big\|_\fm\ \wf^{p_u}\,{\wf'}^{\,p_0+p_\sp}\\
& \le L^5 L^{-\frac{3}{2}p_u-\frac{7}{2}p_0-\frac{5}{2}p_\sp} 
      \big(\sfrac{\wf}{\check\wf}\big)^{p_u}\,
      \big(\sfrac{\wf'}{\check\wf'}\big)^{p_0+p_\sp}
\big\| M\big\|_{\check\fm}
  \check\wf^{p_u}\,\check{\wf'}^{p_0+p_\sp}\\
&=L^5\,\Sdf(\cF)\|\cF\|^{\check{}} 
\end{align*}

\Item (c) 
Once again it suffices to consider monomials
\begin{equation*}
A\big(\,(\al_*,\{\al_{*\nu}\})\,,\,(\al,\{\al_\nu\})\,,\,\ze_*\,,\,\ze\big)(v_0)
=\int dv_1\cdots dv_p\ M(v_0,v_1,\cdots,v_p)
             \smprod_{\ell=1}^p\al_{\si_\ell}(v_\ell)
\end{equation*}
of degree $p$. Here each $\al_{\si_\ell}$ is one of 
$\al_*,\al,\big\{\al_{*\nu},\al_\nu\big\}_{\nu=0}^3,\ze_*,\ze$.
If $\al_{\si_\ell}$ is one of $\al_*,\al,\big\{\al_{*\nu},\al_\nu\big\}_{\nu=0}^3$,
then $v_\ell$ runs over $\cX_{j-1}^{(n+1-j)}$.
If $\al_{\si_\ell}$ is one of $\ze_*,\ze$,
then $v_\ell$ runs over $\cX_{i-1}^{(n+1-i)}$. The argument $v_0$ runs
over $\cX_n$.
We denote by  
\begin{itemize}[leftmargin=*, topsep=2pt, itemsep=2pt, parsep=0pt]
\item
 $p_u$, the number of $\al_{\si_\ell}$'s that is one of 
$\al_*,\al,\ze_*,\ze$.
\item
 $p_0$, the number of $\al_{\si_\ell}$'s that is either 
$\al_{*0}$ or $\al_0$ and 
\item
 $p_\sp$, the number of $\al_{\si_\ell}$'s that is one of 
$\big\{\al_{*\nu},\al_\nu\big\}_{\nu=1}^3$.
\end{itemize}
The analog of \eqref{eqnSAkern} for $A$ is
\begin{equation*}
 M^{(s)}(u_0,u_1,\cdots,u_p)
=L^{\frac{7}{2}p_u +\frac{3}{2}p_0+\frac{5}{2}p_\sp}  M(\bbbl u_0,\bbbl u_1,\cdots,\bbbl u_p)            
\end{equation*}
The analog of \eqref{eqnSAkernNorm} for $A$ is
\begin{equation*}
\big\| M^{(s)}\big\|_\fm
\le L^{-\frac{3}{2}p_u-\frac{7}{2}p_0-\frac{5}{2}p_\sp}  
             \big\| M\big\|_{\check\fm}
\end{equation*}
and the claim follows.
\end{proof}

\newpage
%&&&&&&&&&&&&&&&&&&&&&&&&&&&&&&&&&
\section{Notation}\label{appNotation}
%&&&&&&&&&&&&&&&&&&&&&&&&&&&&&&&&&

The references in the following tables are to \cite{PAR1} and this paper.

%
%   PAGE 1
%
\begin{center}
\renewcommand{\arraystretch}{1.3}
  \begin{tabular}{|c|c|c| }
    \hline
     Notation 
     &Definition
     &Comments \\ \hline
  %%%%%
  $X$
  &\S\sectINTmodel
  &spatial lattice \\ \hline
  %%%%%
  $\bh$
  &\S\sectINTmodel
  &``kinetic energy'' operator\\ \hline
  %%%%%
  $\oh$
  &\S\sectINTmodel
  &periodization of $\bh$\\ \hline
  %%%%%
  $\oh_0 = \th\oh$
  &after (\eqnHTaIn)
  & periodization of $\bh_0$ \\ \hline
  %%%%%
  $\bh_0 = \nabla^*\bH\nabla$
  &\S\sectINTstartPoint
  &  \\ \hline
  %%%%%
  $L_\sp$
  &\S\sectINTmodel
  &spatial cutoff\\ \hline
  %%%%%
  $L_\tp=\sfrac{1}{\th k T}$
  &after (\eqnTHsmallfieldoutput)
  &temporal cutoff\\ \hline
  %%%%%
  $\cX_0 = \big( \bbbz / L_\tp\bbbz\big) \times 
                    \big(\bbbz^3 / L_\sp\bbbz^3 \big)$
  &after (\eqnTHsmallfieldoutput)
  & unit lattice \\ \hline
  %%%%%
  $\cX_n = \big(\sfrac{1}{L^{2n}}\bbbz \big/ \sfrac{L_\tp}{L^{2n}}\bbbz\big)
   \!\times\! \big(\sfrac{1}{L^n}\bbbz^3\! \big/\sfrac{L_\sp}{L^n}\bbbz^3 \big)$
  &  Defn. \defHTbackgrounddomaction.a
  &``fine'' scaled lattice\\ \hline
  %%%%%
  $\cX_0^{(n)}\! 
       =\!\big(\bbbz \big/\!\sfrac{L_\tp}{L^{2n}}\bbbz\big)
        \!\times\!\big(\bbbz^3\!\big/\sfrac{L_\sp}{L^n}\bbbz^3 \big)$
  & $\!$before Defn. \defHTblockspintr$\!$
  & unit blocked lattice\\ \hline
  %%%%%
  $\cX_{-1}^{(n+1)}\! 
    =\!\big(L^2\bbbz \big/\!\sfrac{L_\tp}{L^{2n}}\bbbz\big)
  \!\times\!\big(L\bbbz^3\!\big/\sfrac{L_\sp}{L^n}\bbbz^3 \big)\!$
  & Defn. \defHTblockspintr.a
  &``coarse'' blocked lattice \\ \hline
  %%%%%
  $\cX_{j}^{(n)} $
  &Defn. \defHTbackgrounddomaction.a
  &blocked, scaled lattices\\ \hline
  %%%%%
  $ L $
  & Theorem \thmTHmaintheorem
  & scaling parameter \\ \hline
  %%%%%
  $\cH_n=L^2(\cX_n) $
  &Defn. \defHTbackgrounddomaction.a
  & \\ \hline
  %%%%%
  $\cH_0^{(n)}=L^2(\cX_0^{(n)}) $
  &Defn. \defHTbackgrounddomaction.a
  & \\ \hline
  %%%%%
  $\cH_j^{(n)}=L^2(\cX_j^{(n)}) $
  &Defn. \defHTbackgrounddomaction.a
  & \\ \hline
  %%%%%
  $\< \al_1,\al_2\>_j= \int_{X^{(n)}_j} \al_1(u) \,\al_2(u) \ du$
  &  Defn. \defHTbackgrounddomaction.a
  &  bilinear form for $\cH_j^{(n)}$\\ \hline
  %%%%%
  $\int_{X^{(n)}_j} du = \sfrac{1}{L^{5j}} \smsum_{u\in\cX_j^{(n)} }$
  & Defn. \defHTbackgrounddomaction.a
  & ``integral'' over $\cX_j^{(n)}$ \\ \hline
  %%%%%
  $\bbbl : \cX_j^{(n)}\rightarrow \cX_{j-1}^{(n)}$
  & Defn. \defHTbackgrounddomaction.a
  &$(u_0,\bu) \mapsto (L^2u_0, L\bu)$\\ \hline
  %%%%%
  $\bbbl_*: \cH_j^{(n)}\rightarrow \cH_{j-1}^{(n)}$
  & Defn. \defHTbackgrounddomaction.a
  & $\bbbl_*(\al)(\bbbl u) = \al(u)$ \\ \hline
  %%%%%
  $\bbbs= L^{3/2}\,\bbbl_*^{-1}:\cH_{j-1}^{(k)}\rightarrow\cH_j^{(k)}$
  & Defn. \defSCscaling.a
  & field scaling operator \\ \hline
  %%%%%
  $\bbbs_\nu$
  &Defn. \defSCscaling.a
  &scales differentiated fields \\ \hline
  %%%%%
  $ Q:\cH_0^{(n)}\rightarrow\cH_{-1}^{(n+1)} $
  &Defn. \defHTblockspintr.a 
  & blockspin average \\ \hline
  %%%%%
  $ Q_n:\cH_n^{(0)} \rightarrow \cH_0^{(n)}  $
  & Defn. \defHTbackgrounddomaction.a 
  & blockspin average \\ \hline
  %%%%%
  $ \check Q_n $
  & Lemma \lemSCacheckOne
  & $\check Q_n = \bbbs^{-1} Q_n\bbbs=QQ_{n-1}$\\ \hline
  \end{tabular}
\renewcommand{\arraystretch}{1.0}
\end{center}

\newpage
%
%  PAGE 2
%
\begin{center}
\renewcommand{\arraystretch}{1.3}
  \begin{tabular}{|c|c|c| }
    \hline
     Notation 
     &Definition
     &Comments \\ \hline
  %%%%%
  $q$&
  Definition \defHTbasicnorm.d&
  block spin averaging profile\\ \hline
  %%%%%
  $\cA_0(\psi_*,\psi)$&
  (\eqnHTaZero)&
  initial action\\ \hline
  %%%%%
  $A_0(\psi_*,\psi,\mu,\cV)$&
  Definition \defHTbackgrounddomaction.b&
  dominant part of $\cA_0$    \\ \hline
  %%%%%
  $\cA_n(\psi_*,\psi)$&
  Proposition \propSTmainProp.a& 
  scale $n$ action \\ \hline
  %%%%%
  $A_n(\psi_*,\psi,\phi_*,\phi,\mu,\cV)$&
  Definition \defHTbackgrounddomaction.b&
  dominant part of $\cA_n$   \\ \hline
  %%%%%
  $D_0=\bbbone - e^{-\oh_0} -e^{-\oh_0} \partial_0$&
  (\eqnHTaIn)&
   \\ \hline
  %%%%%
  $\bD_0$&
  \S\sectINTstartPoint&
  $D_0$ is the periodization of $\bD_0$\!\\ \hline
  %%%%%
  $D_n = L^{2n}\ \bbbl_*^{-n} \, D_0\, \bbbl_*^n$&
  Definition \defHTbackgrounddomaction.a &
  scaled $D_0$\\ \hline
  %%%%%
  $\bv$&
  \S\sectINTmodel&
  \!original two--body interaction\!\\ \hline
  %%%%%
  $\cV_0(\psi_*,\psi)$&
  (\eqnHTaIn), \cite[Prop. \propSZprepforblockspin]{PAR1}&
  scale zero interaction\\ \hline
  %%%%%
  $V_0$&
  \S\sectINTstartPoint&
  kernel of $\cV_0$\\ \hline
  %%%%%
  $\bV_0$&
  \S\sectINTstartPoint&
  $V_0$ is the periodization of $\bV_0$\\ \hline
  %%%%%
  $\rv_0=\smsum\limits_{x_2,x_3,x_4} \bV_0(0,x_2,x_3,x_4)$&
  \S\sectINTstartPoint&
    \\ \hline
  %%%%%
  $\fv_0= 2\|\bV_0\|_{2m}$&
  \S\sectINTstartPoint&
    \\ \hline
  %%%%%
  $\fv_n$&
  after (\eqnOSFabbrevwt)&
  $\sfrac{\fv_0}{L^n} = 2\|\cV_n^{(u)}\|_{2m}$ \\ \hline
  %%%%%
  $\cV_n^{(u)}$&
  Definition \defHTbackgrounddomaction.b&
  $n$--fold scaled $\cV_0$\\ \hline
  %%%%%
  $V_n^{(u)}$&
  Definition \defHTbackgrounddomaction.b&
  kernel of $\cV_n^{(u)}$\\ \hline
  %%%%%
  $ \cV_n(\phi_*,\phi)$ &
  Theorem \thmTHmaintheorem&
  scale $n$ interaction\\ \hline
  %%%%%
  $\cR_0(\psi_*,\psi)$&
  (\eqnHTaIn), \cite[Prop. \propSZprepforblockspin]{PAR1}&
     \\ \hline
  %%%%%
  $\cE_0(\psi_*,\psi)$&
  (\eqnHTaIn), \cite[Prop. \propSZprepforblockspin]{PAR1}&
     \\ \hline
  %%%%%
  $\mu$&
  \S\sectINTmodel&
  original chemical potential\\ \hline
  %%%%%
  $\mu_0$&
  (\eqnHTaIn), \cite[Prop. \propSZprepforblockspin]{PAR1}&
  scale zero chemical potential\\ \hline
  %%%%%
  $\mu_*$&
  (\eqnINTmustardef)&
  $\mu_*+\fv_0^{\sfrac{4}{3}-16\eps}\le 
                    \mu_0 \le \fv_0^{\sfrac{8}{9} +\eps}$\\ \hline
  %%%%%
  $\mu_n$&  
  Theorem \thmTHmaintheorem&
  scale $n$ chemical potential\\ \hline
  %%%%%
  $\bbbt$&
  Definition \defHTblockspintr.b   &
  block spin transformation\\ \hline
  %%%%%
  $a=1$&
  Definition \defHTblockspintr.b &
  block spin parameter\\ \hline
  %%%%%
  $a_n$&
  (\eqnOSRan) &
   $a\big(1 +\smsum_{j=1}^{n-1}\sfrac{1}{L^{2j}}\big)^{-1}$\\ \hline
  %%%%%
  \end{tabular}
\renewcommand{\arraystretch}{1.0}
\end{center}

\newpage
%
%  PAGE 3
%
\begin{center}
\renewcommand{\arraystretch}{1.3}
  \begin{tabular}{|c|c|c| }
    \hline
     Notation 
     &Definition
     &Comments \\ \hline
  %%%%%
  $\fQ_n$&
  Definition \defHTbackgrounddomaction.b &
  $\!a\big(1 +\smsum_{j=1}^{n-1}\sfrac{1}{L^{2j}}Q_jQ_j^*\big)^{-1}\!$
      if $n\ge 2$\\ \hline
  %%%%%
  $\check \fQ_n$&
  Lemma \lemSCacheckOne&
  $\check \fQ_n=\sfrac{1}{L^2}\bbbs^{-1} \fQ_n\bbbs$\\ \hline
  %%%%%
  $N^{(n)}_\bbbt$&
  Definition \defHTblockspintr.b&
  normalization constant for $\bbbt$\\ \hline
  %%%%%
  $\bbbt_n^{(SF)}$&
  Definition \defHTapproximateblockspintr&
  small field blockspin transformation \\ \hline
  %%%%%
  $ \tilde N^{(n)}_\bbbt$&
  Definition \defHTapproximateblockspintr&
  normalization constant for $\bbbt_n^{(SF)}$\\ \hline
  %%%%%
  $\phi_{(*)n}(\psi_*,\psi,\mu,\cV)$&
  Proposition \propHTexistencebackgroundfields&
  background fields\\ \hline
  %%%%%
  $\phi_{(*)n}^{(\ge 3)}(\psi_*,\psi,\mu,\cV)$&
  Proposition \propHTexistencebackgroundfields&
  part of $\phi_{(*)n}$ of degree at least 3\\ \hline
  %%%%%
  $\psi_{(*)n}(\th_*,\th,\mu_n,\cV_n)$&
  Proposition \propHTexistencecriticalfields&
  critical fields\\ \hline
  %%%%%
  $\psi_{(*)n}^{(\ge 3)}(\th_*,\th,\mu_n,\cV_n)$&
  Proposition \propHTexistencecriticalfields&
  part of $\psi_{(*)n}$ of degree at least 3\\ \hline
  %%%%%
  $\De^{(n)}$&
  (\eqnHTden)&
    \\ \hline
  %%%%%
  $ C^{(n)}$&
  (\eqnHTcn)&
  covariance\\ \hline
  %%%%%
  $D^{(n)}$&
  before (\eqnHTcn)&
  square root of $C^{(n)}$ \\ \hline
  %%%%%
  $C^{(n)}(\mu)$&
  Proposition \propHTexistencecriticalfields&
  $C^{(n)}(\mu)=\big(\sfrac{a}{L^2}Q^*Q+\De^{(n)}(\mu)\big)^{-1}$\\ \hline
  %%%%%
  $\De^{(n)}(\mu)$&
  Proposition \propHTexistencecriticalfields&
    \\ \hline
  %%%%%
  $\de\psi_{(*)}$&
  (\eqnHTfirstchangeintvar)&
  fluctuation fields\\ \hline
  %%%%%
  $\de\psi_{(*)}=D^{(n)(*)} \ze_{(*)}$&
  after (\eqnHTcn)&
  fluctuation fields\\ \hline
  %%%%%
  $z(w) = \ze(\bbbl w)$&
  before (\eqnSTintuncheck)&
  fluctuation field\\ \hline
  %%%%%
  $\tilde\al = \big(\al,\!\{\al_\nu \}_{\nu=0,1,2,3}\big)\!$&
  (\eqnTHdefexpandedstates)&
  $\al,\al_\nu\in \cH_j^{(n)}$\\ \hline
  %%%%%
  $\tilde \cH_j^{(n)}$&
  (\eqnTHdefexpandedstates)&
   $\big\{\tilde\al\big\}=\tilde\cH^{(n)^{\oplus 4}}_j$\\ \hline
  %%%%%
  $\vp=(p_u,p_0,p_\sp)$&
  Definition \defINTmonomialtype&
  monomial type\\ \hline
  %%%%%
  $\fD$&
  (\eqnINTfDdef)&
  low degree watch list\\ \hline
  %%%%%
  $\fD_\rel$&
  Definition \defINTrelmonomial&
  scaling/weight relevant monomial types \\ \hline
  %%%%%
  $\sdf(\vp;C)$&
   Definition \defOSFsdf&
   Scaling divergence factor\\ \hline
  %%%%%
  $\sdf(C)$&
   Definition \defOSFsdf&
   $\sup_{\vp\notin\fD_\rel}\sdf(\vp;C)$\\ \hline
  %%%%%
  $\De(\vec p)$&
   Definition \defOSFsdf&
   $\sfrac{3}{2}p_u+\sfrac{7}{2}p_0+\sfrac{5}{2}p_\sp$ where
        $\vp=(p_u,p_0,p_\sp)$\\ \hline
  %%%%%
  $\tilde\cR_0(\tilde\psi_*,\tilde\psi)$&
  \S\sectINTstartPoint&
  $\!\cR_0(\psi_*,\psi) \!=\! \tilde\cR_0\big((\psi_*,\{\partial_\nu\psi_*\}),(\psi,\{\partial_\nu\psi\})
          \big)\!\!$\\ \hline
  %%%%%
  $\tilde\cR_0^{(\vp)}$&
  \S\sectINTstartPoint&
  part of $\tilde\cR_0$ of type $\vp$\\ \hline
  %%%%%
  $\tilde\cR_n^{(\vp)}(\tilde\phi_*,\tilde\phi)$&
  Theorem \thmTHmaintheorem&
  polynomial of type $\vp$\\ \hline
  %%%%%
  \end{tabular}
\renewcommand{\arraystretch}{1.0}
\end{center}

\newpage
%
%  PAGE 4
%
\begin{center}
\renewcommand{\arraystretch}{1.3}
  \begin{tabular}{|c|c|c| }
    \hline
     Notation 
     &Definition
     &Comments \\ \hline
  %%%%%
  $ \tilde\cR_n(\tilde\phi_*,\tilde\phi)$&
  Thm. \thmTHmaintheorem &
  $\,\tilde\cR_n(\tilde\phi_*,\tilde\phi)=\sum_{\vp\in\fD}
               \tilde\cR_n^{(\vp)}(\tilde\phi_*,\tilde\phi)\,$\\ \hline
  %%%%%
  $\cR_n(\phi_*,\phi )$&
  Thm. \thmTHmaintheorem&
    $\cR_n(\phi_*,\phi )
= \tilde\cR_n\big((\phi_*,\{\partial_\nu\phi_*\}),(\phi,\{\partial_\nu\phi\})
          \big)$\\ \hline
  %%%%%
  $ \tilde\cE_n(\tilde\psi_*,\tilde\psi)$&
  Thm. \thmTHmaintheorem&
  scaling/weight irrelevant function\\ \hline
  %%%%%
  $ \cE_n(\psi_*,\psi )$&
   Thm. \thmTHmaintheorem&
  $\cE_n(\psi_*,\psi )
= \tilde\cE_n\big((\psi_*,\{\partial_\nu\psi_*\}),(\psi,\{\partial_\nu\psi\})
          \big)$ \\ \hline
  %%%%%
  $ \cZ_n$&
  Thm. \thmTHmaintheorem&
  normalization constant\\ \hline
  %%%%%
  $ \tilde\cZ_n$&
  (\eqnHTmultiBS)&
  $\tilde \cZ_n=\smprod_{j=1}^n L^{3|\cX_0^{(j)}|}$\\ \hline
  %%%%%
  $ \cZ'_n$&
   Prop. \propOSFmainprop&
   normalization constant\\ \hline
  %%%%%
  $\|f\|_\fm$&
  Defn. \defHTkernelnorm&
  $\ell^1$--$\ell^\infty$ norm with mass $\fm$ of $f:\cX\rightarrow\bbbc$\\ \hline
  %%%%%
  $   $&
  Defn. \defHTabstractnorm&
  norm with mass $\fm$ and weights $\ka_1,\cdots,\ka_s$\\ \hline
  %%%%%
  $\tn A\tn$&
  \cite[Defn. \defDEFkrnel]{PAR1}&
  $\!$field--map norm of mass $\fm$ and weights $\ka_j\!\!$\\ \hline
  %%%%%
  $\ka(n)=\sfrac{L^{\eta n}}{\fv_0^{1/3-\eps}}$&
  Defn. \defHTbasicnorm.a&
  weight for $\psi_{(*)}$ in the $n^{\rm th}$ step\\ \hline
  %%%%%
  $\eta=\sfrac{1}{2}+\sfrac{1}{3}\sfrac{\log\fv_0}{\log(\mu_0-\mu_*)}$&
  Defn. \defHTbasicnorm.a&
  $\sfrac{3}{4}+2\eps  < \eta < \sfrac{7}{8} -\sfrac{\eps}{3}$\\ \hline
  %%%%%
  $\ka'(n)=\sfrac{L^{\eta' n}}{\fv_0^{1/3-\eps}}$&
  Defn. \defHTbasicnorm.a&
  weight for $\partial_\nu\psi_{(*)}$ in the $n^{\rm th}$ step\\ \hline
  %%%%%
  $\eta'=\sfrac{3}{2}\!-\!\sfrac{\log\fv_0}{\log(\mu_0-\mu_*)}\!-\!\eps\!$&
  Defn. \defHTbasicnorm.a&
  $\sfrac{3}{8}  < \eta' < \sfrac{3}{4} -8\eps$\\ \hline
  %%%%%
  $\fe_\fl(n)=L^{\eta_\fl n}\,\fv_0^{\sfrac{1}{3}-2\eps }$&
  Defn. \defHTbasicnorm.a&
  bound on fluctuation integral of $n^{\rm th}$ step\\ \hline
  %%%%%  
  $\eta_\fl\!=\!\big(\sfrac{2}{3}\!-\!4\eps\big)\sfrac{\log\fv_0}{\log(\mu_0-\mu_*)}\!$&
  Defn. \defHTbasicnorm.a&
    \\ \hline
  %%%%%
  $\bar\ka$&
  (\eqnOSFabbrevwt)&
  $\ka(n+1)$\\ \hline
  %%%%%
  $\bar\ka'$&
  (\eqnOSFabbrevwt) &
  $\ka'(n+1)$\\ \hline
  %%%%%
  $\bar\ka_\fl$&
  (\eqnOSFabbrevwt)&
  $\ka_\fl(n+1)=4r_n$\\ \hline
  %%%%%
  $\bar\ka^\vp$&
  Lemma \lemOSFmainlem.c &
  $\bar\ka^\vp=\bar\ka^{p_u}\bar\ka'^{p_0+p_\sp}$ 
              where $\vp=(p_u,p_0,p_\sp)$\\ \hline
  %%%%%
  $\|\tilde \cE(\tilde\psi_*,\tilde\psi)\|^{(n)}$&
    Defn. \defHTbasicnorm.a&
    norm with mass $m$ and weights $\ka(n)$, $\ka'(n)$\!\\ \hline
  %%%%%
  $\|\tilde \cE(\tilde\psi_*,\tilde\psi)\|_m$&
    Defn. \defHTbasicnorm.a&
    norm with mass $m$ and weights all one\\ \hline
  %%%%%
  $n_p\le\log_L\sfrac{1}{\fv_0^{\sfrac{2}{3}-8\eps}}$&
    Defn. \defHTbasicnorm.b&
    number of steps in the ``parabolic flow'' \\ \hline
  %%%%%
  $r_n = \sfrac{1}{4} \ka_\fl(n+1) $&
    Defn. \defHTbasicnorm.c&
    radius of domain of integration in $n^{\rm th}$ step\!\\ \hline
  %%%%%
  $\ka_\fl(n)= \big( \sfrac{L^n}{\fv_0} \big)^{\eps/2}$&
    Defn. \defHTbasicnorm.c&
    \\ \hline
  %%%%%
  $ \fr_{\vp}(n,C)$&
   Remark \remHTpreciseinduction&
  $\big\|\tilde\cR_n^{(\vp)}\big\|_m \le \fr_{\vp}(n,\CC_\cR)$\\ \hline
  %%%%%
  $ \Pi^n_\ell(C)$&
   Remark \remHTpreciseinduction&
   \\ \hline
  %%%%%
  \end{tabular}
\renewcommand{\arraystretch}{1.0}
\end{center}

\newpage
%
%  PAGE 5
%
\begin{center}
\renewcommand{\arraystretch}{1.3}
  \begin{tabular}{|c|c|c| }
    \hline
     Notation 
     &Definition
     &Comments \\ \hline
  %%%%%
  $ \CC_\cR$&
   Remark \remHTpreciseinduction&
  $n$, $L$, independent constant \\ \hline
  %%%%%
  $ \CC_{\de\cV}$&
   Remark \remHTpreciseinduction&
  $n$, $L$, independent constant \\ \hline
  %%%%%
  $ \CC_\fl$&
  Lemma \lemOSFmainlem&
  $n$, $L$, independent constant\\ \hline
  %%%%%
  $ \CC_\ren$&
  Lemma \lemRENcRcE&
  $n$, $L$, independent constant\\ \hline
  %%%%%
  $\Gam_\op$&
  Convention \convBGEconstants&
  $n$, $L$, independent constant \\ \hline
  %%%%%
  $ \mu_{\rm up}$&
  Convention \convBGEconstants&
  $n$, $L$, independent constant \\ \hline
  %%%%%
  $ \GGa_1, \GGa_2, \cdots$&
  Convention \convBGEconstants&
  $n$, $L$, independent constants \\ \hline
  %%%%%
  $\GGa_\bg$&
  Convention \convBGEconstants&
  $\max_j\GGa_j$ \\ \hline
  %%%%%
  $\rrho_1,\rrho_2,\cdots$&
  Convention \convBGEconstants&
  $n$, $L$, independent constants \\ \hline
  %%%%%
  $\rrho_\bg$&
  Convention \convBGEconstants&
  $\min\big\{\sfrac{1}{8}\,,\,\min_j\rrho_j\big\}$ \\ \hline
  %%%%%
  $\LLa_{\de\mu}$&
  Lemma \lemRENmunppties&
  \!$n$ independent, $L$ dependent constant\!\\ \hline
  %%%%%
  $\cc_\loc$&
  Lemma \lemLlocalize&
  $n$, $L$, independent constant\\ \hline
  %%%%%
  $\cc_A$&
  (\eqnRENavarbndsa)&
  $n$, $L$, independent constant\\ \hline
  %%%%%
  $\cc_\Om$&
  Lemma \lemRENpsitophi&
  $n$, $L$, independent constant\\ \hline
  %%%%%
  $\cc_{\de\cV}$&
  Lemma \lemRENreninteraction&
  $n$, $L$, independent constant\\ \hline
  %%%%%
  $\GGa_\Phi$&
  (\eqnRENPhiDeltaPhiestimates), (\eqnRENnormsvarphi)&
  $n$, $L$, independent constant\\ \hline
  %%%%%
  $\cc_\gar$&
  Lemma \lemRENrngarbage&
  $n$, $L$, independent constant\\ \hline
  %%%%%
  $\cc_{\mu_*}$&
  Lemma \lemOSImustar&
  $n$, $L$, independent constant\\ \hline
  %%%%%
  $S_n=(D_n\!+\!Q_n^*\fQ_n Q_n)^{-1}$&
  Theorem \HTthminvertibleoperators&
  Green's functions\\ \hline
  %%%%%
  $S_n(\mu)$&
  Theorem \HTthminvertibleoperators&
  $S_n(\mu)=(D_n+Q_n^*\fQ_n Q_n-\mu)^{-1}$\\ \hline
  %%%%%
  $e_\r,\ e_\R,\ e_\mu,\ e_\RP $&
  before \cite[(\eqnSZuvoutput)]{PAR1}&
  parameters in \cite[Hypothesis 2.14]{UV}\\ \hline
  %%%%%
  $\cZ_\th$&
  \cite[(\eqnSZspa)]{PAR1}&
  normalization constant\\ \hline
  %%%%%
  $\cZ_\In $&
  \cite[Prop. \propSZprepforblockspin]{PAR1}&
  $\cZ_\In=\cZ_\th e^{-\th\mu}$\\ \hline
  %%%%%
  $j(t)=e^{-t(h-\mu)}$&
  \cite[(\eqnSZspa)]{PAR1}&
    \\ \hline
  %%%%%
  $V_\th(\al^*,\be)$&
  \cite[(\eqnSZspa)]{PAR1}&
  interaction output from \cite{UV}\\ \hline
  %%%%%
  $\cR_\theta(\al_*,\be)$&
  \cite[(\eqnSZspa)]{PAR1}&
  degree two output from \cite{UV} \\ \hline
  %%%%%
  $\cE_\theta(\al_*,\be)$&
  \cite[(\eqnSZspa)]{PAR1}&
  higher degree output from \cite{UV}\\ \hline
  %%%%%
  \end{tabular}
\renewcommand{\arraystretch}{1.0}
\end{center}

\newpage
%
%  PAGE 6
%
\begin{center}
\renewcommand{\arraystretch}{1.3}
  \begin{tabular}{|c|c|c| }
    \hline
     Notation 
     &Definition
     &Comments \\ \hline
  %%%%%
  $\cD_\theta(\al_*,\be)$&
  \cite[(\eqnSZspa)]{PAR1}&
  $\cD_\theta(\al_*,\be) 
         = \cR_\theta(\al_*,\be) +\cE_\theta(\al_*,\be)$\\ \hline
  %%%%%
  $ \check A_n(\th_*,\th,\check\phi_*,\check\phi,\mu,\cV)$&
  Definition \defSCacheck&
  $A_n(\bbbs\th_*,\bbbs\th,\bbbs\check\phi_*,\bbbs\check\phi,L^2\mu,\bbbs\cV)$\\ \hline
  %%%%%
  $ \check\phi_{(*)n}(\th_*,\th,\mu,\cV)$&
  Definition \defBGAphicheck&
  $\bbbs^{-1}\big[\phi_{(*)n}(\bbbs\th_*,\bbbs\th,L^2\mu,\bbbs\cV)\big]$\\ \hline
  %%%%%
  $\de\phi_{(*)n}\big(\psi_{*},\psi,\de\psi_*,\de\psi,\mu,\cV\big)\!$&
  Definition \defBGAbckgndVarn&
     \\ \hline
  %%%%%
  $\de\check\phi_{(*)n}\big(\th_{*},\th,\de\psi_*,\de\psi,\mu,\cV\big)\!$&
  Definition \defBGAbckgndVarn&
     \\ \hline
  %%%%%
  $\de{\check\phi_{(*)n}}^{(+)}\big(\th_*,\th;\de\psi_*,\de\psi,\mu,\cV\big)$&
  Definition \defBGAbckgndVarn&
     \\ \hline
  %%%%%
  $\hat \psi_{(*)n}(\psi_*,\psi,\mu,\cV)$&
  (\eqnSThatpsi)&
     $\bbbs\big[\psi_{*n}(\bbbs^{-1}\psi_*,\bbbs^{-1}\psi,\mu,\cV)\big]$\\ \hline
  %%%%%
  $\de\hat\phi_{(*)n+1}(\psi_*,\psi,z_*,z)$&
  (\eqnSTdehatphidef)&
     \\ \hline
  %%%%%
  $\de\hat\phi_{(*)n+1}^{(+)}(\psi_*,\psi,z_*,z)$&
  (\eqnOSAhatphiplus)&
     \\ \hline
  %%%%%
  $\check\cC_n(\th_*,\th)$&
  beginning \S \chapSTstrategy &
     \\ \hline
  %%%%%
  $\check\cF_n(\th_*,\th) $&
  beginning \S \chapSTstrategy &
     \\ \hline
  %%%%%
  $\check\cE_{n+1,1}(\th_*,\th)$&
  beginning \S \chapSTstrategy &
  $\cE_n\big( \psi_{*n}(\th_*,\th,\mu_n,\cV_n),\psi_n(\th_*,\th,\mu_n,\cV_n)\big)\!$\\ \hline
  %%%%%
  $\de\check\cE_n(\th_*,\th,\de\psi_*,\de\psi)$&
  beginning \S \chapSTstrategy &
     \\ \hline
  %%%%%
  $\de\check\cR_n(\th_*,\th,\de\psi_*,\de\psi)$&
  beginning \S \chapSTstrategy &
     \\ \hline
  %%%%%
  $\de \check A_n(\th_*,\th,\de\psi_*,\de\psi)$&
  beginning \S \chapSTstrategy &
     \\ \hline
  %%%%%
  $\cC_n(\psi_*,\psi)$&
  before (\eqnSThatpsi) &
     \\ \hline
  %%%%%
  $\cF_n(\psi_*,\psi)$&
  (\eqnOSAfluctInt) &
     Also see Proposition \propOSFmainprop\\ \hline
  %%%%%
  $\cE_{n+1,1}(\psi_*,\psi)$&
  (\eqnSThatpsi)&
   $(\bbbs\cE_n)\big( \hat\psi_{(*)n}(\psi_*,\psi,\mu_n,\cV_n)\big)$\\ \hline
  %%%%%
  $ \de\cE_n(\psi_*,\psi,z_*,z)$&
  (\eqnOSAdeEndef)&
   \\ \hline
  %%%%%
  $ \de\cR_n(\psi_*,\psi,z_*,z)$&
   (\eqnOSAdeRndef)&
   \\ \hline
  %%%%%
  $ \de A_n(\psi_*,\psi,z_*,z)$&
  (\eqnOSAdeAndef)&
   \\ \hline
  %%%%%
  $ \de A_n^{(2)},\ \de A_n^{(\ge 3)}$&
  Lemma \lemOSFmainlem.a&
   \\ \hline
  %%%%%
  $ \tilde\cE_{n+1,1}(\tilde\psi_*,\tilde\psi) $&
  Lemma \lemOSFmainlem.b&
   $\!\cE_{n+1,1}(\psi_*,\psi)\!=\!
\tilde\cE_{n+1,1}\big(\!(\psi_{(*)},\{\partial_\nu\psi_{(*)}\})\big)\!$\\ \hline
  %%%%%
  $ \tilde\cE_{n+1,2}(\tilde\psi_*,\tilde\psi) $&
  Lemma \lemRENcRcE&
   $\tilde\cE_{n+1}=\tilde\cE_{n+1,1}+\tilde\cE_{n+1,2}$\\ \hline
  %%%%%
  $ \de\tilde\cE_n(\tilde\psi_*,\tilde\psi,z_*,z) $&
  Lemma \lemOSFmainlem.b&
   \\ \hline
  %%%%%
  $ \de\tilde\cR_n^{(\vp)}(\tilde\psi_*,\tilde\psi,z_*,z)$&
  Lemma \lemOSFmainlem.c&
   \\ \hline
  %%%%%
  $ \de\tilde\cR_n(\psi_*,\psi,z_*,z)$&
  Lemma \lemOSFmainlem.c&
   \\ \hline
  \end{tabular}
\renewcommand{\arraystretch}{1.0}
\end{center}

\newpage
%
%  PAGE 7
%
\begin{center}
\renewcommand{\arraystretch}{1.3}
  \begin{tabular}{|c|c|c| }
    \hline
     Notation 
     &Definition
     &Comments \\ \hline
  %%%%%
  $  \si_n(\vp)$&
  Lemma \lemOSFmainlem.c&
   \\ \hline
  %%%%%
  $ \tilde\cE_\fl(\tilde\psi_*,\tilde\psi)$&
  Proposition \propOSFmainprop&
   \\ \hline
  %%%%%
  $\tilde\cD(\tilde\psi_*,\tilde\psi,z_*,z)$&
  before Lemma \lemOSFcDtwo  &
     \\ \hline
  %%%%%
  $P^\psi_2$&
  before Lemma \lemOSFcDtwo  &
  degree 1 in each of $\psi_*$, $\psi$, any 
  degree in $z_{(*)}$\\ \hline
  %%%%%
  $P^\psi_1$&
  before Lemma \lemOSFcDtwo &
  extracts degree 1 in $\psi_{(*)}$, any degree in $z_{(*)}$\\ \hline
  %%%%%
  $P^\psi_0$&
  before Lemma \lemOSFcDtwo  &
  degree 0 in $\psi_{(*)}$, $\psi_{(*)\nu}$, 
     any degree in $z_{(*)}$\\ \hline
  %%%%%
  $\cM_n$&
  Lemma \lemOSFcDtwo.a &
     \\ \hline
  %%%%%
  $P_{\psi_*\psi}$&
   Proposition \propOSFcEtwo.a  &
   degree 1 in each of $\psi_*$, $\psi$, 
     degree 0 in $\psi_{(*)\nu}$\\ \hline
  %%%%%
  $M'_n $&
  Proposition \propOSFcEtwo.a &
     \\ \hline
  %%%%%
  $ M_n $&
  Proposition \propOSFcEtwo.b &
     \\ \hline
  \end{tabular}
\renewcommand{\arraystretch}{1.0}
\end{center}

\newpage
%%%%%%%%%%%%%%%%%%%%%%%%%%%%%%%%%%%%%%
\bibliographystyle{plain}
\bibliography{refs}

\begin{thebibliography}{1}

\bibitem{CPC}
T.~Balaban, J.~Feldman, H.~Kn{\"o}rrer, and E.~Trubowitz.
\newblock {Power Series Representations for Complex Bosonic Effective Actions.
  I. A Small Field Renormalization Group Step}.
\newblock {\em Journal of Mathematical Physics}, 51:053305, 2010.

\bibitem{UV}
T.~Balaban, J.~Feldman, H.~Kn{\"o}rrer, and E.~Trubowitz.
\newblock {The Temporal Ultraviolet Limit for Complex Bosonic Many-body
  Models}.
\newblock {\em Annales Henri Poincar{\'e}}, 11:151--350, 2010.

\bibitem{ParOv}
T.~Balaban, J.~Feldman, H.~Kn{\"o}rrer, and E.~Trubowitz.
\newblock {Complex Bosonic Many--body Models: Overview of the Small Field
  Parabolic Flow}.
\newblock Preprint, 2016.

\bibitem{POA}
T.~Balaban, J.~Feldman, H.~Kn{\"o}rrer, and E.~Trubowitz.
\newblock {Operators for Parabolic Block Spin Transformations}.
\newblock Preprint, 2016.

\bibitem{SUB}
T.~Balaban, J.~Feldman, H.~Kn{\"o}rrer, and E.~Trubowitz.
\newblock {Power Series Representations for Complex Bosonic Effective Actions.
  III. Substitution and Fixed Point Equations}.
\newblock Preprint, 2016.

\bibitem{BlockSpin}
T.~Balaban, J.~Feldman, H.~Kn{\"o}rrer, and E.~Trubowitz.
\newblock {The Algebra of Block Spin Renormalization Group Transformations}.
\newblock Preprint, 2016.

\bibitem{PAR1}
T.~Balaban, J.~Feldman, H.~Kn{\"o}rrer, and E.~Trubowitz.
\newblock {The Small Field Parabolic Flow for Bosonic Many--body Models: Part 1
  --- Main Results and Algebra}.
\newblock Preprint, 2016.

\bibitem{BGE}
T.~Balaban, J.~Feldman, H.~Kn{\"o}rrer, and E.~Trubowitz.
\newblock {The Small Field Parabolic Flow for Bosonic Many--body Models: Part 4
  --- Background and Critical Field Estimates}.
\newblock Preprint, 2016.

\end{thebibliography}
%%%%%%%%%%%%%%%%%%%%%%%%%%%%%%%%%%%%%

%\printIssueCount
\end{document}